\def\draft{0}
\def\anonymize{0}

\documentclass[11pt,oneside]{article}
\usepackage[margin=1.0in]{geometry}
\usepackage{amsmath}
\usepackage{amsfonts}
\usepackage{amsthm}
\usepackage{enumerate}
\usepackage[T1]{fontenc}
\usepackage{mathrsfs}
\usepackage{mathdots}
\usepackage{graphicx}
\usepackage{mathtools}
\graphicspath{ {./images/} }
\usepackage{xcolor}
\usepackage{hyperref}
\hypersetup{colorlinks=true,linkcolor=red,citecolor=blue,filecolor=magenta,urlcolor=cyan}
\usepackage{subfiles}
\usepackage{pgfplots}
\usepackage{tikz}
\usepackage{algpseudocode}
\usepackage{algorithm}
\usepackage{tikz}
\usepackage{bbm}
\usepackage{multirow}
\usepackage{geometry}
\usepackage{array} 
\pgfplotsset{compat=1.18} 

\usepackage[utf8]{inputenc}
\usepackage[style=alphabetic,sorting=nyt,maxbibnames=9,abbreviate=false,backend=biber]{biblatex}

\usepackage{boldline} 

\newtheorem{theorem}{Theorem}[section]
\newtheorem{corollary}[theorem]{Corollary}
\newtheorem{lemma}[theorem]{Lemma}
\newtheorem{proposition}[theorem]{Proposition}
\newtheorem{claim}[theorem]{Claim}

\theoremstyle{definition}
\newtheorem{remark}[theorem]{Remark}
\newtheorem{definition}[theorem]{Definition}

\newcommand{\rand}{\mathrm{rand}}

\newcommand{\grand}{g^\rand}
\newcommand{\hrand}{h^\rand}
\newcommand{\coins}{\mathrm{coins}}
\newcommand{\post}{\mathrm{post}}
\newcommand{\Hb}{\mathrm{H}}

\numberwithin{equation}{section}

\newcommand{\D}{\mathcal{D}}

\newcommand{\X}{\mathcal{X}}
\newcommand{\U}{\mathcal{U}}
\newcommand{\V}{\mathcal{V}}
\newcommand{\C}{\mathcal{C}}
\newcommand{\Ha}{\mathcal{H}}

\newcommand{\F}{\mathcal{F}}

\newcommand{\poly}{\mathrm{poly}}
\newcommand{\Bern}{\mathrm{Bern}}
\newcommand{\image}{\mathrm{image}}

\newcommand{\zo}{\{0,1\}}

\newcommand{\Pa}{\mathcal{P}}

\newcommand{\Sa}{\mathcal{S}}

\newcommand{\E}{\mathop{\mathbb{E}}}

\newcommand{\remove}[1]{}

\newtheorem{innercustomgeneric}{\customgenericname}
\providecommand{\customgenericname}{}
\newcommand{\newcustomtheorem}[2]{%
  \newenvironment{#1}[1]
  {%
   \renewcommand\customgenericname{#2}%
   \renewcommand\theinnercustomgeneric{##1}%
   \innercustomgeneric
  }
  {\endinnercustomgeneric}
}

\newcustomtheorem{customthm}{Theorem}
\newcustomtheorem{customlemma}{Lemma}
\newcustomtheorem{customdef}{Definition}
\newcustomtheorem{customcor}{Corollary}

\ifnum\draft=1
\newcommand{\questionc}[1]{\textcolor{red}{\textbf{Question:} #1}} 
\newcommand{\discuss}[1]{\textcolor{red}{\textbf{Discuss:} #1}} 
\newcommand{\notec}[1]{\textcolor{blue}{\textbf{Note:} #1}}
\newcommand{\salil}[1]{{ {\color{olive}{(salil)~#1}}}}
\newcommand{\cynthia}[1]{{ {\color{teal}{(cynthia)~#1}}}}
\newcommand{\silvia}[1]{{ {\color{purple}{(silvia)~#1}}}}
\newcommand{\todo}[1]{{ {\color{red}{(Todo)~#1}}}}
\else 
\newcommand{\questionc}[1]{}
\newcommand{\discuss}[1]{} 
\newcommand{\notec}[1]{}
\newcommand{\salil}[1]{}
\newcommand{\cynthia}[1]{}
\newcommand{\silvia}[1]{}
\newcommand{\todo}[1]{}
\fi

\title{Complexity-Theoretic Implications of Multicalibration\ifnum\draft=1 \\ {\small \textsc{Working Draft: Please Do Not Distribute}}\fi}

\ifnum\anonymize=1
 \author{}
\date{}
\else
 \author{S\'ilvia Casacuberta\thanks{Email: \texttt{silvia.casacuberta.puig@cs.ox.ac.uk}. Supported by the Herchel Smith Undergraduate Science Research Program.}
        \and
        Cynthia Dwork\thanks{Email: \texttt{dwork@seas.harvard.edu}. Supported by grant G-2020-13941 of the Alfred P. Sloan Foundation and the Simons Foundation collaboration project 733782.}
        \and 
        Salil Vadhan\thanks{Email: \texttt{salil\_vadhan@harvard.edu}. Supported by a Simons Investigator Award.} \hspace{1cm}
        }
\date{
Harvard University\\[2 ex]
July 29, 2024}
\fi

\begin{document}

\maketitle

\begin{abstract}

We present connections between the recent literature on multigroup fairness for prediction algorithms and classical results in computational complexity.
\emph{Multiaccurate} predictors are correct in expectation on each member of an arbitrary collection of pre-specified sets.
\emph{Multicalibrated} predictors satisfy a stronger condition: they are calibrated on each set in the collection.  

Multiaccuracy is equivalent to a regularity notion for functions defined by Trevisan, Tulsiani, and Vadhan (2009).  
They showed that, given a class $\F$ of (possibly simple) functions, an arbitrarily complex function $g$ can be approximated by a low-complexity function $h$ that makes a small number of oracle calls to members of $\F$, where the notion of approximation requires that $h$ cannot be distinguished from $g$ by members of~$\F$.  This complexity-theoretic Regularity Lemma is known to have implications in different areas, including in complexity theory, additive number theory, information theory, graph theory, and cryptography.  
Starting from the stronger notion of {\em multicalibration}, we obtain stronger and more general versions of a number of applications of the Regularity Lemma, including the Hardcore Lemma, the Dense Model Theorem, and the equivalence of conditional pseudo-min-entropy and unpredictability.  
For example, we show that \emph{every} boolean function (regardless of its hardness) has a small collection of disjoint hardcore sets, where the sizes of those hardcore sets are related to how balanced the function is on corresponding pieces of an efficient partition of the domain. 

\end{abstract}

\thispagestyle{empty}
\newpage

{
  \hypersetup{linkcolor=black}
  \tableofcontents
}

\thispagestyle{empty}

\newpage
\pagenumbering{arabic}

\section{Introduction}\label{sec:introduction}

In this paper, we give novel complexity-theoretic consequences of recent results in the algorithmic fairness literature regarding ``multicalibration.''  Before stating our results we review the concept of multicalibration as well as the complexity-theoretic Regularity Lemma that provides the context for our theorems.

\subsection{Multicalibration in algorithmic fairness} \label{sec:intro-mc}

Algorithms increasingly inform decisions that can deeply affect our lives, from hiring, to healthcare diagnoses, to granting of release on bail.  A major concern that arises in this context is whether or not prediction algorithms are \textit{fair} across different subpopulations and minority groups \cite{buolamwini2018gender, vincent2018amazon, larson2016we}. 
Part of the effort within the field of algorithmic fairness 
has focused on developing mathematical frameworks to formally define what it means for an algorithm to be ``fair''. 
The various definitions that have been proposed in the literature roughly fall into two main categories: individual fairness notions \cite{dhprz12, ilvento2019metric} and group fairness notions \cite{chouldechova2017fair, barocas2017fairness}.

%\smallskip
%\textbf{Multigroup fairness, multiaccuracy, and multicalibration.} 
The \textit{multigroup} framework was proposed as a way to bridge individual fairness, which requires similar treatment for individuals who are similar with respect to a given task, and group fairness, which requires that (typically disjoint) demographic groups receive similar treatment on average \cite{hkrr18, kearns2018preventing}. 
The underlying principle is the following: we want to establish a group fairness notion that is to be satisfied, simultaneously, for every one of a pre-specified collection of \textit{large, identifiable, subgroups}. 
This versatile framework allows us to consider the intersection of different minority subgroups, such as the intersection of gender, race, and socioeconomic status.  
Multicalibration in particular has proven to be a fruitful notion with wide applications including a new paradigm for loss minimization in machine learning.
%and at the same time enabling loss-minimization simultaneously for large collections of loss functions while undergoing  a {\em single} training.  
See, {\it e.g.},~\cite{barda2021addressing, dkrry21, gkrsw21, kkgkr22, gjra23, gkr23, nrrx23}.

H\'ebert-Johnson, Kim, Reingold, and Rothblum~\cite{hkrr18} 
introduced the notion of a {\em multicalibrated (MC) predictor}, which guarantees calibrated predictions across
every subpopulation from a prespecified family $\F$ of potentially intersecting subsets of $\X$.
More formally, let $\X$ be a domain of individuals.  Consider a collection $\F$ of subpopulations, each described as a boolean indicator function $f: \X \rightarrow \{0,1\}$. 
Then, given an arbitrary and unknown function $g$ mapping individuals in $\X$ to $[0,1]$, and a distribution $\D$ on $\X$, we say that a predictor $h: \X \rightarrow [0,1]$ is a \textit{$(\F,\epsilon)$-multicalibrated (MC)} predictor of $g$ with respect to $\D$ if for all $f \in \F$ and for all $v \in \image(h)$:
\begin{equation}\label{eq:intro-mc}
    \Big| \E_{x \sim \D}[f(x) \cdot (g(x) - h(x)) \mid h(x) = v] \Big| \leq \epsilon.
\end{equation}
Remarkably, H\'ebert-Johnson et al. proved that ``low-complexity'' multicalibrated predictors $h$ exist, provided we slightly relax the definition to only require (\ref{eq:intro-mc}) on level sets $h(x)=v$ that are not too small. We will state this result in more detail below in Section~\ref{sec:intro-mcpartitions}, using a more convenient formulation in terms of partitions of $\X$.  

{\em Multiaccuracy} is a relaxation of multicalibration in which the predictor is merely required to be accurate in expectation on each $f \in \F$~\cite{hkrr18}.
Formally, $h$ is an {\em $(\F,\epsilon)$-multiaccurate (MA)} predictor of $g$ with respect to distribution $\D$ if for all $f \in \F$:
 \begin{equation}\label{eq:intro-ma}
    \Big| \E_{x \sim \D}[f(x) \cdot (g(x) - h(x))] \Big| \leq \epsilon.
\end{equation}

\subsection{The Complexity-Theoretic Regularity Lemma} \label{sec:intro-regularity}

It turns out that the notion of $(\F,\epsilon)$-multiaccuracy is exactly equivalent to the notion of {\em $(\F,\epsilon)$-indistinguishability} defined and studied in 2009 by Trevisan, Tulsiani, and Vadhan~\cite{ttv09} in the complexity theory literature.  Specifically, they proved the following abstract, complexity-theoretic Regularity Lemma.

\begin{theorem}[Regularity Lemma~\cite{ttv09}, informally stated] \label{thm:intro-regularity}
For every finite domain $\X$, every function $g : \X\rightarrow [0,1]$, every distribution $\D$ on $\X$, and every $\epsilon>0$, there exists
a function $h : \X\rightarrow [0,1]$ such that:
\begin{enumerate}
    \item $h$ has ``low complexity'' relative to $\F$.  Specifically, $h$ can be computed by a boolean circuit that has $O(1/\epsilon^2)$ oracle gates instantiated with functions from $\F$ and has size $\poly(\log |\X|,1/\epsilon)$.
    \item $h$ is {\em $(\F,\epsilon)$-indistinguishable} from $g$.  That is,  
    for all $f \in \F$, we have:
 \begin{equation}\label{eq:intro-indist}
    \Big| \E_{x \sim \D}[f(x) \cdot (g(x) - h(x))] \Big| \leq \epsilon.
\end{equation}
\end{enumerate}
\end{theorem}
Notice that Condition (\ref{eq:intro-indist}) is identical to the Definition (\ref{eq:intro-ma}) of multiaccuracy.

\paragraph{Regularity and Complexity.} 
This is referred to as a {\em Regularity Lemma} because it says that an arbitrarily complex function $g$ can be `simulated' by a low-complexity function $h$, in such a way that the family $\F$ of tests cannot distinguish them.  This of a  similar spirit to
Szemer\'edi's Regularity Lemma~\cite{sze78}, whereby an arbitrarily complex graph is shown to be, in a certain sense, indistinguishable from the union of a constant number of Erd\H{o}s-R\'enyi bipartite graphs \cite{sze78}.  Indeed, in \cite{ttv09}, it is shown that Theorem~\ref{thm:intro-regularity} implies the Frieze--Kannan Weak Regularity Lemma for graphs~\cite{fk99}, which is a lower-complexity variant of Szemer\'edi's Regularity Lemma.
In a typical complexity-theoretic application, $\F$ consists of Boolean circuits of some polynomial size in the length $n=\log |\X|$ of inputs and $\epsilon=1/\poly(n)$. In this case, the Regularity Lemma says that the simulator $h$ can also be computed by circuits of size
$\poly(n)$.

In the fairness literature, as well as in uniform-complexity applications in complexity and cryptography (cf. \cite{vz13}), it is important to consider the complexity of \emph{learning} such a predictor $h$ given samples $(x, y)$ where $x \sim \D$ and $y \sim \Bern(g(x))$. 
This may be computationally hard even if $g$ is ``easy''; for example, if $g$ is in the class $\F$.
In this paper, however, we are only concerned with the oracle complexity of $h$; i.e., the complexity of evaluating $h$ given oracle gates for functions $f \in \F$ (which is trivial if $g \in \F$).
We remark that many of the learning algorithms in the fairness literature assume an agnostic learner for $\F$ as an oracle, which also trivializes the learning task if $g \in \F$.

\paragraph{Applications.}  In addition to the Frieze--Kannan Weak Regularity Lemma for graphs, Theorem~\ref{thm:intro-regularity} can be used to derive several other fundamental theorems in various areas of theoretical computer science.  These include Impagliazzo's Hardcore Lemma~\cite{imp95}, the Dense Model Theorem~\cite{gt08,tz08,rttv08}, 
Yao's XOR Lemma \cite{ttv09, gnw11}, the Leakage Simulation Lemma in leakage-resilient cryptography \cite{jp14, ccl18}, characterizations of pseudo-entropy~\cite{vz12,zhe14}, chain rules for computational entropy \cite{gw13, jp14}, Chang's Inequality in Fourier analysis of boolean functions \cite{imr14}, and equivalences between weak notions of zero knowledge \cite{clp15}. 

\paragraph{Fractional vs. Boolean Functions.} Note that we allow all of the functions $f,g,h$ to be $[0,1]$-valued rather than just Boolean.  We think of a $[0,1]$-valued function $h : \X\rightarrow [0,1]$ as representing the randomized Boolean function $\hrand : \X\rightarrow \zo$ where for all $x\in \X$ we have that $\Pr_{\coins(h)}[\hrand(x)=1] = h(x)$.  Then Condition~(\ref{eq:intro-indist}) is essentially equivalent to requiring that the distributions of $(X,\grand(X))$ and $(X,\hrand(X))$, where $X\sim \D$, are computationally indistinguishable by the family $\F$, up to an advantage of $\epsilon$.\footnote{This equivalence holds up to a small modification to the family $\F$, in particular to account from changing the domain of the functions from $\X$ to $\X\times \zo$.} 
In the Regularity Lemma, it does not matter much if we restrict $f,g,h$ to be deterministic Boolean functions, but for multicalibrated predictors it is crucial that $h$ is 
fractional (else there could only be two nonempty level sets, $h(x)=0$ and $h(x)=1$).

\subsection{Multicalibrated partitions} \label{sec:intro-mcpartitions}

It will be more convenient for us to work with an equivalent formulation of multicalibration in terms of \emph{partitions} $\Pa \subseteq 2^\X$ of the domain $\X$, %where we only require Condition~\ref{eq:intro-mc} to hold only on level sets that are not too small, 
as was done in \cite{gkrsw21,grsw22}.
The pieces $P\in\Pa$ of the partition correspond to the level sets $h(x)=v$ in (\ref{eq:intro-mc}); furthermore, it can be shown that we can assume without loss of generality that on piece $P$, we can take the value $v$ to be equal to $v_P = \E_{x \sim \D\vert_P}[g(x)]$, where $\D|_P$ denotes the distribution $\D$ conditioned on being in $P$.  In this language, the
Multicalibration Theorem can be stated as follows:

\begin{theorem}[Multicalibration Theorem \cite{hkrr18}, informally stated]\label{thm:intro-mcpartition}
Let $\X$ be a finite domain, $\F$ a class of functions $f\colon \X \rightarrow [0,1]$, $g\colon \X \rightarrow [0,1]$ an arbitrary function, $\D$ a probability distribution over $\X$, and $\epsilon, \gamma > 0$. There exists a partition $\Pa$ of $\X$ such that:
\begin{enumerate}
    \item $\Pa$ has $k=O(1/\epsilon)$ parts.
    \item $\Pa$ has ``low complexity'' relative to $\F$. 
    Specifically, there is a Boolean circuit\footnote{That is, $C$ is a circuit with Boolean gates of fan-in at most 2 that has $\lceil \log |\X|\rceil$ input gates and $\lceil \log k\rceil$ output gates.} $C : \X\rightarrow [k]$ of size $\poly(1/\epsilon,1/\gamma,\log |\X|)$ with $O(1/\epsilon^2)$ oracle gates instantiated with functions from $\F$ such that $\Pa=\{C^{-1}(1),\ldots,C^{-1}(k)\}$.
    \item \label{cond:MC} $\Pa$ is \emph{$(\F, \epsilon, \gamma)$-approximately multicalibrated} (MC) for $g$ on $\D$: that is, for all $f \in \F$ and all $P \in \Pa$ such that $\Pr_{x \sim \D}[x \in P] \geq \gamma$, we have
\begin{equation}\label{eq:intro-mcp}
    \Big|\E_{x \sim \D|_P} [f(x) \cdot (g(x) - v_P)] \Big| \leq \epsilon 
\end{equation}
where $v_P := \E_{x \sim \D|_P}[g(x)]$ and $\D|_P$ denotes the conditional distribution $\D|_{h(x)\in P}$.
\end{enumerate}
\end{theorem}
Note that the pieces $P$ of probability mass smaller than $\gamma$ (for which (\ref{eq:intro-mcp}) doesn't apply), take up at most a $\gamma\cdot k = O(\gamma/\epsilon)$ fraction of $\D$, which can be made arbitrarily small by taking $\gamma \ll \epsilon$. 

On each piece $P$ of probability mass at least $\gamma$, 
(\ref{eq:intro-mcp}) says that $g$ is indistinguishable from the constant function $v_P$, i.e. the function $h_P : P\rightarrow [0,1]$ where for all $x\in P$, $h_P(x)=v_P$.
Viewing $h_P$ as representing a randomized function $h_P^\rand$, we have $h^\rand_P(x)\sim \Bern(v_P)$ for all $x\in P$, where $\Bern(v)$ is the Bernoulli distribution with expectation $v$.   The key point is that the Bernoulli parameter $v$ is the same value (namely $v_P$) for all $x\in P$.
Thus we informally refer to $h_P$ and $h_P^\rand$ as a {\em constant-Bernoulli function}.  As we will see in the proofs of our results, indistinguishability from a constant-Bernoulli function is a very powerful condition, and this is why we are able to get so much mileage out of the MC Theorem (Theorem~\ref{thm:intro-mcpartition}).

\subsection{Our contributions}

Although the Multicalibration Theorem (Theorem~\ref{thm:intro-mcpartition}) was introduced for the purpose of algorithmic fairness, we now see that it can be viewed as a strengthening of the complexity-theoretic Regularity Lemma (Theorem~\ref{thm:intro-regularity}). 
Thus, in our work, we examine the complexity-theoretic implications of the Multicalibration Theorem. 
Doing so, we obtain stronger and more general versions of (1) Impagliazzo's Hardcore Lemma (IHCL), 
(2) Characterizations of pseudo-average-min-entropy (PAME), and (3) the Dense Model Theorem (DMT).  
In concurrent work to ours, Dwork, Lee, Lin, and Tankala~\cite{dllt23} explore an intermediate notion of graph regularity that corresponds to multicalibration and lies between Frieze-Kannan weak regularity and Szemer\'edi Regularity, and show that Szemer\'edi Regularity corresponds to a stronger notion called {\em strict} multicalibration. 
\begin{figure}[h!]
\centering
    \includegraphics[width=9.5cm]{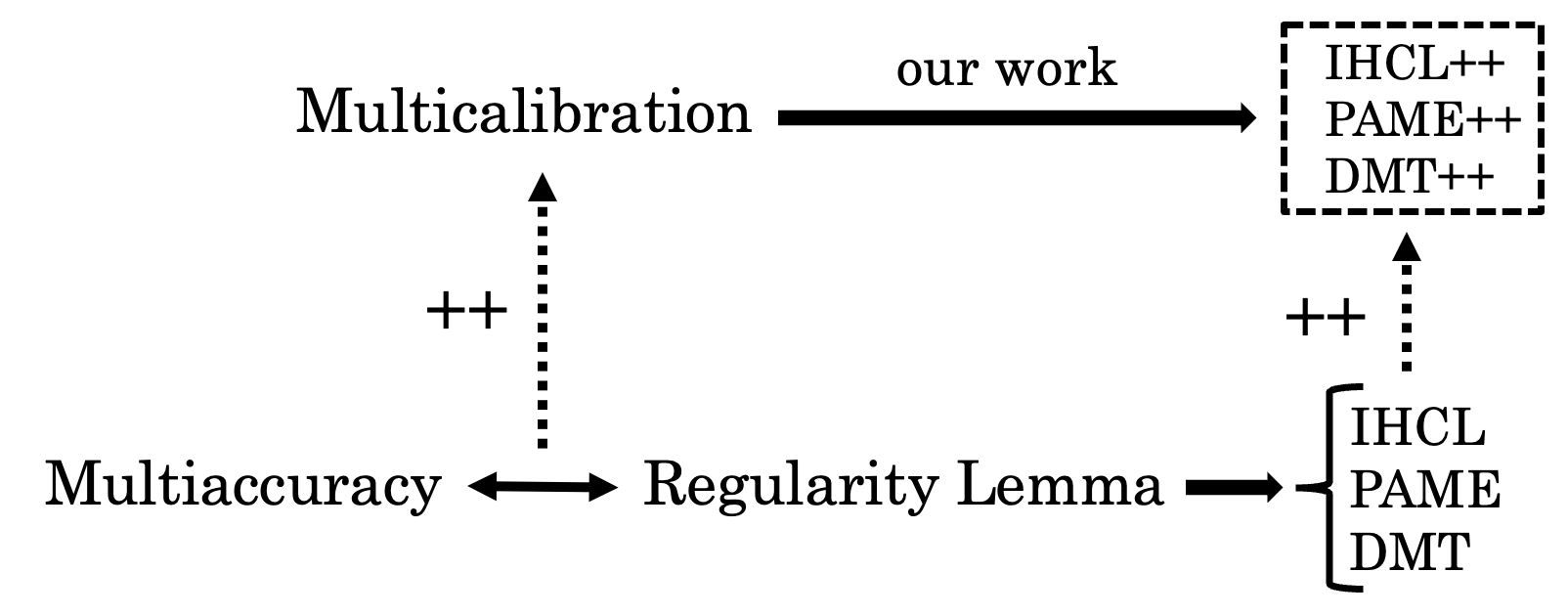}
    %\caption{Caption.}
    \label{fig:Qdiagram}
\end{figure}

We elaborate on our results below, denoting the strengthened theorems as IHCL++, PAME++, and DMT++.
For simplicity, here we will state some of the results for the special case where the initial distribution $\D$ on inputs is the uniform distribution on $\X$ and the functions $g$ and $f\in \F$ are deterministic boolean functions; generalizations to arbitrary distributions and fractional/randomized functions can be found in the later technical sections.

\medskip
 \noindent
 \textbf{Impagliazzo's Hardcore Lemma (IHCL) and IHCL++.} 
    Impagliazzo's Hardcore Lemma (IHCL) \cite{imp95} is a fundamental result in complexity theory stating that if a function $g$ is somewhat hard to compute on average by a family $\F$ of boolean functions,
    then there is a large-enough subset $H$ of the inputs (called the ``hardcore set'') for which the function is very hard to compute, in the sense that $g$ is indistinguishable from a random function by a family $\F '$ of distinguishers of complexity similar to that of $\F$.  A stronger, and optimal, version was obtained by Holenstein~\cite{hol05}:
    \begin{theorem}[IHCL~\cite{imp95, hol05}, informally stated] \label{thm:intro-IHCL}
    Let $\F$ be a family of boolean functions on $\X$, $\epsilon>0$, and
    $g: \X \rightarrow \{0, 1\}$ a function that is {\em $(\F', \delta)$-hard}, meaning that $\Pr[f(x)\neq g(x)]\geq \delta$ for all functions $f$ that have ``low complexity'' relative to $\F$.\footnote{Specifically, take $\F'$ to consist of all functions computable by a boolean circuit of size $\poly(1/\epsilon,1/\delta,\log |\X|)$ with oracle gates instantiated by functions from $\F$.}
    Then there exists a set $H \subseteq \X$ of size at least $2 \delta |\X|$ such that $g$ is $(\F, 1/2-\epsilon)$-hard on $H$.
\end{theorem}

   In \cite{ttv09}, it was shown that the Regularity Lemma implies IHCL, but with a hardcore density of $\delta$ (as in \cite{imp95}), instead of the optimal $2\delta$ from \cite{hol05}.  ($2\delta$ is optimal because erring with probability $1/2$ on a set of density $2\delta$ yields a global error probability of $\delta$.)
   
   Using the Multicalibration Theorem, we prove:
    \begin{theorem}[IHCL$++$, informal version] \label{thm:intro-ihcl++}
     Let $g: \X \rightarrow \{0,1\}$ be an arbitrary function, $\F$ a family of boolean functions, and $\epsilon>0$. There exists partition $\Pa$ of $\X$ that has ``low complexity'' relative to $\F$ such that for every (large enough) $P \in \Pa$, there is a set $H_P \subseteq P$ of size at least $2 b_P |P|$ such that $g$ is $(\F, 1/2-\epsilon_P)$-hard on $H_P$, where $b_P = \min\{\E_{x \sim P}[g(x)], 1 - \E_{x \sim P}[g(x)]\}$ and $\epsilon_P = \epsilon/b_P$.
\end{theorem}
That is, instead of finding a single, globally dense hardcore set $H$, we find many ``local'' hardcore sets $H_P$, each of which is dense within its piece $P$ of the partition.
We illustrate the key differences between IHCL and IHCL$++$ in Figure~\ref{fig:ihcl}.

\begin{figure}[h!]
\centering
    \includegraphics[width=15cm]{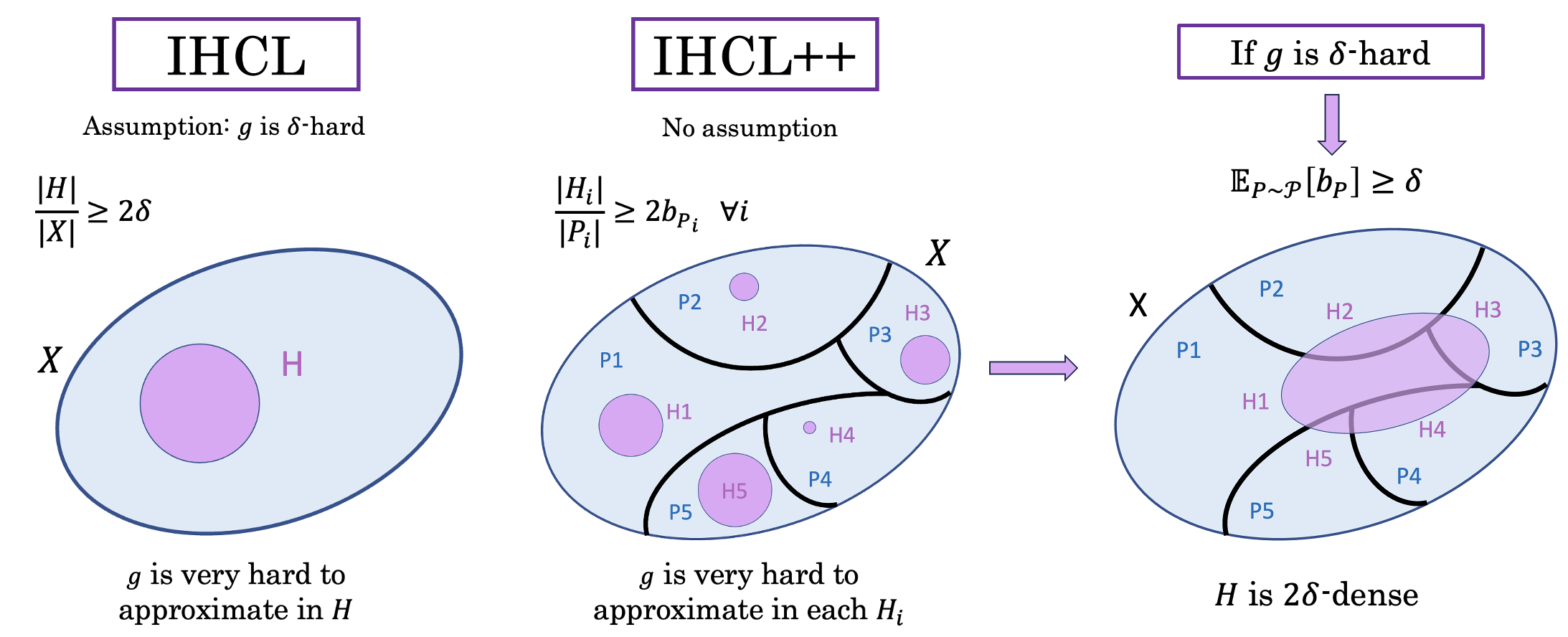}
    \caption{Illustration of the difference between IHCL and IHCL$++$, and how to recover IHCL from our ICHL$++$.}\label{fig:ihcl}
\end{figure}

Here (and in our other results), the ``low complexity'' of the partition $\Pa$ is formulated in the same way as in the MC Theorem (Theorem~\ref{thm:intro-mcpartition}).  Naturally we prove Theorem~\ref{thm:intro-ihcl++} by applying the MC Theorem to the function $g$; see Section~\ref{sec:prooftechniques} for more on our proof techniques.

The {\em balance parameter} $b_P$ provides the moral equivalent of %the role of 
the hardness parameter $\delta$ in IHCL (since we make no hardness assumptions).
We further show that our IHCL++ implies the original IHCL theorem with optimal density parameter $2\delta$.
To do so, we observe that when we bring back the assumption that $g$ is $(\F',\delta)$-hard, then $\E_{P}[b_P] \geq \delta$, where the expectation is taken over sampling piece $P$ with probability $|P|/|\X|$.  
That is, if $g$ is $\delta$-hard, then $g$ is not too imbalanced on average over the pieces of the partition.
Otherwise, we would be able to predict $g(x)$ well on average by determining which piece $P\in\Pa$ contains $x$ (because $\Pa$ has low complexity relative to $\F$) and then guessing the majority value on the piece $P$ (which we can hardwire into our Boolean circuit for each of the $k=O(1/\epsilon)$ pieces). 
We are then able to ``glue'' together the hardcore sets $H_P$ for the pieces $P \in \Pa$, yielding a hardcore set $H \subseteq \X$ that occupies at least a $2 \delta$ fraction of the domain $\X$.  (Actually we get a hardcore {\em distribution} of density at least $2\delta$, but this can be converted to a hardcore {\em set} by a standard probabilistic argument~\cite{imp95}.)

A partition-based variant of the IHCL was previously formulated and proved in the work of Reingold, Trevisan, Tulsiani, and Vadhan~\cite{rttv08}.  Their result is similar in spirit to IHCL++, but with two important differences.
First, their partition has complexity {\em exponential} in $1/\epsilon$, in contrast to the polynomial complexity we obtain through the MC Theorem.  The exponential complexity severely limits the complexity-theoretic applicability of their result.   
Second, they maintain the original assumption of IHCL that the function $g$ is weakly hard on average, whereas in our IHCL$++$ we remove it and find local hardcore sets $H_P$ for an arbitrary function $g$.  However, their proof can be modified to also remove the hardness assumption and yield a similar conclusion to ours.  Indeed, their proof, which proceeds by iteratively partitioning the domain, can be viewed in retrospect as constructing an MC partition with exponential complexity. At the time, exponential complexity seemed inherent in such iterative partitioning proofs~\cite{rttv08,ttv09}. In this light, the power of the MC Theorem and successors~\cite{hkrr18,gkrsw21,grsw22}
is that they give us the same kind of indistinguishability as provided by iterative partitioning but with polynomial complexity.  Intuitively, the savings in complexity comes from using merging steps in addition to partitioning ones to avoid making too many, too small pieces.

\medskip
\noindent
\textbf{Characterizations of pseudo-average min-entropy (PAME).} 
    A result of Vadhan and Zheng shows that we can characterize pseudoentropy, which is a computational analogue of Shannon entropy, in terms of hardness of sampling \cite{vz12}.\footnote{Vadhan and Zheng prove their results for both uniform and non-uniform families of distinguishers.  We focus on the non-uniform case.}
    In \cite{zhe14}, they also provide a characterization for the related notion of {\em pseudo-average min-entopy (PAME)}, which is the computational analogue of average min-entropy \cite{dors04}.\footnote{For a joint distribution $(X,C)$, the {\em average min-entropy} of $C$ given $X$ is defined as  $\tilde{\Hb}_{\infty}(C|X) = \log \left( 1/\E_{x \sim X}\big[2^{-\Hb_{\infty}(C|_{X=x})}\big]\right),$ where all logarithms in the paper are taken base 2 unless otherwise specified.} 
 
Informally, for a joint distribution $(X, B)$ and a class of distinguishers $\F$, $B$ has {\em $(\F, \epsilon)$-PAME at least $k$ given $X$} if there exists a random variable $C$ jointly distributed with $X$ such that 1) the distributions $(X, B)$ and $(X, C)$ are $(\F, \epsilon)$-indistinguishable, and 2) $C|X$ has average min-entropy at least $k$.
The PAME Theorem of Vadhan and Zheng shows that the PAME of $B$ given $X$ is precisely characterized by the hardness of predicting $B$ given $X$:

\begin{theorem}[PAME~\cite{vz12,zhe14}, informally stated] \label{thm:intro-PAME}
    Let $\F$ be a family of boolean functions on $\X=\zo^n\times \zo^\ell$, $\epsilon>0$, and let $(X, B)$ be a joint distribution over $\{0, 1\}^n \times \{0,1\}^{\ell}$, where $\ell = O(\log n)$.
    Suppose that $B$ is $(\F', \delta)$-hard to predict from $X$, meaning that $\Pr[f(X)\neq B]\geq \delta$ for all functions $f : \zo^n\rightarrow \zo$ that have ``low complexity'' relative to $\F$. Then $B$ has $(\F, \epsilon)$-PAME at least $\log(1/(1-\delta))$ given $X$.
\end{theorem}
The reason for the setting $\ell=O(\log n)$ in Theorem~\ref{thm:intro-PAME} is that the ``low complexity'' parameter for $\F'$ includes a factor of $2^\ell$.

Like the MC Theorem and IHCL++, our PAME++ will involve partitions. 
For $P \subseteq \X$, let $(X_P, B_P )$ denote $(X, B)$ conditioned on $X \in P$, and denote the min-entropy of a random variable $B$ by $\Hb_{\infty}(B)=\min_b \log(1/\Pr[B=b])$, where all logarithms in our paper are base 2 unless otherwise specified.

\begin{theorem}[PAME$++$, informally stated]
    Let $(X, B)$ be a joint distribution over $\{0,1\}^n \times \{0,1\}^{\ell}$, where $B=g(X)$ and $\ell = O(\log \log n)$, $\epsilon>0$, and let $\F$ be a family of boolean functions
    on $\zo^n\times \zo^\ell$.
    There exists a partition $\Pa$ of $\X = \{0,1\}^n$ that has ``low complexity'' relative to $\F$ such that $B_P$ has $(\F, \epsilon)$-PAME at least $\Hb_{\infty}(B_P)$ given $X_P$ for every (large enough) $P \in \Pa$.
\end{theorem}
That is, on each (large-enough) piece $P$ of the partition $\Pa$, the PAME of $B|_P$ given $X|_P$ is the same as the min-entropy of $B|_P$ without conditioning.  That is, $(X|_P,B|_P)$ is indistinguishable from having as much average-min-entropy as it would if $X|_P$ and $B|_P$ were independent. 

To prove this theorem, we apply
a \emph{Multiclass} Multicalibration Theorem from the literature \cite{gkrsw21,gksz22,dllt23} to the randomized function $\grand(x) \sim B|_{X=x}$. These theorems incur a doubly-exponential dependence on $\ell$ in the complexity of the partition (recall that $\ell$ is the logarithm of the size of the output set), which is why our PAME++ Theorem restricts to $\ell=O(\log\log n)$ instead of the $\ell=O(\log n)$ in the original PAME Theorem.  

Note that $\Hb_{\infty}(B_P) = \log(1/m_P)$, where $m_P = \max_{y\in \zo^\ell} \Pr[B_P=y]$ provides a multiclass generalization of our previous notion of balance $b_P$. (Specifically, $m_P=1-b_P$ when $\ell=1$.)
As in the case of IHCL$++$, the more balanced $g$ is on a cell $P\in\Pa$, the better the corresponding PAME lower bound. 
Also as in the case of IHCL$++$, we are able to recover the original PAME statement (in the case that $\ell=O(\log\log n)$) from our PAME$++$ theorem.  Specifically, when we bring back the assumption that $g$ is $\delta$-hard to predict, we have that $\E_{P}[m_P] \leq 1-\delta$, where $P$ is sampled according to its mass under $X$. 
Again, we are able to ``glue'' together the distributions $C_P$ of high average min-entropy (which are indistinguishable from $B_P$) for all the pieces $P \in \Pa$ that have enough mass according to the distribution $X$ and on which $g$ is balanced enough.
This yields a conditional distribution $C|X$ over $\{0,1\}^{\ell}$ that has PAME at least $\log(1/(1-\delta))$.

We remark that Vadhan and Zheng~\cite{vz12} also formulate and prove an analogue of the PAME Theorem for distinguishers and predictors $f$ that are given by uniform probabilistic algorithms.  
Our results only apply to the nonuniform case (e.g., Boolean circuits); a uniform treatment of multicalibration is an interesting problem for future work. 

\medskip
\noindent
\textbf{The Dense Model Theorem (DMT).} 
    This result originated from additive number theory~\cite{gt08,tz08} and states that 
    if $R$ is a pseudorandom set, whereby we mean that the uniform distribution on $R$ is indistinguishable from the uniform distribution on the entire domain $\X$ by some family of tests, and 
    $S$ is a dense subset of $R$,  then there exists a truly dense set $M \subseteq \X$ (called the {\em model} for $S$) that is indistinguishable from $S$ by a related family of tests. 
    The DMT was a crucial proof component used in Green and Tao's celebrated result that there exist arbitrarily long arithmetic progressions among the prime numbers~\cite{gt08}.
    
    A more general, complexity-theoretic version of the Dense Model Theorem was given by Reingold, Trevisan, Tulsiani, and Vadhan~\cite{rttv08}.
    Following Impagliazzo~\cite{imp08,imp09},  the idea of a dense subset $S$ of pseudorandom set $R$ can be generalized to that of a {\em pseudodense} set $S$, which is defined in the theorem statement below.

\begin{theorem}[DMT~\cite{gt08,tz08,rttv08,imp08,imp09}, informally stated] \label{thm:intro-dmt}
    Let $\F$ be a family of boolean functions on domain $\X$.
    For every $\epsilon, \delta > 0$, let $S \subseteq \X$ be $(\F', \epsilon, \delta)$-pseudodense, meaning that
    $\Pr_{x \in \X}[f(x)=1] \geq \delta \cdot \Pr_{x \in S} [f(x)=1] - \epsilon,$
    for all functions $f$ of ``low complexity'' relative to $\F$.
    Then there exists a set $M$ of density $|M|/|\X| \geq \delta - O(\epsilon)$ 
    such that the uniform distribution on $M$ is $(\F, \epsilon/\delta)$-indistinguishable from the uniform distribution on $S$.
\end{theorem}

Using the Multicalibration Theorem, we instead obtain the following: 

\begin{theorem}[DMT$++$, informally stated]
    Let $S,V$ be two disjoint sets, let $\F$ be a family of boolean functions on $S\cup V$, and let $\epsilon>0$.
    There exists a ``low-complexity'' partition $\Pa$ of $S \cup V$ such that for every (large enough) piece $P$, the uniform distribution over $P \cap V$ is $(\F, \epsilon_P)$-indistinguishable from $P\cap S$, for an appropriate choice of $\epsilon_P$.
\end{theorem}
To see the connection between this statement and the original DMT, think of $V$ as a disjoint copy of the entire domain $\X$ that $S$ lives in.  
Then we think of $P\cap V$ as a model for $P\cap S$, one that has true density $|P\cap V|/|V|$.  
If we bring back the assumption that $S$ is $\delta$-pseudodense, then we can take an appropriate convex combination of the models $P\cap V$ to obtain a distribution that is truly $\delta$-dense in $\X$ and is indistinguishable from $S$.  
Specifically, our model distribution will select $P$ with probability $|P\cap S|/|S|$ and then return a uniformly random element of $P\cap V$.  
As usual, we will only use the pieces $P$ that are both large enough and not too imbalanced; i.e., ones where neither $|P\cap S|/|S|$ nor $|P\cap V|/|V|$ is too small.  
We remark that, given a dense model distribution, it is possible to obtain a dense set as in Theorem~\ref{thm:intro-dmt} by a standard probabilistic argument, similarly to the IHCL~\cite{imp95}.

To prove Theorem~\ref{thm:dmt++}, we apply the Multicalibration Theorem to characteristic function of the set $S$ on the distribution that selects a uniformly random element of $S$ with probability 1/2 and selects a uniformly random element of $V$ with probability $1/2$.

\subsection{Common themes}

\label{sec:themes} 

There are several common themes that recur across all of our $++$ theorems.

\emph{Reproducing the original theorems locally.} 
Because multicalibration yields a low-complexity partition of the domain, all of our stronger and more general $++$ theorems find a low-complexity partition of the domain such that the original theorem is reproduced ``locally'' in each of the pieces of the partition simultaneously. 
In the case of IHCL, we find a hardcore set within each piece of the partition. 
In the case of PAME, we construct a distribution within in each piece of the partition.
For the DMT, we construct a model within each piece of the partition. 

\emph{Our theorems remove the original assumptions.} 
All of the IHCL, PAME, and DMT theorems assume some kind of computational hardness in their hypotheses: In IHCL, the function $g$ is assumed to be weakly hard on average. In the case of PAME, $B$ is assumed to be unpredictable given $X$.
In the case of the DMT, the set $S$ is assumed to be pseudodense in $\X$. 
In all of the three original theorems, given the hardness assumption, an object is then constructed achieving: 1) A stronger hardness condition, and 2) Some density guarantee.

In the case of IHCL (Theorem~\ref{thm:intro-IHCL}), the original theorem finds a subset $H$ of the domain  such that 1) the input function is maximally unpredictable inside $H$, and 2) $H$ occupies at least a certain fraction of the domain.
For PAME (Theorem~\ref{thm:intro-PAME}), the original theorem builds a conditional distribution $C|X$ that is 1) indistinguishable from the input distribution, and 2) that has at least some amount of average min-entropy.
For the DMT (Theorem~\ref{thm:intro-dmt}), the original theorem constructs a set that is 1) indistinguishable from uniform distribution on the pseudodense set $S$ and 2) satisfies a lower bound on its density. 

When we use the Multicalibration theorem instead of the Regularity/Multiaccuracy theorem, we find that 
we can maintain 1) the strong hardness conclusion in each piece of the partition, but \emph{without requiring the hardness assumption in the hypotheses}, and so our stronger theorems hold for an arbitrary function (in the case of IHCL$++$), an arbitrary conditional distribution (in the case of PAME$++$), and an arbitrary set (in the case of DMT$++$).
That is, we show that \emph{every} input object admits a certain regular partition/decomposition, where this ``regularity'' captures the required hardness/unpredictability/indistinguishability in the sense of the original theorems.

%\smallskip
\emph{Hardness vs Density.} 
In the original theorems, the density lower bounds depend on a parameter that is given by the hardness assumption.
For example, in the case of IHCL and PAME, the input function/distribution is assumed to be $\delta$-hard for some $\delta$, and then the hardcore set given by IHCL is shown to have density at least $2\delta$ and the indistinguishable distribution given by PAME is shown to have average min-entropy at least $\log(1/(1-\delta))$.  In DMT, the input set $S$ is assumed to be $\delta$-pseudodense, and the conclusion gives a model of true density at least $\delta$.
Given that our $++$ theorems remove the hardness assumption, we no longer have such a hardness parameter. Instead, the density lower bounds that we show in our stronger $++$ theorems relate to how balanced our function/object is on each piece of the partition. 

%\smallskip
\emph{Recovering the original theorems.} 
Moreover, our $++$ theorems are also stronger than the original ones because we are able to derive the original theorems as a corollary of our new results. 
Indeed, we show that if we bring back the assumption from the original theorems, we can ``stitch'' together the objects that we have built within each piece $P$ in our $++$ theorems (i.e., hardcore sets in the case of IHCL$++$, high-entropy distributions in the case of PAME$++$, and dense models in the case of DMT$++$) so that the resulting object satisfies the conclusion of the original theorems.  We are able to lower-bound the weighted average of the balance parameters in terms of the hardness parameter from our assumption, crucially using the fact that the partition is of low complexity.

These themes are summarized in Table~\ref{tab:summary}.

\begin{table}[h]
\centering
\def\arraystretch{1.3}
\scalebox{0.7}
{\begin{tabular}{V{4} c|c|c|c|c|c|c V{4}}
%{\begin{tabular}{|c|c|c|c|c|c|c|}
\hlineB{4}
%\hline
& \textbf{IHCL} & \textbf{IHCL$++$} & \textbf{PAME} & \textbf{PAME$++$} & \textbf{DMT} & \textbf{DMT$++$} \\
\hline
\textbf{Input object} & \multicolumn{2}{c|}{Function $f$} & \multicolumn{2}{c|}{Joint distribution $(X, B)$} & \multicolumn{2}{c V{4}}{Sets $S, V$} \\
%\multicolumn{2}{c|}{Sets $S, V$} \\
\hline 
\raisebox{-0.5cm}{\textbf{Assumption}} & \raisebox{-0.5cm}{$f$ is} & & \raisebox{-0.5cm}{$B$ is} & & \raisebox{-0.5cm}{$S$ is $(\F', \epsilon, \delta)$-} & \\
\textbf{on the input} & $(\F', \delta)$-hard & \raisebox{0.3cm}{---} & $(\F', \delta$)-hard & \raisebox{0.3cm}{---}  & pseudodense & \raisebox{0.3cm}{---} \\[0.1cm]
&  & & \raisebox{0.3cm}{given $X$} &  & & \\[-0.2cm]
\hline 
\raisebox{-0.7cm}{\textbf{Indist.}} & \raisebox{-0.6cm}{$\exists$ hardcore } & \raisebox{-0.4cm}{$\exists$ low-complexity} & \raisebox{-0.6cm}{$\exists$ dist. $C$ s.t.} & \raisebox{-0.6cm}{$\exists$ $O(1/\epsilon)$ dists. $C_P$ s.t.} & \raisebox{-0.6cm}{$\exists$ set $M$ s.t.} & \raisebox{-0.5cm}{$\exists$ $O(1/\epsilon)$ sets} \\
\raisebox{-0.1cm}{\textbf{guarantee}} & set $H$ & \raisebox{0.13cm}{partition w/ $O(1/\epsilon)$} & \footnotesize{$(X, C) \approx_{(\F, \epsilon)} (X, B)$} & \footnotesize{$(X_P, B_P) \approx_{(\F, \epsilon)} (X_P, C_P)$} & $\U_M \approx_{(\F, \epsilon/\delta)} \U_S$  &  $S_P, V_P$ s.t. \\
& & \raisebox{0.3cm}{hardcore sets $H_P$} &  &  & & \raisebox{0.2cm}{\footnotesize{$\U_{S_P} \approx_{(\F, \epsilon_P)} \U_{V_P}$}} \\
\hline
\raisebox{0.5cm}{\textbf{Density}} & \raisebox{0.2cm}{$|H|/|\X| \geq 2\delta$} & \raisebox{0.2cm}{$|H_P| / |\X| \geq 2 b_P$} & \raisebox{0.5cm}{\phantom{\huge X}$\hspace*{-0.5cm}\tilde{\Hb}_{\infty}(C|X)$}
%\geq \log(1/(1-\delta))$} 
& \raisebox{0.2cm}{$\tilde{\Hb}_{\infty}(C_P|X_P) \geq \Hb(B_P)$} & \raisebox{0.2cm}{$|M|/|\X| \geq \delta - O(\epsilon)$} & \raisebox{0.2cm}{$|P \cap V| / |V|$} \\[-0.7cm]
\raisebox{-0.1cm}{\textbf{guarantee}} & & & \raisebox{-0.1cm}{$\geq \log(1/(1-\delta))$} 
&  &  &  
\\[0.2cm]
\hline
\raisebox{-0.1cm}{\textbf{Recovery}} & \raisebox{-0.1cm}{---} & \raisebox{-0.1cm}{Recovers IHCL} & \raisebox{-0.1cm}{---} & \raisebox{-0.1cm}{Recovers PAME} & \raisebox{-0.1cm}{---} & \raisebox{-0.1cm}{Recovers DMT} \\
[0.2cm]
\hlineB{4}
%\hline
\end{tabular}}
    \caption{Summary of how the results in this paper generalize the original IHCL, PAME, and DMT theorems. The notation $\U_{S}$ denotes the uniform distribution over the set $S$, and $\approx_{(\F, \epsilon)}$ denotes $(F, \epsilon)$-indistinguishability.}
    \label{tab:summary}
\end{table}

\subsection{Proof techniques}
\label{sec:prooftechniques}
As discussed in Section~\ref{sec:intro-mcpartitions}, the power of the Multicalibration Theorem (Theorem~\ref{thm:intro-mcpartition}) is that it partitions the domain $\X$ into pieces $P$ such that on each piece, $g$ is indistinguishable from a constant-Bernoulli function.  This is a powerful condition, as we see through the following lemma.

\begin{lemma}[Characterizing Indistinguishability from Constant-Bernoulli Functions, informally stated]
Let $\F$ be a family of boolean functions on $\X$, and let $\epsilon>0$. Suppose that $g : \X\rightarrow \zo$
is $(\F,\epsilon)$-indistinguishable from a constant-Bernoulli function with 
expectation $v=\E_{x\in \X}[g(x)]$ and balance $b=\min\{v,1-v\}$. Then, each of the following hold up to small changes in the family $\F$ and/or the parameter $\epsilon$ (denoted as $\F'$ and $\epsilon'$):
\begin{enumerate}
\item There is a set $H$ of density $2b$ in $\X$ such that $g$ is $(\F',1/2-\epsilon'/b)$-hard on $H$.
\item The distribution $(X,g(X))$ is $(\F',\epsilon)$-indistinguishable from the distribution $(X,\Bern(v))$, where $X$ is sampled uniformly from $\X$.
\item The uniform distributions on $g^{-1}(1)$ and $g^{-1}(0)$ are $(\F',\epsilon'/b)$-indistinguishable from each other. 
\end{enumerate}
\end{lemma}
The three parts of the lemma are what we use in the proofs of the IHCL$++$, PAME$++$, and IHCL$++$, respectively.  Indeed, $H$ is our local hardcore set, $\Bern(v)$ is our local distribution of high average min-entropy, and $g^{-1}(0)$ is our local dense model of $g^{-1}(1)$. 

\subsection{Organization of the paper}
We begin by introducing the necessary definitions and notation in Section~\ref{sec:prelims}, including the precise Multicalibration Theorem, followed by our lemma characterizing indistinguishability from constant-Bernoulli functions.
Section~\ref{sec:ihcl} is devoted to Impagliazzo's Hardcore Lemma and our IHCL$++$ theorem.
Section~\ref{sec:pame} presents our results related to the PAME theorem, and Section~\ref{sec:dmt} is concerned with the Dense Model Theorem.
We conclude in Section~\ref{sec:conclusions} by discussing directions for future work.
An expository and preliminary version of the results presented in this paper can be found in the thesis of one of the authors \cite{casacuberta2023thesis}, although there are some relevant differences content-wise.

\section{Notation and Preliminaries}\label{sec:prelims}

We denote the domain by $\X$, the class of distinguishers by $\F = \{f\}$, the function to which we apply the multicalibration theorem by $g$, and the MA/MC predictor by $h$. 
We use $\Pa = \{P\} \subseteq 2^{\X}$ to denote a partition of the domain.
The notation $\Pr_{x \in \X}$ means that $x$ is sampled uniformly from $\X$, whereas $x \sim \D$ denotes that $x$ is sampled according to distribution $\D$.
We denote the constant 0 and 1 functions by $\mathbf{0}$ and $\mathbf{1}$, respectively.
We say that a class of functions $\F$ is \emph{closed under negation} if for all $f \in \F$, the function $-f$ is also in $\F$.
All the logarithms in this paper are assumed to be in base 2.

Following the intuition that we provided in Section~\ref{sec:introduction}, the formal definition of multiaccuracy is as follows:

\begin{definition}[Multiaccuracy \cite{hkrr18,kgz19}]\label{def:ma}
{\rm
Let $\X$ be a finite domain, $\F$ a collection of functions $f\colon \X \rightarrow [0,1]$, $g\colon \X \rightarrow [0,1]$ an arbitrary function, $\D$ a probability distribution over $\X$, and $\epsilon > 0$. We say that $h\colon \X \rightarrow [0,1]$ is an $(\F, \epsilon)$\emph{-multiaccurate} (MA) predictor for $g$ on $\D$ if, for all $f \in \F$,
\[
    \Big| \E_{x \sim \D}[ f(x) \cdot (g(x) - h(x))] \Big| \leq \epsilon.
\]
}
\end{definition}

We think of a $[0,1]$-valued function $f$ as describing a randomized $\{0,1\}$-valued function $f^{\rand}$ where $\Pr_{\coins(f^{\rand})}[f^{\rand}(x)=1] = f(x)$. 
Then,
\[
    \E_{x \sim \D}[f(x) \cdot (g(x)-h(x))] = \E_{\substack{x \sim \D, \, \coins(f^{\rand}), \\ \coins(g^{\rand}), \,\coins(h^{\rand})}}\big[f^{\rand}(x) \cdot (g^{\rand}(x) - h^{\rand}(x))\big].
\]
Therefore, the notion of multiaccuracy applies to the randomized functions as well.

The starting point of this work is the observation that multiaccuracy corresponds \emph{exactly} to the classical notion of indistinguishability with respect to a class of functions:

\begin{definition}[$(\F, \epsilon)$-indistinguishability \cite{ttv09}]\label{def:indist}
\rm{
    Let $\X$ be a finite domain, $\F$ a class of functions $f\colon \X \rightarrow \{0,1\}$, $g\colon \X \rightarrow [0,1]$, $\D$ a distribution on $\X$ and $\epsilon>0$. We say that a function $h\colon \X \rightarrow [0,1]$ is $(\F, \epsilon)$-\emph{indistinguishable} from $g$ on $\D$ if, for all $f \in \F$,
    \[
        \Big| \E_{x \sim \D}[f(x) \cdot (g(x) - h(x))] \Big| \leq \epsilon.
    \]
}
\end{definition}

Multicalibration is a stronger notion than multiaccuracy, where the predictor $h$ satisfies that $\E_{x \sim \D|_{h(x)=v}}[f(x) \cdot (g(x) - v)] \leq \epsilon$ for every $v \in \textrm{range}(h)$ and every $f \in \F$ \cite{hkrr18}, where $\D|_{h(x)=v}$ denotes the conditional distribution.
Thus the level sets of $h$ induce a partition $\Pa$ of the domain, and as we show in Appendix~\ref{sec:mcpartitions}, the value of $h$ in each piece $P \in \Pa$ can be made to be equal to the expected value of $g$ over $P$, which we denote by $v_P$:

\begin{definition}[Balance of $g$]\label{def:vpbp}
    Given an arbitrary function $g: \X \rightarrow [0,1]$ and a partition $\Pa = \{P\}$ of $\X$, we let $v_P = \E_{x \sim \D|_P}[g(x)]$ for each $P \in \Pa$ and $b_P = \min\{v_P, 1-v_P\}\leq 1/2$, where $\D|_P$ denotes the conditional distribution $\D|_{h(x)\in P}$
    We call $b_P$ the \emph{balance} of $g$ on $P$.
\end{definition}

In particular, $b_P=1/2$ corresponds to $g^{\rand}$ being perfectly balanced; i.e.,
\[
    \Pr_{\substack{x \sim \D|_P, \\ \coins(g^{\rand})}}[g^{\rand}(x)=1] = \Pr_{\substack{x \sim \D|_P, \\ \coins(g^{\rand})}}[g^{\rand}(x)=0] = 1/2,
\]
whereas $b_P = 0$ corresponds to $g$ being completely imbalanced; i.e., $g^{\rand}(x)$ is always 0 or 1.

Moreover, as explained in the introduction, we need to relax the notion of multicalibration to \emph{approximate multicalibration} by introducing a lower bound $\gamma$ on the size of each $P \in \Pa$ (according to distribution $\D$).
This yields the definition of an \emph{approximate MC partition}:

\begin{definition}[Approximate MC partition]\label{def:mcpartition}
{\rm
Let $\X$ be a finite domain, $\F$ a class of functions $f\colon \X \rightarrow [0,1]$, $g\colon \X \rightarrow [0,1]$ an arbitrary function, $\D$ a probability distribution over $\X$, and $\epsilon, \gamma > 0$. We say that a partition $\Pa$ of $\X$ is \emph{$(\F, \epsilon, \gamma)$-approximately multicalibrated} (MC) for $g$ on $\D$ 
if for all $f \in \F$ and all $P \in \Pa$ such that $\Pr_{x \sim \D}[x \in P] \geq \gamma$,
\[
    \Big|\E_{x \sim \D|_P} [f(x) \cdot (g(x) - v_P)] \Big| \leq \epsilon
\]
where $v_P := \E_{x \sim \D|_P}P[g(x)]$ and $\D|_P$ denotes the conditional distribution $\D|_{h(x)\in P}$.
}
\end{definition}

Note that achieving Definition~\ref{def:mcpartition} is trivial if we allow $O(1/\gamma)$ pieces $P$.
We will want to achieve a partition that is much smaller; the goal is to satisfy approximate multcalibration with only $O(1/\epsilon)$ pieces, where $\epsilon \gg \gamma$.
This turns out to be possible, as demonstrated by Theorem~\ref{thm:mcpartition}.

In the case where $\D$ corresponds to the uniform distribution over $\X$, then $\Pr_{x \sim \D}[x \in P] = |P|/|\X|$.
That is, in this case, Definition~\ref{def:mcpartition} should be understood as saying that we do not make any guarantees about sets that are too small (namely, about sets that occupy less than a $\gamma$ fraction of the space). Additionally, in order to keep track of the impact of the size of each $P \in \Pa$, we introduce the following notation:

\begin{definition}\label{def:densityparam}
{\rm
Given $P \subseteq \X$, we %ill 
let $\eta_P = \Pr_{x \sim \D}[x \in P]$ denote the \emph{size} parameter of $P$ in $\X$. If $\D$ corresponds to the uniform distribution over $\X$, then $\eta_P := |P| / |\X|$.}
\end{definition}

We also use the following notation for sampling the pieces $P$ from $\Pa$:

\begin{definition}
    Given a partition $\Pa$ of $\X$ and a distribution $\D$ over $\X$, $\Pa(\D)$ denotes the distribution on $\Pa$ that selects each $P \in \Pa$ with probability $\sum_{x \in P} \D(x)$.
\end{definition}

Next, we study the notion of complexity of a partition.
We use the number of wires of a circuit as the circuit size measure.

\smallskip
\textbf{Complexity of a partition.} As we developed in the introduction, a key property of a multicalibrated partition is that it is a \textit{low-complexity} partition of the domain $\X$. We now formalize this idea.

\begin{definition}[Relative complexity of a function {\cite[Definition 6]{jp14}}]\label{def:relcomplexity}
{\rm
Let $\F$ be a family of functions $f\colon \X \rightarrow [0,1]$. A function $h$ has \emph{complexity $(t, q)$ relative to $\F$} if it can be computed by an oracle-aided circuit of size $t$ with $q$ oracle gates, where each oracle gate is instantiated with a function from~$\F$.}
\end{definition}

The notion of relative complexity captures the idea that we can make oracle calls to the functions in $\F$ without these factoring into the complexity. 
The algorithms to construct MA and MC predictors $h$ use an oracle for a weak agnostic learner for the family $\F$ \cite{ttv09, hkrr18, gkrsw21, ghkrs23}.
In such a case, the parameter $q$ corresponds to the number of oracle calls to the weak agnostic learner $q$. 

%In the multicalibration literature, the MC algorithms use a weak agnostic learner for the oracle calls \cite{hkrr18, gkrsw21, ghkrs23}. 
%Hence, the complexity parameter that we are interested in for our theorems is the number of calls to the weak learner during the MC algorithm, which is captured by the parameter $q$ above.

\begin{definition}\label{def:Ftq}
{\rm
Given an arbitrary class of functions $\F$, we denote by $\F_{t, q}$ the class of functions that have complexity at most $(t, q)$ relative to $\F$.}
\end{definition}

\begin{remark}\label{remark:dropqs}
{\rm
From Definitions \ref{def:relcomplexity} and \ref{def:Ftq} it follows that when working with the class $\F_{t, q}$ we can always assume that $q \leq t$.}
\end{remark}

\begin{remark}
{\rm
In some applications, $\F$ is a class of distinguishers implemented by size-$s$ circuits. In that case, every function in $\F_{t, q}$ can be computed by a circuit
of size $t + sq$.}
\end{remark}

That is, in measuring complexity, $t$ is additive, whereas $q$ is multiplicative. 
In combinatorial applications, usually only $q$ matters. The parameter $q$ can be independent of $|\X|$, but $t$ generally cannot, as at least $\log|\X|$ wires are needed to read an input from $\X$.

\begin{definition}\label{def:Ftqk}
{\rm
Given a set of functions $\F = \{f\}$ on a finite domain $\X$, $\F_{t, q, k}$ denotes the class of partitions $\Pa$ of $\X$ such that there exists  $\hat{f} \in \F_{t, q}$, $\hat{f}: \X \rightarrow [k]$, satisfying $\Pa = \{\hat{f}^{-1}(1), \ldots, \hat{f}^{-1}(k)\}$.
}
\end{definition}

The condition $P_i = \hat{f}^{-1}(i)$ stated in Definition~\ref{def:Ftqk} ensures that we can always know to which level set each $x \in \X$ belongs to by performing an oracle call to a function in $\F_{t, q}$. 
Intuitively, we are associating each $P \in \Pa$ with an integer in $[k]$, and then Definition~\ref{def:Ftqk} requires the existence of a function in $F_{t, q}$ that we use to query to which $P$ each $x \in \X$ belongs to. 
We call this function $\hat{f}$ the \textit{partition membership function}. 
Naturally, this $\hat{f}$ is constant on each $P \in \Pa$.

Having formalized the complexity class $\F_{t, q, k}$ of partitions, we can now state the theorem that is the backbone of all our results in this paper:

\begin{theorem}[Multicalibration Theorem \cite{hkrr18}]\label{thm:mcpartition}
Let $\X$ be a finite domain, $\F$ a class of functions $f\colon \X \rightarrow [0,1]$, $g\colon \X \rightarrow [0,1]$ an arbitrary function, $\D$ a probability distribution over $\X$, and $\epsilon, \gamma > 0$. There exists an $(\F, \epsilon, \gamma)$-approximately multicalibrated partition $\Pa$ of $\X$ for $g$ on $\D$  
such that $\Pa \in \F_{t, q, k}$, where
\begin{itemize}
    \item[{\rm 1.}] $t = O(1/(\epsilon^4 \gamma) \cdot \log(|\X|/\epsilon))$,
    \item[{\rm 2.}] $q = O(1/\epsilon^2)$,
    \item[{\rm 3.}] $k = O(1/\epsilon)$.
\end{itemize}
\end{theorem}

We defer to Appendix~\ref{sec:mcpartitions} for the proof of Theorem~\ref{thm:mcpartition}, given that it follows from tweaking well-known algorithms from the multicalibration literature. 
Whenever we use the term \emph{low-complexity partition} in this paper, we are referring to a partition with the complexity parameters stated in Theorem~\ref{thm:mcpartition}.
For our results regarding the PAME theorem, we will require a more general multicalibration theorem; namely, one that is adapted to the multiclass case, where the function $g$ maps $\X$ to $\{0,1\}^{\ell}$.
We defer the presentation and discussion of the Multiclass Multicalibration Theorem to Section~\ref{sec:pame}.

\subsection{Hardness and indistinguishability notions}\label{sec:hardnessindist}

One of the contributions of our paper is to provide a complexity-theoretic perspective on the power of a multicalibrated partition.
Indeed, as we developed in the introduction, given the definition of $(\F, \epsilon)$-indistinguishability, the MC partition theorem is equivalent to stating that, given an arbitrary function $g$, we can find a low-complexity partition of the domain such that $g$ is indistinguishable from the constant function $v_P$ on each piece $P$ of the partition.

In order to illustrate why this is a powerful statement, in this section we relate indistinguishability from a constant function to Yao's lemma on the equivalence between pseudorandomness and unpredictability \cite{yao82}.
To do so, we formally define what it means for a function to be \emph{hard} with respect to a class of functions $\F$:

\begin{definition}[Hardness of a function]\label{def:hardnessfn}
    {\rm
    Given a class $\F$ of randomized functions $f\colon \X \rightarrow \{0,1\}^{\ell}$, a distribution $\D$ on $\X$, an arbitrary randomized function $g\colon \X \rightarrow \{0,1\}^{\ell}$, and $\delta > 0$, we say that $g$ is \emph{$(\F, \delta)$-hard} on $\D$ if, for all $f \in \F$,
    \[
        \Pr_{\substack{x \sim \D, \, \coins(f) \\ \coins(g)}}[f(x) = g(x)] \leq 1-\delta.
    \]
    }
\end{definition}

Note that here we consider randomized functions with discrete range rather than deterministic functions with range $[0,1]$.
When $\ell=1$, the maximal possible hardness occurs when $\delta = 1/2 - \epsilon$, given that being $(\F, 1/2-\epsilon)$-hard corresponds to stating that 
\[
    \Pr_{\substack{x \sim \D, \, \coins(f) \\ \coins(g)}}[f(x) = g(x)] \leq 1/2 + \epsilon.
\]
That is, no distinguisher in $\F$ can guess $g$ noticeably better than a random bit. 
This is why we sometimes refer to being $(\F, 1/2-\epsilon)$-hard as being $\epsilon$-\emph{strongly hard}, whereas being $(\F, \delta)$-hard is sometimes referred to as being $\delta$\emph{-weakly} hard.
In the case of Impagliazzo's Hardcore Lemma, the task is precisely to find a subset $H$ of the domain on which $g$ is $(\F, 1/2-\epsilon)$-hard, and hence maximally unpredictable.

Yao first showed a relationship between pseudorandomness (i.e., indistinguishability from a constant $1/2$ function $h$, for which $\Pr_{\coins(h)}[h^{\rand}(x) = 1] = \Pr_{\coins(h)}[h^{\rand}(x)=0] = 1/2$) and unpredictability:

\begin{lemma}[Equivalence between indistinguishability and pseudorandomness \cite{yao82}]\label{lemma:yao}
    Given a class of functions $\F$, a distribution $\D$ on $\X$ and $\epsilon>0$, a function $g: \X \rightarrow [0,1]$ such that $\E_{x \sim \D}[g(x)=1/2]$ is $(\F, \epsilon)$-indistinguishable on $\D$ from the constant $1/2$ function if and only if $g^{\rand}$ is $(\F^{\rand}, 1/2-2\epsilon)$-hard.
\end{lemma}

In other words, stating that $g$ is strongly hard corresponds exactly to stating that $g$ is indistinguishable from a uniform random bit.
Lemma~\ref{lemma:yao} follows from the identity 
\begin{equation}\label{eq:yaoid}
    \Pr_{\substack{x \sim \D, \, \coins(f^{\rand}), \\ \coins(g^{\rand})}} [f^{\rand}(x) = g^{\rand}(x)] = 2\E_{x \sim \D}[(f(x)-1/2)(g(x)-1/2)] + 1/2.
\end{equation}

Next, we relate Yao's Lemma to the MC Theorem. For it, we need to introduce the definition of the density of a distribution and extend the notion of $(\F, \epsilon)$-indistinguishability to distributions:

\begin{definition}[$\delta$-dense distribution]\label{def:densitydist}
{\rm
A distribution $A$ is \emph{$\delta$-dense} in a distribution $B$ if for all $x \in \X$,
\[
    \delta \cdot \Pr[A=x] \leq \Pr[B=x].
\]
}
\end{definition}

If $B$ is the uniform distribution on $X$, then this becomes $\Pr[A=x] \leq 1/(\delta |\X|)$; i.e., a condition that is satisfied by the uniform distribution on any set of size at least $\delta |\X|$.

\begin{definition}[Indistinguishable distributions]\label{def:indistdist}
{\rm
Given a class $\F$ of functions $f\colon \X \rightarrow [0,1]$ and two distributions $\D_1$, $\D_2$ on $\X$, we say that $\D_1$ and $\D_2$ are \emph{$(\F, \epsilon)$-indistinguishable} if, for all $f \in \F$,
\[
   \Big|\E_{x \sim \D_1}[f(x)] - \E_{x \sim \D_2}[f(x)]\Big| \leq \epsilon.
\]
}
\end{definition}

\subsection{Characterizations of constant-Bernoulli functions}\label{sec:yao}

While Yao's lemma characterizes indistinguishability from the constant $1/2$ function, in a multicalibrated partition (Definition~\ref{def:mcpartition}) we have indistinguishability from the constant $v_P = \E_{x \sim \D\vert_P}[g(x)]$ function on each piece $P$, where $v_P$ can take any value in $[0,1]$.
Thus we seek to characterize functions $g: \X \rightarrow [0,1]$ that are indistinguishable from constant functions $h: \X \rightarrow [0,1]$ (i.e., there is  a $v \in [0,1]$ such that $h(x)=v$ for all $x \in \X$).
Note that such an $h$ represents a randomized function $h^{\rand}$ such that $h^{\rand}(x)$ is identically distributed to $\Bern(v)$ for all $x \in \X$, which we refer to as a \emph{constant-Bernoulli} function.
If $g$ is $(\F,\epsilon)$ indistinguishable from a constant-Bernoulli function $h(x)=v$ on distribution $\D$ and $\F$ contains the constant $1$ function, then we can assume that $v=\E_{x\sim \D}[g(x)]$ with a factor 2 change in $\epsilon$.  Thus we state the following lemma with this assumption, which is anyhow guaranteed to us by the MC Theorem as formulated in Theorem~\ref{thm:mcpartition}. 

\begin{lemma}[Characterizing indistinguishability from constant-Bernoulli functions]\label{thm:tfae}
    Let $\F$ be a class of functions $f: \X \rightarrow [0,1]$ closed under negation and such that $\mathbf{0}, \mathbf{1} \in \F$, let $\D$ be a distribution over the domain $\X$, and let $\epsilon>0$.
    Let $g: \X \rightarrow [0,1]$ be $(\F, \epsilon)$-indistinguishable from the constant function $v$ on $\D$, where $v = \E_{x \sim \D}[g(x)]$ and $b = \min\{v, 1-v\}$.
    Then, the following statements hold:
    \begin{enumerate}
        \item\label{item:bern} The distribution $(X, g^{\rand}(X))$ is $(\F', \epsilon)$-indistinguishable from the distribution $(X, \Bern(v))$, where $X \sim \D$ and $\F'$ is any class such that $\F'_{O(n), O(1)}\subseteq \F$. 
        \item\label{item:hard} The function $g^{\rand}$ is $(\F^{\rand}, b-2\epsilon)$-hard on $\D$. 
        \item\label{item:0s1s} The distribution $\D \vert_{g^{\rand}(x)=1}$ is $\big(\F', \frac{\epsilon}{v(1-v)}\big)$-indistinguishable from $\D \vert_{g^{\rand}(x)=0}$ for any class $\F'$ such that $\F'_{c\log|\X|, c} \subseteq \F$, where $c$ is a universal constant.
        \item\label{item:dist2b} Assuming $g$ is Boolean (i.e., $g(x) \in \{0,1\}$ for all $x \in \X$), there exists a distribution of density $2b$ in $\D$ such that $g$ is $\big(\F^{\rand}, 1/2-\frac{\epsilon}{2v(1-v)}\big)$-hard on it.
    \end{enumerate}
\end{lemma}

\smallskip
\emph{Why Lemma~\ref{thm:tfae} is central to our $++$ theorems.}
As we mentioned in the introduction, we prove our $++$ theorems as consequences of Lemma~\ref{thm:tfae}.
First, Lemma~\ref{thm:tfae} is applicable to our setting because of our complexity-theoretic recasting of the definition of a multicalibrated partition.
Namely, a multicalibrated partition $\Pa = \{P\}$ for $g, \F, \D, \epsilon>0$ is such that for each (large enough) $P \in \Pa$, the function $g$ is $(\F, \epsilon)$-indistinguishable from the constant function $v = \E_{\D}[g(x)]$. 
Hence, each such piece $P \in \Pa$ satisfies the assumption of Lemma~\ref{thm:tfae}. 
We then use the statements of Lemma~\ref{thm:tfae} as follows:
\begin{enumerate}
    \item The proof of IHCL$++$ follows from statement (\ref{item:dist2b}).
    Essentially, the distribution given by statement (\ref{item:dist2b}) corresponds to a ``small'' hardcore set contained within each piece of the partition.
    \item The proof of PAME$++$ follows from statement (\ref{item:bern}).
    Roughly, when showing the existence of a distribution that has high average min-entropy and that is indistinguishable from $B\vert_P$ (as required by the definition of PAME), we use a Bernoulli distribution with parameter $v_P$.
    \item The proof of DMT$++$ follows from statement (\ref{item:0s1s}).
    While there is no function $g$ in the statement of the DMT, we define $g$ precisely as the characteristic function of the set $S$. 
    Then, the fact that $\D\vert_{g^{-1}(0)}$ and $\D\vert_{g^{-1}(1)}$ are indistinguishable allows us to argue that the sets $S \cap P$ and $U \cap P$ are indistinguishable with respect to $\F$, which in turn shows that $U \cap P$ is a model for the corresponding set $S \cap P$.
\end{enumerate}

\begin{proof}[Proof of Lemma~\ref{thm:tfae}]
We prove each of the implications separately.

(1.)
Given a class of functions $\F' = \{f': \X \times \{0,1\} \rightarrow [0,1] \}$, we construct a class of functions $\F'' = \{f'': \X \rightarrow [0,1]\}$ as follows.
For each $f' \in \F'$, we add the function $f''(x) = f'(x,1) - f'(x, 0)$ and its negation to $\F''$.
Note that $\F'' \subseteq \F'_{cn, c} \subseteq \F$.
Then,
\begin{align*}
    \E_{\substack{x \sim \D, \\ b \sim \Bern(v), \\ \coins(g^{\rand})}}[f'(x, g^{\rand}(x)) - f'(x, b)] &= 
    \E_{x \sim \D}[g(x)f'(x, 1) + (1-g(x))f'(x, 0) - v f'(x, 1) - (1-v) f'(x, 0)] \\
    &= \E_{x \sim \D}[f''(x)(g(x)-v)].
\end{align*}

(2.) Assume without loss of generality that $b = v \leq 1/2$; otherwise, we can replace $f$ and $g$ by $-f$ and $-g$ respectively in the argument below.
Then,
\begin{align*}
    \Pr_{\substack{x \sim \D, \, \coins(f^{\rand}), \\ \coins(g^{\rand})}} [f^{\rand}(x) = g^{\rand}(x)] 
    &= \E_{x \sim \D}[f(x)g(x) + (1-f(x))(1-g(x))] \\
    &= \E_{x \sim \D}[(2f(x)-1)(g(x)-v) + 1-v - (1-2v)f(x)] \\[0.3cm]
    &\leq 1 - v + 2\epsilon,
\end{align*}
given that $v \leq 1/2$ implies $(1-2v)f(x) \geq 0$ and that, by assumption, $\big| \E_{\D}[f(x) \cdot (g(x) - v)] \big| \leq \epsilon$.

(3.) By the assumption on $g$, it follows that
\begin{align*}
    \Big| \E_{x \sim \D}[f(x) \cdot (g(x) - v)] \Big| & = 
    \Big| \E_{\substack{x \sim \D, \\ \coins(g^{\rand})}}[f(x) \cdot (g^{\rand}(x)-v)] \Big| \\
    &= \Big| v \cdot \E_{\substack{x \sim \D\vert_{g^{\rand}(x)=1}, \\ \coins(g^{\rand})}}[f(x) \cdot (1-v)] + (1-v) \cdot \E_{\substack{x \sim \D\vert_{g^{\rand}(x)=0}, \\ \coins(g^{\rand})}}[f(x) \cdot (-v)] \Big|
\\[0.2cm]
    & = v \cdot (1-v) \cdot \Big| \E_{\substack{x \sim \D\vert_{g^{\rand}(x)=1}, \\ \coins(g^{\rand})}}[f(x)] - \E_{\substack{x \sim \D\vert_{g^{\rand}(x)=0}, \\ \coins(g^{\rand})}}[f(x)]\Big| \leq \epsilon.
\end{align*}
Therefore, the distributions $\D\vert_{g^{\rand}(x)=1}$ and $\D\vert_{g^{\rand}(x)=0}$ are $\big(\F, \frac{\epsilon}{v(1-v)}\big)$-indistinguishable.

(4.)
In light of Yao's Lemma, the idea for showing this implication is to define a probability distribution $\mu$ such that $\E_{x \sim \mu}[g(x)] = 1/2$. 
That is, given that $g$ has expected value $v$ when sampling according to distribution $\D$, we want to ``shift'' $v$ back to $1/2$ when sampling according to distribution $\mu$. 
Intuitively, if $v>1/2$, then we should add more weight to the points $x$ in the domain such that $g(x)=0$, and viceversa if $v \leq 1/2$.

We can do this boosting of the minority values by defining $\mu$ as follows:
\[
\mu(x) = 
    \begin{cases} 
      \dfrac{1}{2v} \cdot \D(x) & \text{ if } g(x) = 1, \\[0.4cm]
      \dfrac{1}{2(1-v)} \cdot \D(x) & \text{ if } g(x) = 0.
   \end{cases}
\]
It is direct to check that this is indeed a probability distribution.
Next, we show that the expected value of $g$ when sampling according to $\mu$ is indeed $1/2$.
Let $G^0 = \{x \in \X \mid g(x) = 0\}$ and $G^1 = \{x \in \X \mid g(x) = 1\}$. 
Then, given that $\E_{\D}[g(x)] = \sum_{x \in G^1} \D(x) = v$, it follows that
\[
    \E_{x \sim \mu}[g(x)] = \sum_{x \in \X} \mu(x) \cdot g(x) = \sum_{x \in G^1} \mu(x) = \sum_{x \in G^1} \dfrac{1}{2v} \cdot \D(x) = 
    \dfrac{1}{2v} \sum_{x \in G^1} \D(x) = 
    \dfrac{1}{2v} \cdot v = 1/2.
\]
Next, we show that $\mu$ has density $2b$ in $\D$. 
This follows directly by our construction of $\mu$, the definition of density for distributions (Definition~\ref{def:densitydist}) and the definition of $b$ as $b=\min\{v, 1-v\}$: If $x \in g^{-1}(1)$, then $2b \cdot \mu(x) \leq \D(x)$ because $\mu(x) = \frac{1}{2v} \cdot \D(x)$ and $b \leq v$.
If $x \in g^{-1}(0)$, then  $2b \cdot \mu(x) \leq \D(x)$ as well because $\mu(x) = \frac{1}{2(1-v)} \cdot \D(x)$ and $b \leq 1-v$.

Lastly, we show that $g$ is $(\F, 1/2 - \epsilon'/2)$-hard on $\mu$, where $\epsilon' = \epsilon/(v\cdot(1-v))$.
By statement \ref{item:0s1s}, we know that the distributions $\D\vert_{G^1}$ and $\D\vert_{G^0}$ are $(\F, \epsilon')$-indistinguishable. 
Let $\mu_0$ correspond to the restriction of $\mu$ on the domain $G^0$, and let $\mu_1$ correspond to the restriction of $\mu$ on the domain $G^1$. 
By the definition of $\mu$, it follows that
\begin{align*}
    \Pr_{\substack{x \sim \mu, \coins(f^{\rand}), \\ \coins(g^{\rand})}}[f(x) = g(x)] & = \dfrac{1}{2} \Pr_{\substack{x \sim \mu_1, \\ \coins(f^{\rand})}}[f(x)=1] + \dfrac{1}{2} \Pr_{\substack{x \sim \mu_0, \\ \coins(f^{\rand})}}[f(x) = 0]
\\[0.2cm]
  &  = \dfrac{1}{2} + \dfrac{1}{2} \cdot \Big(\E_{x \sim \mu_1}[f(x)] - \E_{x \sim \mu_0}[f(x)]\Big) \leq \dfrac{1}{2} + \dfrac{\epsilon'}{2},
\end{align*}
and hence $g$ is $(\F^{\rand}, 1/2-\epsilon'/2)$-hard on $\mu$, as we wanted to show.
We remark that statement (\ref{item:dist2b}) need not be restricted to boolean functions; in the case of a non-boolean function $g$ we can show it using joint distributions instead.
However, the boolean case suffices for our applications of statement (\ref{item:dist2b}) in this paper.
\end{proof}

\begin{remark}
    By the definition of $b$, it follows that $b/2 \leq v(1-v) \leq b$. 
    Therefore, statement (\ref{item:dist2b}) in Lemma~\ref{thm:tfae} implies that $g$ is $(\F^{\rand}, 1/2-\epsilon/b)$-hard.
\end{remark}

\textbf{Indistinguishability from a constant-Bernoulli function is stronger than average-case hardness.}
Lemma~\ref{thm:tfae} states unidirectional relationships, namely that statements (\ref{item:bern})-(\ref{item:dist2b}) are \emph{implied} by $g$ being indistinguishable from a constant-Bernoulli function. 
It is natural to ask whether the converses also hold, as it is the case in Yao's Lemma (Lemma~\ref{lemma:yao}).
As it can be seen from the proofs of statements (\ref{item:bern}) and (\ref{item:0s1s}) in Lemma~\ref{thm:tfae}, these two statements do imply that $g$ is indistinguishable from a constant-Bernoulli function.
However, this is \emph{not} the case for statements (\ref{item:hard}), and (\ref{item:dist2b}), which are \emph{weaker} than being indistinguishable from a constant-Bernoulli function.

The fact that statement (\ref{item:hard}) does not imply the assumption of Lemma~\ref{thm:tfae}, even when we allow large changes in $\F$ and $\epsilon$ can be shown with the following counter-example. 
Let $\D$ be the uniform distribution on $\X = \{0,1\}^n$ and let $\F$ be all circuits of size $n^c$.
Let $g$ be a random function where $\Pr[g(x)=1] = 3/4$ if $x_1 = 0$ and $\Pr[g(x)=1] = 1$ if $x_1 = 1$, where $x_1$ denotes the first bit of $x$. 
Then, by Chernoff and union bounds it can be shown that with high probability, $v \geq 7/8 - 2^{-\Omega(n)}$ and $g^{\rand}$ is $\big(v - 2^{-\Omega(n)}\big)$-hard against circuits of size $2^{\Omega(n)}$.
On the other hand, the distinguisher $f(x)=x_1$ has 
\[
    \E_{x \sim \D}[f(x)(g(x)-7/8)] = 1/2 - 7/16 = 1/16,
\]
so $g$ is not $\epsilon$-indistinguishable from the constant function $v$ for any $\epsilon < 1/16$ and circuits of size $O(1)$.
We note that the fact that statement (\ref{item:hard}) does not imply indistinguishability from a constant-Bernoulli function is in contrast to Yao's Lemma, which states that (\ref{item:hard}) $\iff$ (\ref{item:dist2b}) in the special case where $v = 1/2$.

Lastly, the reason why (\ref{item:dist2b}) is weaker than the assumption of Lemma~\ref{thm:tfae} is because (\ref{item:hard}) $\implies$ (\ref{item:dist2b}) by IHCL (given that statement (\ref{item:hard}) corresponds to the assumption of IHCL and statement (\ref{item:dist2b}) to the conclusion of IHCL), and so it follows that (\ref{item:dist2b}) does \emph{not} imply that $g$ is a constant-Bernoulli function.

\section{The Hardcore Lemma}\label{sec:ihcl}

Impagliazzo's Hardcore Lemma (IHCL) is a fundamental result in complexity theory that dates to 1995 \cite{imp95}. 
Informally, it states that if a function is somewhat hard to compute on average by a family $\F$ of boolean functions, then there is a fairly large subset of the inputs (called the ``hardcore set'') for which the function is very hard to compute, in the sense that $g$ is maximally unpredictable to the family $\F$.
That is, no distinguisher can do better than random guessing.

When we state the theorem formally, we will see that the ``somewhat hard'' assumption and the ``very hard'' conclusion are with respect to slightly different families $\F$. 
In particular, the former family is an enlarged class of distinguishers than the latter family.

\subsection{The original IHCL statement}

Impagliazzo's Hardcore Lemma can be formally stated as follows:

\begin{theorem}[IHCL, \cite{imp95, hol05}]\label{thm:ihcl}
Let $\F$ be a family of functions $f: \X \rightarrow \{0,1\}$, let $\D$ be a probability distribution over $\X$, and let $\epsilon, \delta > 0$. 
There exists $t = \mathrm{poly}(\log|\X|, 1/\epsilon, 1/\delta)$
and $q = \emph{poly}(1/\epsilon, 1/\delta)$ 
such that the following holds: If $g: \X \rightarrow \{0,1\}$ is $(\F_{t, q}, \delta)$-hard on $\D$, then there is a distribution $\Ha$ that is $2\delta$-dense in $\D$ and for which $g$ is $(\F, 1/2-\epsilon)$-hard on $\Ha$.
\end{theorem}

That is, if $g$ is weakly hard on $\D$ with respect to the class $\F_{t, q}$, then we can find a distribution $\Ha$ on which $g$ is strongly hard with respect to the smaller class $\F$ of distinguishers.
Moreover, while the conclusion of the IHLC Theorem deals with hardcore \emph{distributions}, it is possible to convert a hardcore distribution into a hardcore set \cite{imp95, ks03}, as we review in Appendix~\ref{sec:setihcl}.
An important consideration about IHCL which will be key to our proposed IHCL$++$ is that the statement is proving two different things about the distribution $\Ha$:
\begin{itemize}
    \item \textbf{Density.} 
    The distribution $\Ha$ is $2 \delta$-dense in $\D$. (We think of $\delta$ as a small constant, like 0.1.)
    The original theorem shown by Impagliazzo \cite{imp95} finds a hardcore density of density $\delta$, rather than $2 \delta$.
    This difference is important because $2 \delta$ is the \textit{optimal} density parameter for the hardcore distribution. 
    This is because if there exists a hardcore distribution for $g$ of density $\rho$, then $g$ is $\big(\rho (1/2 - \epsilon) \big)$-weakly hard on average on $\D$ with respect to $\F$. 
    It took 10 years for Holenstein to prove that we can indeed achieve the optimal $2 \delta$ density parameter \cite{hol05}. 
    However, we note that Trevisan et al.'s proof of IHCL using the Regularity Lemma is only able to recover the original $\delta$-density parameter, but not Holenstein's optimal $2\delta$-density parameter \cite{ttv09}. 
    \item \textbf{Indistinguishability.} 
    When we sample according to $\Ha$, $g$ is $(\F, \epsilon)$-strongly hard. 
    We call this the ``indistinguishability'' condition because Yao's Lemma (Lemma~\ref{lemma:yao}) tells us that this is equivalent to stating that $g$ behaves like a random function to the class of distinguishers $\F$.
\end{itemize}

\subsection{The IHCL$++$ Theorem}\label{sec:ihcl++}

For the ease of notation, recall the definitions of $v_P$ and $b_P$ (which we call the \emph{balance} of function $g$) from Definition~\ref{def:vpbp}.
From Section~\ref{sec:prelims}, recall that we also need to consider the size parameter $\eta_P = \Pr_{x \sim \D}[x \in P]$ of each $P \in \Pa$ (Definition~\ref{def:densityparam}). 
Because we are using the notion of approximate multicalibration, we will only be considering the sets $P \in \Pa$ such that $\eta_P \geq \gamma$. 

We can now introduce our IHCL$++$ statement:

\newpage

\begin{theorem}[IHCL$++$]\label{thm:ihcl++}

Let $\X$ be a finite domain, let $\F$ be a family of functions $f\colon \X \rightarrow [0,1]$, let $g\colon \X \rightarrow [0, 1]$ be an arbitrary function, $\D$ a probability distribution over $\X$, and let $\epsilon, \gamma > 0$.
There exists a partition $\Pa \in \F_{t, q, k}$ of $\X$ with $t = O(1/(\epsilon^4 \gamma) \cdot \log(|\X|/\epsilon))$, $q = O(1/\epsilon^2)$, $k = O(1/\epsilon)$ which satisfies that for all $P \in \Pa$ such that $\eta_P \geq \gamma$, 
there exists a distribution $\Ha_P$ in $P$ of density $2b_P$ in $\D\vert_P$ such that $g^{\rand}$ is $\big(\F^{\rand}, 
1/2-\frac{\epsilon}{2b_P(1-b_P)}\big)$-hard
on $\Ha_P$. 
\end{theorem}

\smallskip
\textbf{Interpretation of IHCL$++$.} Before going into the proof of Theorem~\ref{thm:ihcl++}, we explain what our IHCL$++$ theorem entails and how it is a stronger and more general version of the original IHCL.

\begin{itemize}
    \item In Theorem~\ref{thm:ihcl++}, we remove the $\delta$-weakly hard assumption from the original IHCL theorem, but still obtain that $g$ is strongly hard on some distribution. 
    The caveat is that the lower bound on the density of each hardcore distribution $\Ha_p$ depends on the balance $b_P$ of $g$ on $P$. 
    Namely, if $g$ is an ``easy'' function, then the density of the hardcore sets will be small. 
    However, in our IHCL$++$, we can always guarantee strong hardness of $g$ within each $P \in \Pa$ on $\Ha_P$. 
    
    \item We provide a general lower bound for the density of the hardcore distribution $\Ha_P$ on each $P \in \Pa$ that depends on the expected value of $g$ on $P$ (i.e., on $b_P$). 
    The parameter $b_P$ is an abstraction of the original parameter $\delta$ in IHCL, given that in our Theorem~\ref{thm:ihcl++} we have no assumption whatsoever on the function $g$, and hence we also have no $\delta$ parameter.
    The indistinguishability guarantee within each set $P \in \Pa$ also degrades as $b_P \rightarrow 0$ and as $b_P \rightarrow 1/2$.
    
    \item In our $++$ theorem, the original IHCL occurs both ``locally'' (on each $P \in \Pa$) and ``globally'' (on $\X)$. 
    Theorem~\ref{thm:ihcl++} states that IHCL occurs locally; namely, we obtain a hardcore distribution $\Ha_P$ within each $P \in \Pa$. 
    However, as we show in Section~\ref{sec:ihcl++toihcl}, we can also ``glue'' the different hardcore distributions together $\Ha_P$ in order to obtain a global hardcore distribution $\Ha$.
\end{itemize}

First, we present a proof of our proposed IHCL$++$ as a direct corollary of our theorem characterizing maximal hardness (Lemma~\ref{thm:tfae}). 
Next, we summarize a second proof of IHCL$++$ by adapting the proof of Trevisan et al. \cite{ttv09}.
We include this second proof because it helps in understanding how multicalibration relates to the Regularity Lemma and because we are able to improve the density parameter from $\delta$ to $2\delta$. 

\begin{proof}[Proof of Theorem~{\rm \ref{thm:ihcl++}}]
The proof is a combination of the MC Theorem (Theorem~\ref{thm:mcpartition}) and our lemma characterizing indistinguishability from constant-Bernoulli functions (Lemma~\ref{thm:tfae}).
We first apply the MC Theorem (Theorem~\ref{thm:mcpartition}) to $\F, g, \D$ with the same parameters $\epsilon, \gamma$.
This yields a partition $\Pa \in \F_{t, q, k}$ of $\X$ with $t = O(1/(\epsilon^4 \gamma) \cdot \log(|\X|/\epsilon))$, $q = O(1/\epsilon^2)$, $k = O(1/\epsilon)$
satisfying 
\[
    \Big|\E_{x \sim \D\vert_P} [f(x) \cdot (g(x) - v_P)] \Big| \leq \epsilon
\]
for all $f \in \F$ and for all $P \in \Pa$ such that $\eta_P \geq \gamma$. 
Next, we apply Lemma~\ref{thm:tfae} to each piece $P \in \Pa$ such that $\eta_P \geq \gamma$: Given that $g$ is $(\F, \epsilon)$-indistinguishable from the constant function $v_P$, by statement (\ref{item:dist2b}) in Lemma~\ref{thm:tfae} it follows that there exists a distribution $\Ha_P$ of density $2b_P$ in $\D\vert_P$ such that $g^{\rand}$ is $(\F^{\rand}, 1/2-\epsilon_P)$-hard on it, where $\epsilon_P = \frac{\epsilon}{2b_P(1-b_P)}$.
Hence $\Ha_P$ is a hardcore distribution for $g$ on $P$, as required.
\end{proof}

\smallskip
\textbf{Another proof of IHCL$++$ as a modification from \cite{ttv09}.} 
Another way to prove IHCL$++$ is as a modification of Trevisan et al.'s proof that the Regularity Lemma implies IHCL \cite{ttv09}.
We summarize the approach of this proof here; full details appear in \cite{casacuberta2023thesis}.
The Trevisan et al. proof of IHCL begins by invoking the Regularity Lemma / Multiaccuracy Theorem (Theorem~\ref{thm:intro-regularity}) on $g$ and where $\D$ corresponds to the uniform distribution; we begin by invoking the MC Theorem instead (Theorem~\ref{thm:mcpartition}).\footnote{Trevisan et al. consider the restricted version of IHCL where $\D$ corresponds to the uniform distribution on $\X$, whereas we consider the general version of IHCL with an arbitrary distribution.}
In \cite{ttv09}, they use the resulting MA predictor $h$ to define the following distribution $\Ha$ over the domain $\X$:
\[
    \Ha(x) := \dfrac{|g(x)-h(x)|}{\sum_{y \in \X} |g(y) - h(y)|}.
\]
The intuition behind this choice of distribution is to put more mass where $h$ and $g$ disagree. Trevisan et al.\ then show that 1) $\Ha$ is $\delta$-dense in $\D$, and that 2) $g$ is strongly hard on $\Ha$. Inspired by their proof, we define the following probability distribution on each (large enough) set $P \in \Pa$:
\[
    \Ha_P(x) := \dfrac{|g(x) - v_p|}{\sum_{y \in P} |g(y) - v_p|}.
\]
We remark that, unlike in the multiaccuracy case of \cite{ttv09}, the denominator in the expression for $\Ha_P$ sums over the set $P$ instead of over the entirety of the domain $\X$. 

We then can show that 1) $\Ha_P$ is $2b_P$-dense in $\D\vert_P$, and that 2) $g$ is strongly hard on $\Ha_P$.
To show 1), we analyze the quantity $\sum_{x \in P} |g(x)-v_P|$; the fact that $v_P$ is a constant is what allows us to recover the optimal $2b_P$ density parameter, whereas the same analysis carried out in \cite{ttv09} for IHLC does not.
That is, a multiaccurate predictor does not seem to imply IHCL with optimal density parameters, but a multicalibrated predictor can.
To show 2), we relate the probability that $f^{\rand}(x) = g^{\rand}(x)$ to the expected value in the definition of multicalibration using a similar expression to the one used in the proof of Yao's Lemma (i.e., Equation~\ref{eq:yaoid}).

\smallskip
\textbf{Another proof of IHCL$++$ using the original IHCL.} A third approach to proving our IHCL$++$ theorem is by using the original IHCL coupled with our theorem characterizing maximal hardness (Theorem~\ref{thm:tfae}).
Namely, we begin by applying the MC theorem to obtain a partition $\Pa$, thus satisfying the assumption of Lemma~\ref{thm:tfae} on each $P \in \Pa$ such that $\eta_P \geq \gamma$.
By statement (\ref{item:hard}) in Lemma~\ref{thm:tfae}, it follows that the assumption for IHCL is satisfied on each such piece with weak hardness $\delta := b_P - 2\epsilon$, and hence we can apply IHCL to each such $P$ to obtain a hardcore distribution within each piece of the partition, which corresponds to the conclusion in IHCL$++$.
However, this yields worse indistinguishability parameters for the hardcore distribution than those in Theorem~\ref{thm:ihcl++}.

\subsection{Recovering IHCL from IHCL++}\label{sec:ihcl++toihcl}

Having proved IHCL$++$, we now show how to recover the original IHCL theorem from it.  
The key idea is to ``glue together'' the hardcore distributions $\Ha_P$ within each $P \in \Pa$, where in this gluing 
each $P \in \Pa$ is weighted according to its size parameter $\eta_P$ of the set $P$. 
When we bring back the assumption that $g$ is $\delta$-weakly hard (as in the original IHCL statement), by using the fact that the multicalibrated partition is of low-complexity, it follows that the glued hardcore distribution $\Ha$ has density at least $2\delta$ in $\D$, which corresponds to the optimal density parameter in IHCL.

We begin by showing that if $g$ is $\delta$-hard, then the $g$ cannot be too imbalanced on average over the pieces of the partition.

%\begin{tcolorbox}
\begin{proposition}\label{prop:kpdelta}
Let $\X, \D, \F, g, \epsilon, \gamma, \Pa, t, q, k$ as in Theorem~{\rm\ref{thm:ihcl++}}. 
Moreover, assume that $g$ is $(\F_{t+k, q}, \delta)$-hard, and suppose that $\eta_P \geq \gamma$ for all $P \in \Pa$. 
Then,
\[
    \E_{P \sim \Pa(\D)}[b_P] \geq \delta.
\]
\end{proposition}
%\end{tcolorbox}

\begin{proof}
We will argue by contradiction; hence assume that $\E_{P \sim \Pa(\D)}[b_P] < \delta$. 
We show that this contradicts the fact that $g$ is $\delta$-hard on $\D$. 
More specifically, we show that we can construct a function $f \in \F_{t+k, q}$ such that
\[
    \Pr_{\substack{x \sim \D, \, \coins(g^{\rand}), \\ \coins(f^{\rand})}}[f^{\rand}(x) = g^{\rand}(x)] > 1 - \delta.
\]
Let $\hat{f} \in \F_{t, q}$, where $t = O(1/(\epsilon^4 \gamma) \cdot \log(|\X|/\epsilon))$, $q = O(1/\epsilon^2)$,
be the partition membership function for $\Pa$ as given by Definition~\ref{def:Ftqk}. 
That is, $\Pa = \{\hat{f}^{-1}(1), \ldots, \hat{f}^{-1}(k)\}$.
We define our $f$ as $f=f_{\post} \circ \hat{f}$, where $f_{\post}: [k] \rightarrow \{0,1\}$ is the indicator function $f_{\post}(i) = \mathbbm{1}[v_{\hat{f}^{-1}(i)} \geq 1/2]$.
(Recall that $v_P$ is the expected value of $g$ on $\D\vert_P$.)
Thus, $f_{\post}$ can be computed by a circuit of size $k$ (see \cite[\S 9.1.1.]{boazbook2}), so $f \in \F_{t+k, q}$.

The intuitive meaning of the indicator function $f_{\post}(i)$ is the following: we want to show that $f$ approximates $g$ ``quite well'', in the sense that $\Pr[f^{\rand}(x)=g^{\rand}(x)] > 1-\delta$. 
The above construction is saying that $f$ is equal to 0 in all of the $P \in \Pa$ such that $\E_{\D\vert_{P}}[g(x)] \leq 1/2$, and equal to 1 in all of the $P \in \Pa$ such that $\E_{\D\vert_{P}}[g(x)] > 1/2$. 
We now show that this is indeed a good approximation of $g$; good enough that it contradicts the assumption that $g$ is $(\F_{t+k, q}, \delta)$-hard. 

Fix some $P \in \Pa$, and as usual let $v_P = \E_{x \sim \D\vert_P}[g(x)]$.
Since $g$ is $\{0,1\}$-valued and $f$ equals the majority value of $g$ on $P$ by construction, it follows that
\begin{align*}
    \Pr_{\substack{x \sim \D\vert_P, \, \coins(g^{\rand}), \\ \coins(f^{\rand})}}[f^{\rand}(x) = g^{\rand}(x)] 
    &= \Pr_{\substack{x \sim \D\vert_P, \\ \coins(g^{\rand})}}[f(x) = g^{\rand}(x)] \\
    &= 1 - \min\{v_P, 1-v_P\} = 1 - b_P,
\end{align*}
since $f = \mathbf{0}$ when $v_P < 1/2$ and $f = \mathbf{1}$ when $v_P \geq 1/2$. 

Because this expression holds for every $P \in \Pa$, when we consider the probability that $f(x)=g(x)$ over $\X$ it follows that
\[
    \Pr_{\substack{x \sim \D, \\ \coins(g^{\rand})}}[f(x)=g^{\rand}(x)] = 1 - \E_{P \sim \Pa(\D)}[b_P].
\]
Since by assumption $\E_{P \sim \Pa(\D)}[b_P]<\delta$, it follows that
\[
    \Pr_{\substack{x \sim \D, \\ \coins(g^{\rand})}}[f(x) = g^{\rand}(x)] > 1 - \delta,
\]
which contradicts the $(\F_{t+k,q}, \delta)$-hardness of $g^{\rand}$.
\end{proof}

In Proposition~\ref{prop:kpdelta}, we are assuming that $\eta_P \geq \gamma$ for all $P \in \Pa$ in order to make its proof cleaner. 
 However, we should only be ``gluing'' together the pieces $P \in \Pa$ that have enough size and enough mass; i.e., such that $\eta_P$ and $b_P$ are larger than some threshold. 
In the case of the size parameter $\eta_P$, its threshold corresponds to the $\gamma$ parameter in the approximate MC definition. 
In the case of the balance parameter $b_P$, we introduce a new parameter $\tau$:

\begin{definition}\label{def:goodP}
{\rm
Let $\gamma, \tau > 0$, and let $\Pa$ be a partition of the domain $\X$ and $\D$ a distribution on $\X$. 
We say that a set $P \in \Pa$ is \emph{$(\gamma, \tau)$-good} with respect to $\D$ if $\eta_P \geq \gamma$ and $b_P \geq \tau$. 
We use the indicator random variable $\mathbbm{1}_{G}$ to denote whether $P$ is good. 
Namely, $\mathbbm{1}_{G}(P)$ for $P \in \Pa$ returns 1 only if $\eta_P \geq \gamma$ and $b_P \geq \tau$; otherwise, it returns 0. 
(The letter $G$ stands for ``good''.)}
\end{definition}
%\end{tcolorbox}

Given this definition, we now re-evaluate Proposition~\ref{prop:kpdelta}. Namely, the next fact follows directly from coupling the proof of Proposition~\ref{prop:kpdelta} and Definition~\ref{def:goodP}:

\begin{corollary}\label{cor:kpdeltagood}
Let $\X, \D, \F, g, \epsilon, \gamma, \Pa, t, q, k$ as in Theorem~{\rm \ref{thm:ihcl++}}, and let $\tau > 0$.
Moreover, assume that $g$ is $(\F_{t+k, q}, \delta)$-weakly hard for some $\delta>0$. 
Then,
\[
    \E_{P \sim \Pa(\D)}[b_P \cdot \mathbbm{1}_{G}(P)] \geq \delta - \gamma k - \tau = \delta - O(\gamma/\epsilon) - \tau,
\]
where $\mathbbm{1}_{G}(P)$ returns $1$ only if $\eta_P \geq \gamma$ and $b_P \geq \tau$. 
\end{corollary}

We can now prove the original IHCL from IHCL$++$:

\begin{proof}[Proof of IHCL using IHCL$++$]

Let $\F, \X, \epsilon, \delta$ be the assumption parameters in IHCL. 
We define the parameters $\epsilon' := \epsilon^2 \delta$, $\gamma := \epsilon \epsilon'$, $\tau := \epsilon \delta$, and invoke the IHCL$++$ Theorem with these parameters $\epsilon', \gamma$. 
By IHCL$++$ (Theorem~\ref{thm:ihcl++}), we obtain a partition $\Pa \in \F_{t,q,k}$ of $\X$ with 
$t = O(1/(\epsilon'^4 \gamma) \cdot \log(|\X|/\epsilon'))$, $q = O(1/\epsilon'^2)$, $k = O(1/\epsilon')$ such that, for each $P \in \Pa$ where $\eta_p \geq \gamma = \epsilon \epsilon'$,
there exists a distribution $\Ha_P$ in $P$ of density $2 b_P$ in $\D\vert_P$ such that $g$ is $(\F, \epsilon/b_P)$-hard on $\Ha_P$. 

Let $\tau := \epsilon\delta$. By Corollary~\ref{cor:kpdeltagood}, when we only consider the $P \in \Pa$ that are $(\gamma, \tau)$-good, we obtain that
\[
    \E_{P \sim \Pa(\D)}[b_P \cdot \mathbbm{1}_G(P)] \geq \delta - O(\gamma/\epsilon') - \tau.
\]
By plugging in the definition of each value $\epsilon' = \epsilon^2 \delta$, $\gamma = \epsilon \epsilon'$, and $\tau = \epsilon \delta$, the expression $\delta - O(\gamma/\epsilon') - \tau$ simplifies to $\delta \cdot (1 - O(\epsilon))$. 

We now construct a hardcore distribution $\Ha$ on $\X$ as follows: we define $\Ha$ by ``gluing up'' the distributions $\Ha_P$ such that the corresponding $P$ is $(\gamma, \tau)$-\emph{good}. 
Formally, let $\D(\Pa)\vert_G$ denote the distribution $\D(\Pa)$ restricted to the set $\{\cup_{P \in \Pa} \, x \in P \mid \mathbbm{1}_G(P)=1\} \subseteq \X$.
Then, for each $x \in \X$, we let $\Ha(x) = \Ha_P(x)$, where $P \sim \D(\Pa)\vert_G$.

We now analyze (1) the density of $\Ha$, and (2) the hardness of $g$ on $\Ha$. 
Since each $\Ha_P$ such that $P$ is \textit{good} has density $2 b_P$ in $\D\vert_P$ and $\E_{P \sim \Pa(\D)} [b_P \cdot \mathbbm{1}_{G}(P)] \geq \delta \cdot (1 - O(\epsilon))$ by Proposition~\ref{prop:kpdelta}, it follows that $\Ha$ has density $2\delta \cdot (1 - O(\epsilon))$ in $\D$. 

For the hardness of $g$, we see that
\begin{align*}
    \Pr_{\substack{x \sim \Ha, \, \coins(f^{\rand}), \\ \coins(g^{\rand})}}[f^{\rand}(x) = g^{\rand}(x)] & = 
    \E_{P \sim \D(\Pa)\vert_G}\Bigg[\Pr_{\substack{x \sim \Ha_P, \, \coins(f^{\rand}), \\ \coins(g^{\rand})}}[f^{\rand}(x)=g^{\rand}(x)\big] \Bigg] \\
    &\leq \E_{P \sim \D(\Pa)\vert_G} \Big[ \Big(1/2 + \dfrac{\epsilon'}{2b_P(1-b_P)} \Big)\Big] =
\\[0.2cm]
&    = \dfrac{1}{2} +\E_{P \sim \D(\Pa)\vert_G}\Bigg[ \dfrac{\epsilon'}{2\tau(1-\tau)} \Bigg] \leq \dfrac{1}{2} + \frac{\epsilon'}{\tau}.
\end{align*}
By plugging in the definitions of the parameters, namely $\epsilon' = \epsilon^2 \delta$ and $\tau = \epsilon \delta$, we obtain that $\epsilon'/\tau = \epsilon$. Hence, we obtain that
$g$ is $(\F, 1/2-\epsilon)$-hard on $\Ha$.

Therefore, we have shown that $\Ha$ is a distribution of density $2\delta \cdot (1 - O(\epsilon))$ in $\D$ such that $g$ is strongly hard on $\Ha$. 
In order to recover the original IHCL statement, we observe that we can modify $\Ha$ in order to make its density at least $2\delta$ in $\D$ while maintaining the strong hardness of $g$ on $\Ha$. 
Specifically, we can modify all of the probability masses in $\Ha$ by a factor of $(1 \pm O(\epsilon))$ to make it $2\delta$-dense in $\D$.
Since this only changes $\Ha$ by a total variation distance of $O(\epsilon)$, it follows that $g$ is still $(\F, 1/2-O(\epsilon))$-hard.
\end{proof}

\subsection{Set version of IHCL$++$}\label{sec:setsmeasures}

When working with Impagliazzo's Hardcore Lemma, there are two possible ways of describing the hardcore set of inputs \cite{imp95, hol05}.
One is to find a hardcore \emph{distribution}, which is defined over the domain $\X$, and the other is to find a hardcore \emph{set}.
Historically (and in this paper), proofs involving the Hardcore Lemma first find a hardcore distribution, but in applications it is often more intuitive to deal with sets instead of distributions \cite{imp95, ks03, hol05, ttv09}.
In the case of sets, the natural definition for the density of a set is the following:

\begin{definition}[$\delta$-dense set]\label{def:denseset}
{\rm
Given a set $S \subseteq \X$, we say that $S$ is \emph{$\delta$-dense} in $\X$ if $|S| \geq \delta \cdot |\X|$. 
}
\end{definition}

In other words, $S$ is $\delta$-dense in $\X$ if $S$ occupies at least a fraction $\delta$ of the domain $\X$. 
As it was originally shown by Impagliazzo \cite{imp95} and further detailed by Kilvans and Servedio \cite{ks03}, we can obtain a hardcore set from a hardcore distribution via a probabilistic method argument.
To obtain a hardcore set $H$ from a hardcore distribution $\Ha$, we construct $H$ probabilistically as follows: for each $x \in \X$, we add $x$ to $H$ with probability $\Ha(x)/\max_y \Ha(y)$.
Then, using Hoeffding's inequality one can show that the fact that $\Ha$ is $\delta$-dense 
with respect to the uniform distribution over $\X$ implies that $H$ is $\delta$-dense in $\X$ on expectation (i.e., that $|H| \leq \delta |\X|$). 
We defer the proof of this conversion to Appendix~\ref{sec:setihcl}. 

Given this conversion, we can formally state the set version of our IHCL$++$; the only two relevant differences compared to the distribution-version of IHCL$++$ are that Theorem~\ref{thm:ihcl++set} is restricted to the special case where $\D$ corresponds to the uniform distribution over $\X$ and that it requires an upper-bound on the size of $\F$:

\begin{theorem}[IHCL$++$, Set version]\label{thm:ihcl++set}

Let $\X$ be a finite domain, let $\F$ be a family of functions $f\colon \X \rightarrow [0,1]$, let $g\colon \X \rightarrow [0,1]$ be an arbitrary function, let $\epsilon, \gamma > 0$. There exists a partition $\Pa \in \F_{t, q, k}$ of $\X$ with $t = O(1/(\epsilon^4 \gamma) \cdot \log(|\X|/\epsilon))$, $q = O(1/\epsilon^2)$, $k = O(1/\epsilon)$ satisfying that for all $P \in \Pa$ such that $\eta_P \geq \gamma$
and $|\F| \leq \exp(|\X|(\epsilon/4)^2 b_P^2)$,
there exists a set $H_P \subseteq P$ of density $|H_P|/|P| \geq 2b_P$ such that $g^{\rand}$ is $\big(\F^{\rand}, 
1/2-\frac{\epsilon}{2b_P(1-b_P)}\big)$-hard on the uniform distribution over~$H_P$.
\end{theorem}

\subsection{Previous instances of partitioning in the Regularity Lemma literature}\label{sec:rttv08}

In the original paper on the Hardcore Lemma, Impagliazzo gave two different proofs for it: a boosting proof and a proof based on the min-max theorem \cite{imp95}.
Reingold, Trevisan, Tulsiani, and Vadhan were the first to give a partition-based proof of IHCL \cite{rttv08}.
Using this partitioning-based approach, Reingold et al. proved a more general version of IHCL \cite[Thm. 3.2]{rttv08}, which turns out to be closely related to our IHCL$++$ theorem. 
In this section, we summarize the similarities and differences between Theorem 3.2 in \cite{rttv08} and our IHCL$++$.
In doing so, we take inspiration from our IHCL$++$ and are able to improve Reingold et al.'s result by removing the assumption that $g$ is $\delta$-weakly hard.

Both Reingold et al.'s Theorem 3.2 and our IHCL$++$ find a partition $\Pa$ of the domain such that there exists a hardcore distribution within each piece $P \in \Pa$.
In the case of \cite{rttv08}, they establish statement~\ref{item:0s1s} in our lemma characterizing indistinguishability from constant-Bernoulli functions (Lemma~\ref{thm:tfae}).
Reingold et al. also show that their general IHCL theorem recovers a form of the original IHCL theorem with optimal density $2\delta$, which we do with our IHCL$++$ in Section~\ref{sec:ihcl++toihcl}.

There are two main differences between Reingold et al.'s general version of IHCL and our IHCL$++$ (Theorem~\ref{thm:ihcl}): first, in \cite{rttv08}, they construct a partition of \emph{exponential complexity} $\exp(\poly(1/\epsilon, 1/\delta))$, whereas our partition is of polynomial complexity $t, q, k = \poly(1/\epsilon, 1/\delta)$, coming from the fact that the MC Theorem has polynomial complexity (Theorem~\ref{thm:mcpartition}). 
Second, in \cite{rttv08} they require the input function $g$ to be $\delta$-hard, whereas our IHCL$++$ holds for an arbitrary function $g$.
The first difference is the fundamental one: Before the development of the multicalibration tools from the field of algorithmic fairness, it seemed that exponential complexity was inherent in partition-based proofs of complexity-theoretic theorems.
Indeed, the proof of Theorem 3.2 in \cite{rttv08} essentially constructs a multicalibrated partition of exponential complexity.

The second difference is less fundamental; we show that that we can actually remove the assumption that $g$ is $\delta$-hard from their proof by following the strategy in our $++$ theorems.
Namely, we replace the $\delta$ parameter with a threshold parameter $\tau$.
We defer the proof of this claim to Appendix~\ref{sec:rttvproof}.

\section{Characterizations of Pseudoentropy}\label{sec:pame}

In computational complexity and cryptography, a key development has been the study of \textit{computational analogues} of concepts from information theory. 
For example, the notion of \textit{computational indistinguishability}, which has become one of the most fundamental notions in theoretical computer science \cite{gm82, yao82}, can be thought of as the computational analogue of \textit{statistical distance}. Formally, this corresponds to the definition of indistinguishable distributions as stated in Definition~\ref{def:indistdist} in Section~\ref{sec:hardnessindist}.
There is an important difference in the definition of computational indistinguishability depending on the choice of $\F$: In this section, we only consider the classes $\F$ that correspond to boolean circuits, and hence our results correspond to the non-uniform setting. 
However, we remark that the results can be extended the the uniform setting; in particular, the work of Vadhan and Zheng on which this chapter is based on \cite{zhe14, vz13}, state the theorems and definitions that we use in both the uniform and non-uniform settings. 

Computational analogues of entropy were subsequently introduced by Yao \cite{yao82} and H\aa stad, Impagliazzo, Levin, and Luby \cite{hill99}, the latter being known as \textit{pseudoentropy}. 
The notion of pseudoentropy allowed H\aa stad et al.\ to prove the fundamental result that establishes the equivalence between pseudorandom generators and one-way functions \cite{vz12}.
Later, Vadhan and Zheng showed that the notion of pseudoentropy is equivalent to hardness of sampling \cite{vz12}. 
In his PhD thesis, Zheng proved a similar theorem but for average-case variants of the H\aa stad et al.\ instead, known as \textit{pseudo-average min-entropy}, which we will refer to as \textit{PAME} \cite{zhe14}. 

\subsection{Definitions}

Throughout Section~\ref{sec:pame}, we use the notation $X$ to denote a random variable instead of writing $x \sim \D$.

\begin{definition}[Min-entropy]\label{def:minentropy}
{\rm
The \emph{min-entropy} of a random variable $X$ is defined as
\[
    \Hb_{\infty}(X) := \min_x \Big\{ \log \Big( \dfrac{1}{\Pr[X=x]} \Big) \Big\}.
\]
}
\end{definition}

The computational analogue of min-entropy is defined as follows:

\begin{definition}[Pseudo-min-entropy \cite{hill99}]\label{def:pseudominen}
{\rm
Let $\F$ be a family of functions $f: \X \rightarrow [0,1]$.
A distribution $X$ on $\X$ has $(\F, \epsilon)$-\emph{pseudo-min-entropy at least} $k$ if there exists a distribution $Y$ such that
\begin{enumerate}
    \item[{\rm 1.}] $\Hb_{\infty}(Y) \geq k$, where $\Hb_{\infty}(\cdot)$ denotes min-entropy. Equivalently, $\Pr[Y=x] \leq 2^{-k}$ for all $x$.
    \item[{\rm 2.}] $X$ is $(\F, \epsilon)$-indistinguishable from $Y$, in the sense of Definition~\ref{def:indistdist}.
\end{enumerate}
}
\end{definition}

Hsiao, Lu, and Reyzin considered the conditional version of pseudo-min-entropy known as \emph{pseudo-average-min-entropy}, which we will refer to as \emph{PAME}.
For that, we need to first define \textit{average min-entropy}:

\begin{definition}[Average min-entropy \cite{dors04}]\label{def:AME}
{\rm
For every joint distribution $(X, B)$, the \emph{average min-entropy of} $B$ \emph{given} $X$ is defined as
\[
    \tilde{\Hb}_{\infty}(B|X) = \log \left( \dfrac{1}{\E_{x \sim X}[1 / 2^{H_{\infty}(B|_{X=x})}]} \right) = \log \left( \dfrac{1}{\E_{x \sim X}[\max_a \Pr[B=a|X=x]]} \right). 
\]  
}
\end{definition}

The computational analogue of average min-entropy, which we formalize in the non-uniform setting (i.e., where $\F$ is a class of circuits), is then defined as follows:

\begin{definition}[Pseudo-average min-entropy (PAME) \cite{zhe14}]\label{def:pamenonunif}
{\rm
Let $(X, B)$ be a joint distribution on $\X \times [L]$, let $\F$ be a family of functions $f: \X \times [L] \rightarrow [0,1]$, and let $\epsilon > 0$. 
We say that $B$ has $(\F, \epsilon)$\emph{-pseudo-average min-entropy at least $k$ given $X$} if there exists a random variable $C$ jointly distributed with $X$ such that the following holds:
\begin{enumerate}
    \item[{\rm 1.}] $\tilde{\Hb}_{\infty}(C|X) \geq k$.
    \item[{\rm 2.}] $(X, B)$ and $(X, C)$ are $(\F, \epsilon)$-indistinguishable.
\end{enumerate}
}
\end{definition}

\subsection{The PAME Theorem}

Given $L \in \mathbb{N}$, we let $[L] = \{1, 2, \ldots, L\}$.
The following is a known fact about (conditional) min-entropy:

\begin{proposition}[{\cite[Proposition 4.10]{dors04}}]\label{prop:pamepredict}

For every joint distribution $(X, B)$ on $\X \times [L]$,
\[
    \tilde{\Hb}_{\infty}(B|X) \geq k \iff 
    \Pr[f(X) = B] \leq 2^{-k} \quad \forall f\colon \{0,1\}^n \rightarrow [L].
\]
\end{proposition}

In other words, Proposition~\ref{prop:pamepredict} characterizes the notion of (average) min-entropy in terms of unpredictability: If $B|X$ has high average min-entropy, then $B$ is hard to predict from $X$. 
Moreover, Proposition~\ref{prop:pamepredict} establishes an equivalence between the two notions, given that it is an if and only if statement. 
That is, we can characterize the notion of (average) min-entropy through the notion of unpredictability and viceversa. 

Proposition~\ref{prop:pamepredict} considers an information-theoretic notion of hardness.
The PAME theorem corresponds to the computational analogue of Proposition~\ref{prop:pamepredict}, for which we need to define the notion of hardness for \emph{distributions}:

\begin{definition}[Hardness of a distribution {\cite[Def. 4.13]{zhe14}}]
    Given a joint distribution $(X, B)$ on $\X \times [L]$, a class $\F$ of randomized functions $f: \X \rightarrow [L]$, and $\delta>0$, we say that $B$ is $(\F, \delta)$\emph{-hard to predict given $X$} if, for all $f \in \F$,
    \[
        \Pr_{\coins(f)}[f(X) = B] \leq 1-\delta.
    \]
\end{definition}

One of the main theorems shown in Vadhan and Zheng \cite{vz12} is the computational analogue of Proposition~\ref{prop:pamepredict}.\footnote{The notation used in \cite{zhe14} corresponds to $1-2^{-r}$ in lieu of our $\delta$ and, accordingly, $r$ in lieu of $\log(1/(1-\delta))$.}

\begin{theorem}[PAME theorem {\cite[Theorem 4.15]{zhe14}}]\label{thm:pamenonunif}

Let $(X, B)$ be a joint distribution on $\X \times [L]$, $\F$ any class of functions $f\colon \X \times [L]\rightarrow [0,1]$, and let $\delta, \epsilon>0$. 
Then, the following two statements hold:
\begin{enumerate}
    \item[{\rm 1.}] If $B$ is $(\F^{\rand}_{t, q},  \delta)$-hard to predict given $X$ for $t = q = L \cdot \textrm{\rm poly} (n,\log L,1/\epsilon)$, then $B$ has $(\F, \epsilon)$-PAME at least $\log(1/(1-\delta))$ given $X$.
    \item[{\rm 2.}] If $B$ has $(\F_{t, q}, \epsilon)$-PAME at least $\log(1/(1-\delta))$ given $X$ for $t = q = O(1)$, then $B$ is non-uniformly $(\F^{\rand}, \delta - \epsilon)$-hard to predict. 
\end{enumerate}
\end{theorem}

Observe that the parameters of the class $\F_{t, q}$ in item 1 of Theorem~\ref{thm:pamenonunif} imply that we need $L = O(n)$ in order to maintain polynomial complexity in $n$.
We remark that in \cite{vz12, zhe14} a similar theorem is shown for Shannon entropy instead of min-entropy and finding an equivalence to hardness of sampling (which is quantified using the Kullback-Leibler divergence) instead of unpredictability. 

Item 1 is the more interesting and difficult direction of Theorem~\ref{thm:pamenonunif}, given that Item 2 is rather immediate from the definitions.
Hence, our PAME$++$ Theorem will seek to generalize item Item 1 from Theorem~\ref{thm:pamenonunif}.

\subsection{Multiclass multicalibration}

For the PAME$++$ theorem we require a more general multicalibration theorem; namely, one that is adapted to the multiclass case. 
This corresponds to the case where the function $g$ maps $\X$ to $[L]$, where $L \in \mathbb{N}$, and we think of each of the integers $y \in L$ as a ``class''.
Equivalently, by setting $L = 2^{\ell}$, the multiclass case corresponds to the case where $g$ maps $\X$ to bit-strings $\{0,1\}^{\ell}$ instead of to real values.

\begin{definition}[Multiclass approximate multicalibration, adapted from {\cite[Def. 5.3]{gkrsw21}}]
    Let $\X$ be a finite domain, 
    $g^{\rand}: \X \rightarrow [L]$ a randomized function, 
    $\F$ a class of functions $f: \X \rightarrow [0,1]$, $\D$ be a probability distribution over $\X \times [L]$, where $L \in \mathbb{N}, L \geq 2$, and let $\epsilon, \gamma >0$. A partition $\Pa$ of $\X$ is \emph{$(\F, \epsilon, \gamma)$-approximately multicalibrated} for $g^{\rand}$ on $\D$ if for all $f \in \F$, for all $y \in [L]$, and for all $P \in \Pa$ such that $\Pr_{x \sim \D}[x \in P] \geq \gamma$,
    \[
        \Big|\E_{x \sim \D\vert_P} [f(x) \cdot (\mathbbm{1}[\grand(x)=y] - v_{P_y})] \Big| \leq \epsilon,
    \]
    where $v_{P_y} = \E_{\substack{x \sim \D\vert_P, \\ \coins(g^{\rand})}}[\mathbbm{1}[g^{\rand}(x)=y]]$.
\end{definition}

Intuitively, we can think of the term $\mathbbm{1}[g(x)=y]$ as the ``booleanization'' of the function $g$ for the corresponding label $y \in [L]$.
We adapt the Gopalan et al.'s definition in two ways: first, we use the approximate MC relaxation instead of the MC on average approach (see Section~\ref{sec:mcpartitions} for the definition and relationship between the two definitions via the reparametrization $\epsilon \leftarrow \epsilon \gamma$).
Second, we use the usual multicalibration expression, whereas Gopalan et al. use a recasting of multicalibration that uses the covariance \cite{gkrsw21}.
The justification for why we can make this change is given in \cite[\S C.2]{gkrsw21}.

Gopalan et al. give an algorithm for constructing a partition that satisfies the multiclass multicalibration definition:

\begin{theorem}[Multiclass approximate MC is realizable {\cite[Thm. 9.2]{gkrsw21}}]\label{thm:mcmc}
    Let $\X$ be a finite domain, $\F$ a class of functions $f: \X \rightarrow [0,1]$, $g^{\rand}: \X \rightarrow [L]$ an arbitrary randomized function for $L \in \mathbb{N}, L \geq 2$, $\D$ a probability distribution over $\X$, and $\epsilon, \gamma > 0$. There exists an $(\F, \epsilon, \gamma)$-approximately multicalibrated partition $\Pa$ of $\X$ for $g^{\rand}$ on $\D$ such that $P \in \F_{t, q, k}$, where
    \begin{enumerate}
        \item $t = O\big( L \cdot (1/(\epsilon^4 \gamma))^{L} \cdot \log(|\X|/\epsilon) \big)$,
        \item $q = O\big((L/\epsilon^2)^{O(L)}\big)$,
        \item $k = O\big( L \cdot (1/\epsilon^4)^{\L-1} \big)$.
    \end{enumerate}
\end{theorem}

As we pointed out in the introduction, the complexity parameters in the Multiclass MC Theorem have an exponential dependence on the number of classes $L$, which is not the case in the real-valued setting of multicalibration (Theorem~\ref{thm:mcpartition}).
Moreover, in the boolean case (i.e., when $L=2$), the parameters of the Multiclass MC Theorem do not match those stated in the real-valued case (Theorem~\ref{thm:mcpartition}).
This is because these two theorems use different algorithms: Theorem~\ref{thm:mcpartition} uses the parameters from H\'ebert-Johnson et al. MC algorithm \cite{hkrr18}, whereas the Multiclass Theorem uses the parameters from the branching program algorithm from Gopalan et al. \cite{gkrsw21}.

\subsection{The PAME$++$ Theorem}

For the multiclass version of the PAME$++$ theorem, we need to generalize the definition of $1-b_P$ to $\ell$-bit functions:

\begin{definition}\label{def:mp}
    Given a joint distribution $(X, B)$ defined over the domain $\X \times [L]$ and $P \subseteq \X$, we let $m_P = \max_z \Pr[B=z \mid x\in P] = \max_z \Pr[B_P=z]$.
\end{definition}

Thus, $m_P \in [0,1]$ for $L$.
Moreover, in the case where $L=2$, then $b_P = 1 - m_P$.
We can now state our proposed PAME$++$ theorem; we use $X_P$ to indicate that $X$ is conditioned on $x \in P$.

\begin{theorem}[PAME$++$]\label{thm:pame++}
Let $\F$ be any class of functions $f: \X  \times [L] \rightarrow [0,1]$ closed under negation and such that $\mathbf{0}, \mathbf{1} \in \F$, let $(X, B)$ be a joint distribution on $\{0, 1\}^n \times [L]$, and let $\epsilon, \gamma > 0$. 
There exists a partition $\Pa \in \F_{t, q, k}$ of $\X$ with $t = O\big( L \cdot (1/(\epsilon^4 \gamma))^{L} \cdot \log(|\X|/\epsilon) \big)$, $q = O\big((L/\epsilon^2)^{O(L)}\big)$, $k = O\big( L \cdot (1/\epsilon^4)^{L-1} \big)$ which satisfies that for all $P \in \Pa$ such that $\eta_P \geq \gamma$,
$B_P$ has $(\F, \epsilon)$-PAME at least $\log(1/m_P) = \Hb_{\infty}(B_P)$
given $X_P$.
\end{theorem}

Observe that in the PAME$++$ case, the complexity parameters of the partition imply that we need $L = O(\log n)$.
The change from $O(n)$ to $O(\log n)$ is due to the complexity parameters obtained in the Multiclass MC Theorem (Theorem~\ref{thm:mcmc}).

\begin{remark}
    In the case where $L=2$, we can obtain better complexity parameters by applying the regular MC Theorem (Theorem~\ref{thm:mcpartition}) instead of the Multiclass MC Theorem (Theorem~\ref{thm:mcmc}).
    In particular, in the boolean case where $L=2$, we obtain a partition $\Pa \in \F_{t, q, k}$ of $\X$ with $t = O(1/(\epsilon^4 \gamma) \log (|\X|/\epsilon))$, $q = O(1/\epsilon^2)$, $k = O(1/\epsilon)$. 
\end{remark}

\smallskip
\textbf{Interpretation of PAME$++$.} 
Before going into the proof of Theorem~\ref{thm:pame++}, we explain what our PAME$++$ theorem entails and how it relates to the original PAME theorem.

\begin{itemize}
    \item In Theorem~\ref{thm:pame++}, we remove the assumption that the function $g$ (where $B=g(X)$) is $\delta$-hard to predict, but still can obtain a partition of the domain such that $B$ has high PAME within each of the pieces of the partition.
    The caveat is that the lower-bound on the amount of PAME that $B$ has within each of the pieces $P$ depends on the parameter $m_P$.
    
    \item The parameter $m_P$ thus provides an abstraction of the original parameter $\delta$ in Theorem~\ref{thm:pamenonunif}, given that the PAME$++$ theorem holds for an arbitrary function $g$.
    Thus, the more balanced the function $g$ is on a piece, the better our bounds are (as it was the case of IHCL$++$).
    
    \item In our $++$ theorem, the original PAME occurs both ``locally'' (on each $P \in \Pa$) and ``globally'' (on $\X)$. 
    Moreover, as we show in Section~\ref{sec:pamefrompame++}, if we define a distribution $B$ as the union of the distributions $B_P$ for each $P \in \Pa$, then we recover the original PAME theorem, although with the restriction that $L = O(\log n)$ instead of the original $O(n)$.
\end{itemize}

We provide two separate proofs; one for the boolean case where $L=2$ and one for the general case $L>2$. 
The latter is a generalization of the former, but we believe that the proof of the boolean case helps provide more intuition.
Moreover, for the case $L=2$ we apply the regular MC Theorem (Theorem~\ref{thm:mcpartition}), whereas in the case where $L>2$ we need to apply the Multiclass MC Theorem.

\begin{proof}[Proof of Theorem~\ref{thm:pame++} for $L=2$]

Without loss of generality, we let $B = g^{\rand}(X)$, where $g^{\rand}: \X \rightarrow [L]$ corresponds to the randomized function satisfying $\Pr[g(x)=b] = \Pr[B=b \mid X = x]$ for all $(x, b) \in \X \times [L]$.

Given the input class $\F$ of functions $f: \X \times \{0,1\} \rightarrow [0,1]$, we define a class $\F'$ of functions $f': \X \rightarrow [0,1]$ as follows: For each $f \in \F$, we add the function $f'(x) = f(x, 1)-f(x, 0)$ and its negation to $\F'$.

We apply the MC partition theorem (Theorem~\ref{thm:mcpartition}) to $\F', g^{\rand}, \D$ and parameters $\gamma, \epsilon$. 
This yields a partition $\Pa \in \F'_{t, q, k}$ of $\X$ with $t = O(1/(\epsilon^4 \gamma) \cdot \log(|\X|/\epsilon))$, $q = O(1/\epsilon^2)$, $k = O(1/\epsilon)$
satisfying 
\[
    \Big|\E_{\substack{x \sim \D\vert_P, \\ \coins(g^{\rand})}} [f'(x) \cdot (g^{\rand}(x) - v_P)] \Big| \leq \epsilon
\]
for all $P \in \Pa$ such that $\eta_P \geq \gamma$. 
By statement \ref{item:bern} in Lemma~\ref{thm:tfae}, it follows that for each such $P$, the distribution $(X_P, B_P)$ is $(\F, \epsilon)$-indistinguishable from the distribution $(X_P, \Bern(v_P))$, where $v_P = \E_{x \sim \D\vert_P}[g(x)]$.

For each $P \in \Pa$, 
we define a joint distribution $(X_P, C_P)$ on the domain $P \times \{0, 1\}$ as follows. 
Given $X_P = x$, we set $C_P(x) = \mathrm{Bern}(v_P)$.
In order to show that $B\vert_P$ has non-uniform $(\F, \epsilon)$-PAME at least $\log(1/(1-b_P))$ given $X_P$, we prove that the joint distribution $(X_P, C_P)$ is such that (1) $\tilde{\Hb}_{\infty}(C_P|X_P) \geq \log(1/(1-b_P))$, and (2) $(X_P, B_P)$ and $(X_P, C_P)$ are $(\F, \epsilon)$-indistinguishable (by Definition~\ref{def:pamenonunif}).
We have just shown that (2) holds by Lemma~\ref{thm:tfae}; we now show (1).

To show (1), we use the original definition of average min-entropy (Definition~\ref{def:AME}):
\[
    \tilde{\Hb}_{\infty}(C_P|X_P) := 
    \log \left( \dfrac{1}{\E_{x \sim X}[1 / 2^{\Hb_{\infty}(C_P|_{X_P=x})}]} \right).
\]
Given that $C_P(x) = \mathrm{Bern}(v_P)$, it follows that
\[
    \Hb_{\infty}(C_P\vert_{X_P=x}) = -\log(\max\{v_P, 1-v_P\}) = -\log(1-b_P);
\]
\[
    \tilde{\Hb}_{\infty}(C_P|X_P) = 
    \log \Big( \dfrac{1}{1/2^{-\log(1-b_P)}} \Big) = 
    \log \Big( \dfrac{1}{1-b_P} \Big),
\]
as we wanted to show.

Lastly, we prove that $\log(1/m_P) = \log (1/(1-b_P)) = \Hb_{\infty}(B_P)$.
We observe that
\begin{align*}
    1 - b_P & = 1 - \min\{ \mathbb{E}_{x \sim X_P} [g(x)], 1 - \mathbb{E}_{x \sim X_P} [g(x)] \} = 1 - \min \{\Pr[g(X_P) = 1], \Pr[g(X_P) = 0]\} 
\\[0.2cm]
    & = \max \{\Pr[g(X_P) = 0], \Pr[g(X_P)=1]\} = 2^{-\Hb_{\infty}(B\vert_P)}.
\end{align*}
where the last equality follows from the definition of min-entropy (Definition~\ref{def:minentropy}). Then,
\[
    \log \Big( \dfrac{1}{1-b_P} \Big) = \log \Big( \dfrac{1}{2^{-\Hb_{\infty}(B_P)}} \Big) = \log \Big( 2^{\Hb_{\infty}(B_P)} \Big) = \Hb_{\infty}(B_P),
\]
as required. 
\end{proof}

\begin{proof}[Proof of Theorem~\ref{thm:pame++} for $L > 2$]

We again let $B = g^{\rand}(X)$ for $g^{\rand}: \X \rightarrow [L]$ without loss of generality.
Given the input class $\F$ of functions $f: \X \times [L] \rightarrow [0,1]$, we define a class $\F'$ of functions $f': \X \rightarrow [0,1]$ as follows.
For each $f \in \F$ and $y \in [L]$, we add $f_y(x) = f(x, y)$ to $\F'$.
We apply the Multiclass MC Theorem (Theorem~\ref{thm:mcmc}) to $\F', g, \D = X$, and with parameters $L, \gamma, \epsilon' := \epsilon/2^{\ell}$.
This yields a partition $\Pa \in \F_{t,q,k}$ of $\X$ with $t = O\big( L \cdot (1/(\epsilon'^4 \gamma))^{L} \cdot \log(|\X|/\epsilon') \big)$, $q = O\big((L/\epsilon'^2)^{O(L)}\big)$, $k = O\big( L \cdot (1/\epsilon'^4)^{L-1} \big)$ satisfying
\[
    \Big|\E_{\substack{x \sim X_P, \\ \coins(g^{\rand})}} [f'(x) \cdot (\mathbbm{1}[g^{\rand}(x)=y] - v_{P_y})] \Big| \leq \epsilon',
\]
for all $P \in \Pa$ such that $\eta_P \geq \gamma$ and for all $y \in \{0,1\}^{\ell}$, where $v_{P_y} := \E_{\substack{x \sim \D\vert_P, \\ \coins(g^{\rand})}}[\mathbbm{1}[g^{\rand}(x)=y]]$.

For each $P \in \Pa$, we define a joint distribution $(X_P, C_P)$ on the domain $P \times [L]$ as follows. 
Given $X_P = x$, we define the distribution $C_P$ as $\Pr[C_P = y] = v_{P_y}$ for each $y \in [L]$, where the $v_{P_y}$ correspond to the $L$ parameters returned by the MC Multiclass Theorem (Theorem~\ref{thm:mcmc}).
By the definition of $v_{P_y}$, it follows that $\sum_{y \in \{0,1\}^{\ell}} \Pr[C_P = y] = 1$, and so $C_P$ is indeed a probability distribution.
Moreover, by the definition of $C_P$ it follows that $\mathbbm{1}[C_P = y] \sim \Bern(v_{P_y})$ for every $y \in \{0,1\}^{\ell}$.
First, we compute a lower bound on $\tilde{\Hb}_{\infty}(C_P|X_P)$.
Recall that $m_P = \max_z \Pr[B_P = z]$.
By the definition of min-entropy, it follows that
\[
    \Hb_{\infty}(C_P\vert_{X_P=x}) = -\log(\max_c \{\Pr[C_P = c]\}) = -\log \big(\max_y v_{P_y} \big).
\]
Then,
\[
    \tilde{\Hb}_{\infty}(C_P|X_P) = \log \Bigg( \dfrac{1}{\E_{x \sim X_P} \big[ 1/2^{-\log(\max_y \{ 1-b_{P_y} \})} \big]} \Bigg) = - \log\big(\max_y v_{P_y}\big) = -\log(m_P),
\]
given that $\Pr[B_P=z] = \E[\mathbbm{1}[g^{\rand}(x)=z]]$.
Lastly, the definition of $m_P$ as the maximum probability mass assigned by $B_P$ directly implies that $m_P = 2^{-\Hb_{\infty}(B_P)}$, by definition of min-entropy. 
Hence, it follows that $\log(1/m_P) = \Hb_{\infty}(B_P)$, as claimed.

Second, we want to show that $(X_P, B_P)$ and $(X_P, C_P)$ are $(\F, \epsilon)$-indistinguishable.
For each $y \in \{0,1\}^{\ell}$ and $P \in \Pa$ such that $\eta_P \geq \gamma$, by the Multiclass MC Theorem it follows that the boolean function $\mathbbm{1}[B_P=y]$ is a constant-Bernoulli function, given that it is $(\F', \epsilon')$-indistinguishable from $v_{P_y}$. 
Therefore, by applying Lemma~\ref{thm:tfae} to the function $\mathbbm{1}[B_P=y]$, class of distinguishers $\F'$, and domain $X_P$, it follows from statement (\ref{item:bern}) that the distribution $(X_P, \mathbbm{1}[B_P=y])$ is $(\F'', \epsilon')$-indistinguishable from the distribution $(X_P, \Bern(v_{P_y}))$, where $\F'' = \F'_{O(n), O(1)}$.

This allows us to show the required indistinguishability guarantee between distributions $(X_P, B_P)$ and $(X_P, C_P)$: 
\begin{align*}
    \big| \E[f(X_P, B_P) - f(X_P, C_P)] \big| &= 
    \Big| \sum_{y \in [L]} \E_{x \sim X_P} \big[f(x, y) \cdot \big(\Pr[B_P=y \mid X=x] - \Pr[C_P = y \mid X=x] \big) \big] \Big| \\
    %&= \Big| \sum_{y \in [L]} \E_{x \sim X_P} \big[f(x, y) \cdot \big(\Pr[B_P=y \mid X=x] - v_{P_y} \big) \big] \Big| \\
    &= \Big| \sum_{y \in [L]} \E_{\substack{x \sim X_P, \\ \coins(g^{\rand})}} \big[f(x, y) \cdot \big(\mathbbm{1}[g^{\rand}(x)=y] - v_{P_y} \big) \big] \Big| \\
    &\leq \sum_{y \in [L]} \Big| \E_{\substack{x \sim X_P, \\ \coins(g^{\rand})}}[f_y(x) \cdot  (\mathbbm{1}[g^{\rand}(x)=y]-v_{P_y})] \Big| \\
    &\leq \Big| \sum_{y \in [L]} \E_{x \sim X_P} [\epsilon'] \Big|,
\end{align*}
given that $(X_P, \mathbbm{1}[B_P=y])$ being $(\F'', \epsilon')$-indistinguishable from the distribution $(X_P, \Bern(v_{P_y}))$ implies that, with respect to the functions in $\F'$, $\Pr[B_P=y \mid X=x]]$ is $\epsilon'$-close to the expected value of $\Bern(v_{P_y})$, which is $v_{P_y}$.
Thus, we conclude that
\[
    \big| \E[f(X_P, B_P) - f(X_P, C_P)] \big| \leq 2^{\ell} \cdot \epsilon' = \epsilon,
\]
by definition of $\epsilon'$.
\end{proof}

\smallskip
\textbf{Another proof of PAME$++$ using the original PAME theorem.}
Similar to the case of IHCL$++$, we can also prove our PAME$++$ theorem by coupling our lemma characterizing constant-Bernoulli functions (Lemma~\ref{thm:tfae}) with the original PAME theorem. 
We remark that this construction only works in the boolean case $L=2$.
Namely, we begin by apply the MC theorem (Theorem~\ref{thm:mcpartition}) to $\X$ to obtain a multicalibrated partition $\Pa$, and hence every $P \in \Pa$ (such that $\eta_P \geq \gamma$) satisfies the assumption of Lemma~\ref{thm:tfae}.
Then, by statement (\ref{item:hard}) in Lemma~\ref{thm:tfae}, it follows that $B=g(X)$ is $(b-2\epsilon)$-hard over $P$. 
Then, we can apply the original PAME theorem to each (large enough) $P \in \Pa$.

\subsection{Recovering the original PAME from PAME$++$}\label{sec:pamefrompame++}

Similar to how we proceeded in Section~\ref{sec:ihcl}, we now show how to recover the original PAME theorem using the PAME$++$ theorem. 
As in the case of IHCL$++$ (Section~\ref{sec:ihcl++toihcl}, the key idea is to ``glue'' together the sets $P \in \Pa$ that have enough size and are balanced enough.
Namely, such that $\eta_P$ and $m_P$ are larger than some threshold. 
Recall the definition of a ``good'' $P \in \Pa$ from Section~\ref{sec:ihcl}: we say that a set $P \in \Pa$ is $(\gamma, \tau)$\emph{-good} for some $\gamma, \tau > 0$ if $\eta_P \geq \gamma$ and $m_P \geq \tau$. 
Recall that in the original PAME theorem (Theorem~\ref{thm:pamenonunif}) we are assuming that $g$ is $\delta$-hard, unlike in the PAME$++$ statement.

There is one difference in the parameters of the PAME theorem that we recover: given that the PAME$++$ theorem only holds for $L = O(\log n)$, the subsequent recovery of the PAME theorem also only holds for $L = O(\log n)$, rather than the original $L = O(n)$ achieved by Vadhan \& Zheng \cite{vz12, zhe14}.
This is because the parameters in the Multiclass Multicalibration Theorem (Theorem~\ref{thm:mcmc}) have an exponential dependence on the number of classes.

\begin{proof}[Proof of the PAME theorem using PAME$++$]

Let $\X, \F, L, X, B, \epsilon, \delta$ be the assumption parameters in PAME. We define the parameters $\epsilon' := \epsilon^2 \delta$, $\gamma := \epsilon \epsilon'$, and invoke the PAME$++$ theorem with parameters $\epsilon, \gamma$.
By the PAME$++$ theorem, we obtain a partition $\Pa \in \F_{t,q,k}$ of $\X$ with $t = O\big( L \cdot (1/(\epsilon^4 \gamma))^{L} \cdot \log(|\X|/\epsilon) \big)$, $q = O\big((L/\epsilon^2)^{O(L)}\big)$, $k = O\big( L \cdot (1/\epsilon^4)^{L-1} \big)$ such that, for each $P \in \Pa$ where $\eta_p \geq \gamma = \epsilon \epsilon'$, $B_P$ has non-uniform $(\F, \epsilon)$-PAME at least $\log(1/m_P) = H_{\infty}(B_P)$ given $X_P$. 
Let $\tau := \epsilon \delta$. 
By Proposition~\ref{prop:kpdelta} from Section~\ref{sec:ihcl} we know that
\[
    \E_{P \sim \Pa}[m_P \cdot \mathbbm{1}_G(P)] \geq (1-\delta) \cdot (1 - O(\epsilon)).
\]
In order to prove the PAME theorem, we need to show that $B$ has $(\F, \epsilon)$-PAME at least equal to $\log(1/(1-\delta))$ given $X$. 
To do so, by definition of PAME, we need to show that there exists a distribution $C$ jointly distributed with $X$ such that (1) $\tilde{\Hb}_{\infty}(C|X) \geq \log(1/(1-\delta))$, and (2) $(X, B)$ and $(X, C)$ are $(\F, \epsilon)$-indistinguishable. 

As we did in Section~\ref{sec:ihcl}, we construct such a $C$ by doing ``gluing together'' the distributions $C_P$ obtained from invoking the PAME$++$ theorem such that $P$ is $(\gamma, \tau)$-\emph{good} (i.e., such that $\eta_P \geq \gamma$ and $m_P \geq \tau$).
Namely, for each $P \in \Pa$, let $C_P$ be distributed as $\C_P$, and recall that $\D(\Pa)\vert_G$ denotes the distribution $\D(\Pa)$ restricted to the set $\{\cup_{P \in \Pa} \, x \in P \mid \mathbbm{1}_G(P)\} \subseteq \X$.
Then, we define the distribution $\C$ on $\X$ as follows: for each $x \in \X$, we let $\C(x) = \C_P(x)$, where $P \sim \D(\Pa)|_G$.
Let $C$ correspond to the random variable distributed as $\C$; i.e., $C \sim \C$. 
By the PAME$++$ theorem, we know that for all $P \in \Pa$ such that $\eta_P \geq \gamma$,
\[
    \tilde{\Hb}_{\infty}(C_P|X_P) \geq \log \Big( \dfrac{1}{m_P} \Big).
\]
Hence, by the definition of $C$ and by applying the bound $\E_{P \sim \Pa}[m_P \cdot \mathbbm{1}_G(P)] \geq (1-\delta) \cdot (1- O(\epsilon))$, it follows that
\[
    \tilde{\Hb}_{\infty}(C|X) \geq \log \Bigg( \dfrac{1}{(1-\delta) \cdot (1-O(\epsilon))} \Bigg).
\]
Similarly, by the PAME$++$ theorem we know that every joint distribution $(X_P, C_P)$ is $(\F, \epsilon)$-indistinguishable from $(X|_P, B|_P)$ for each $P \in \Pa$ such that $\eta_P \geq \gamma$. 
Thus, by the definition of $C$, and following the same reasoning as in Section~\ref{sec:ihcl++toihcl},
it follows that the distributions $(X, C)$ and $(X, B)$ are $(\F, \epsilon)$-indistinguishable.
Putting the two facts together, by definition of PAME (Definition~\ref{def:pamenonunif}), it follows that $B$ has non-uniform $(\F, \epsilon)$-PAME at least 
\[
    \log \Bigg( \frac{1}{(1-\delta) \cdot (1-O(\epsilon))} \Bigg).
\]
As in Section~\ref{sec:ihcl++toihcl}, we can modify the distribution $\C$ to achieve the lower bound $\log(1/(1-\delta))$ on the PAME of $B$ while changing the indistinguishability parameter between $(X, C)$ and $(X, B)$ by at most $O(\epsilon)$.
Hence, this proves the PAME theorem (Theorem~\ref{thm:pamenonunif}) with $L = O(\log n)$.
\end{proof}

\subsection{Relationship between IHCL and the PAME Theorem}\label{sec:ihcltopame}

In his PhD thesis \cite{zhe14}, Zheng explains how to view the PAME theorem as a generalization of Impagliazzo's Hardcore Lemma to larger alphabets, and how IHCL implies the PAME theorem with $\ell=1$ \cite[Ch. 4]{zhe14}.
We summarize this implication and then use it coupled with our IHCL$++$ in order to sketch another proof for our PAME$++$ theorem with $L=2$.
The full proof of this implication can be found in one of the author's thesis \cite{casacuberta2023thesis}.

The key idea is that we can use a hardcore set to construct a distribution with high average min-entropy.
This is a sensible approach precisely because by Yao's Lemma or by our Lemma~\ref{thm:tfae}, we know that a hardcore set implies indistinguishability to a random bit, and a random bit has maximal min-entropy.
More specifically, given a distribution $(X, B)$ defined on $\X \times \{0,1\}$, where $B = g^{\rand}(X)$ for a $\delta$-hard function $g^{\rand}$, a family of distinguishers $\F$, we first apply IHCL to $g^{\rand}$ to obtain a hardcore set $H$ with parameter $\epsilon$.
Then, we define a distribution $(X, C)$ as follows.
Given $X=x$, let
\begin{equation}\label{eq:defC}
    C(x) := 
    \begin{cases} 
      \{0, 1\} \textrm{ each with probability } 1/2, & \textrm{ if } x \in H, \\
      g^{\rand}(x), & \textrm{ if } x \notin H.
   \end{cases}
\end{equation}
That is, if $x \in H$, then $C(x)$ returns a uniform random bit.
Then, one can show that (1) $\tilde{\Hb}_{\infty}(C|X) \geq \log (1/(1-\delta))$, and that (2) $(X, B)$ and $(X, C)$ are $(\F, \epsilon)$-indistinguishable in a similar way as in the proof of Theorem~\ref{thm:pame++}.
By definition of PAME (Definition~\ref{def:pamenonunif}), it follows that $(X, B)$ has $(\F, \epsilon)$-PAME at least $\log(1/(1-\delta))$, as required.

Using this idea, one can then prove the PAME$++$ theorem with $L=2$ using the set-version of our IHCL$++$ (Theorem~\ref{thm:ihcl++set}) as follows: 
By IHCL$++$ (our directly from Lemma~\ref{thm:tfae}), we can obtain a low-complexity partition $\Pa$ such that for every $P \in \Pa$ satisfying $\eta_P \geq \gamma$ for some $\gamma>0$, there exists a hardcore set $H_P \subseteq P$ with indistinguishability parameter $\epsilon / b_P$ and of density $|H_P|/|P| \geq 2b_P$.
Given this collection of hardcore sets, following the IHCL $\implies$ PAME construction, we build a collection of distributions $(X_P, C_P)$ on the domain $P \times \{0, 1\}$ as follows: Given $X_P = x$,
\[
    C_P(x) = 
    \begin{cases} 
      \{0, 1\} \textrm{ each with } \Pr = 1/2, & \textrm{ if } x \in H_P, \\
      g^{\rand}(x), & \textrm{ if } x \notin H_P.
   \end{cases}
\]
Then, one can show that (1) $\tilde{H}_{\infty}(C_P|X_P) \geq \log(1/(1-b_P))$, and that $(X_P, C_P)$ and $(X_P, B_P)$ are $(\F, \epsilon/b_P)$-indistinguishable, where $\vert_P$ denotes the restriction of the distribution on the domain $P \subseteq \X$.
By the definition of PAME, this shows that every $(X_P, B_P)$ has PAME at least $\log(1/(1-b_P))$, as required.
A caveat of this proof of PAME$++$ is that the set version of IHCL$++$, unlike the distribution version of IHCL$++$, requires an upper-bound on the size of $\F$ (see the parameters in Theorem~\ref{thm:ihcl++set}).
This is due to the probabilistic-method construction of the hardcore sets from the hardcore distributions (see Appendix~\ref{sec:setsmeasures} for details). 
However, we can get rid of this upper bound by modifying the previous construction as follows: Given a hardcore distribution $\Ha_P$, let $Y \sim \mathrm{Bern}(\Ha(x))$.
Then, given $X_P = x$, $C_P(x)$ is equal to a uniform random bit if $Y=1$, and equal to $g(x)$ if $Y=0$.
This is the same construction (but in reverse) in the conversion from the IHCL distribution version to the IHCL set version (see Theorem~\ref{thm:ihcl++set} and Appendix~\ref{sec:setihcl}).
This construction allows us to instead use the distribution version of IHCL$++$ to prove PAME$++$.

\section{The Dense Model Theorem}\label{sec:dmt}

The Dense Model Theorem (DMT) is a result originating from additive number theory \cite{gt08, tz08} that states the following in complexity-theoretic terminology \cite{rttv08, ttv09}: Let $S \subseteq \X$ be a $(\F, \epsilon, \delta)$-pseudodense set, which means that $S$ appears to have density at least $\delta$ to the distinguishers in $\F$ (with an $\epsilon$ slack), but is not necessarily dense.
The DMT shows that there exists a distribution that is $(\F', \epsilon/\delta)$-indistinguishable from $S$ that is actually dense (which is referred to as a ``dense model'' for $S$).
A DMT was one of the crucial proof components used in Green and Tao's celebrated result that there exist arbitrarily long arithmetic progressions of primes \cite{gt08}. 
A more general Dense Model Theorem was later proven by Tao and Ziegler \cite{tz08}, which generalized it to other domains. 
Reingold, Trevisan, Tulsiani, and Vadhan \cite{rttv08} gave complexity-theoretic analogues of the Dense Model Theorem, as well as a simpler proof of the DMT (see also \cite{gow10}).

\subsection{Definitions}

For this section, we require the definition of indistinguishable distributions (Definition~\ref{def:indistdist}) and density of a distribution (Definition~\ref{def:densitydist}).
Given a distribution $\D$ on $\X$, the density of a subset $S \subseteq \X$ corresponds to the probability mass endowed to $S$ by the distribution $\D$.
Hence a set $S$ also induces a distribution over $\X$.
When we say that a distribution $\D$ on $\X$ is a ``model'' for a distribution $\Sa$ defined on $\X$, we mean that the distributions $\D$ and $\Sa$ are indistinguishable (Definition~\ref{def:indistdist}).

If a distribution $\Sa$ has density $\delta$ in $\D$, then for every $f: \X \rightarrow [0,1]$ the following holds \cite{imp08}:
\[
    \E_{x \sim \D}[f(x)] \geq \delta \cdot \E_{x \sim \Sa}[f(x)].
\]
This motivates the definition of \textit{pseudodensity}, which relaxes the notion of density: 

\begin{definition}[Pseudodensity \cite{imp09}]
{\rm
Given a class $\F$ of functions $f: \X \rightarrow [0,1]$, distributions $\D, \Sa$ on $\X$, and $\epsilon, \delta > 0$, we say that $\Sa$ is \emph{$(\F, \epsilon, \delta)$-pseudodense} in $\D$ if for all $f \in \F$,
\[
    \E_{x \sim \D}[f(x)] \geq \delta \cdot \E_{x \sim \Sa} [f(x)] - \epsilon.
\]
}
\end{definition}

We can think of pseudodensity as saying that the distinguishers $f \in \F$ cannot tell whether $\Sa$ has low density or not. 
As shown above, if $\Sa$ is $\delta$-dense, then $\Sa$ is $(\F, \epsilon, \delta)$-pseudodense. 
But another way for $\Sa$ to be pseudodense in $\D$ is the following: If there is a distribution $\mu$ that is $(\F, \epsilon)$-indistinguishable from $\Sa$ by $\F$ and such that $\mu$ is $\delta$-dense in $\D$ \cite{simons17}. 
Assuming the $\delta$-pseudodensity of $\Sa$, the Dense Model Theorem precisely states that we can always find such a $\mu$. 
Hence, the interesting case occurs when the density of $\Sa$ is much smaller than $\delta$, because then the DMT shows that, even in that case, we can still find a distribution of large density that is indistinguishable from $\Sa$ to $\F$. 
Similar to the difference of the classes of distinguishers in the assumption and the conclusion of IHCL and DMT, the pseudodensity assumption on $\Sa$ is with respect an enlarged class of distinguishers, whereas the conclusion is with respect to the original class $\F$. 

We can now state the Dense Model Theorem:

\begin{theorem}[DMT, {\cite[Thm.\ 3]{imp09}}, \cite{rttv08}]\label{thm:DMT1}
Let $\X$ be a finite domain, $\F$ a family of functions $f\colon \X \rightarrow [0,1]$, $\V$ and $\Sa$ two distributions on $\X$, and $\epsilon, \delta > 0$. Then there exist $t, q = \emph{poly} (1/\epsilon, 1/\delta)$ such that if $\Sa$ is $(\F_{t, q}, \epsilon, \delta)$-pseudodense in $\V$, then there exists a $(\delta - O(\epsilon))$-dense distribution $\mu$ in $\V$ such that $\mu$ is $(\F, O(\epsilon/\delta))$-indistinguishable from $\Sa$. 
\end{theorem}

Therefore, in order to prove the DMT, we need to construct a distribution $\mu$ such that:

\begin{enumerate}
    \item \emph{Density.} Distribution $\mu$ has density $\delta - O(\epsilon)$ in $\V$.
    \item \emph{Indistinguishability.} Distributions $\mu$ and $\Sa$ are $O(\F, \epsilon/\delta)$-indistinguishable.
\end{enumerate}

\smallskip
\textbf{Similarities to Impagliazzo's Hardcore Lemma.} The Dense Model Theorem shares many similarities to Impagliazzo's Hardcore Lemma (Section~\ref{sec:ihcl}). 
We highlight some of their parallels, which also apply to their $++$ counterparts, as we will see with our DMT$++$ statement. 
\begin{enumerate}
    \item \textbf{The assumption.} 
    In IHCL, we assume that the function $g$ is $(F_{t, q}, \delta)$-hard. 
    In DMT, we assume that $\Sa$ is $(\F_{t,q}, \epsilon, \delta)$-pseudodense in $\V$. 
    \item \textbf{The density.} 
    In IHCL, we want to ensure that the hardcore distribution $\mu$ is dense enough; more concretely, $2\delta$-dense. 
    In DMT, we want the distribution $\mu$ to be dense enough; more concretely, $(\delta-O(\epsilon))$-dense. 
    \item \textbf{Indistinguishability.} 
    In IHCL, we need to ensure that $g$ is strongly hard when sampling according to $\mu$, which corresponds to $g$ being indistinguishable from the constant $1/2$ function when sampling according to $\mu$.
    In DMT, we want the distribution $\mu$ to be indistinguishable from $\Sa$. 
\end{enumerate}

\subsection{The DMT++ Theorem}

We can now state our DMT$++$ theorem:

\begin{theorem}[DMT$++$]\label{thm:dmt++}

Let $\X$ be a finite domain, let $\F$ be a family of functions $f\colon \X \rightarrow [0,1]$, let $S, V \subseteq \X$ be two disjoint sets, let distributions $\Sa$ and $\V$ be defined on $S$ and $V$, respectively, and let $\epsilon, \gamma > 0$.
There exists a partition $\Pa \in \F_{t, q, k}$ of $S \cup V$ with $t = O(1/(\epsilon^4 \gamma) \cdot \log(|\X|/\epsilon))$, $q = O(1/\epsilon^2)$, $k = O(1/\epsilon)$, which satisfies that for each $P \in \Pa$ such that $\eta_P \geq \gamma$, distributions $\Sa_P$ and $\V_P$ are $(\F, \epsilon_P)$-indistinguishable for all $P \in \Pa$, where
\[
    \eta_P = \dfrac{1}{2} \Big( \Pr_{x \sim \Sa}[x \in P] + \Pr_{x \sim \V}[x \in P] \Big), \quad 
    v_P = \dfrac{\Pr_{x \sim \Sa}[x \in P]}{2 \eta_P}, \quad
    \epsilon_P = \dfrac{\epsilon}{v_P \cdot (1-v_P)}.
\]
\end{theorem}

By the definition of density of a distribution, it immediately follows that $\V\vert_P$ is $\delta_P$-dense in $\V$ for $\delta_P = \Pr_{x \sim \V}[x \in P]$.
We can also write the density of $\V\vert_p$ in $\V$ as a function of $v_P$ and $\eta_P$ as follows: $\delta_P = Pr_{x \sim \V}[x \in P] = 2\eta_P \cdot (1-v_P)$.
When $g: \X \rightarrow \{0,1\}$ corresponds to the indicator function for the set $S$, then the parameters $v_P, \eta_P$ correspond to the usual interpretations (i.e., the expected value of $g$ over $P$ and the size parameter of $P$, respectively). 

We remark that, since $b_P = \min\{v_P, 1-v_P\}$, it follows that $\frac{1}{2} \cdot b_P \leq v_P (1-v_P) \leq b_P$. 
Therefore, the term $v_P(1-v_P)$ in $\epsilon_P$ is capturing the balance of $g$ (i.e., the characteristic function of the set $S$) on the set $P \in \Pa$.

\textbf{Interpretation of the DMT$++$.} Before going into the proof of Theorem~\ref{thm:dmt++}, we explain what our DMT++ theorem entails and how it is a stronger and more general version of the original DMT theorem.

\begin{itemize}
    \item In Theorem~\ref{thm:dmt++}, we remove the assumption that the $\Sa$ is $(\F_{t, q}, \epsilon, \delta)$-pseudodense, but we can still create a partition $\Pa$ of the domain $S \cup V$ such that every piece $P$ in the partition has a ``model'' for a piece of the distribution $\Sa$.
    \item We provide a generalization of the notion of pseudodensity, and we are able to argue about the density of each model with the general parameter $\delta_P = \Pr_{x \sim \D}[x \in P]$, rather than using the parameter $\delta$ in the original DMT statement (which we do not have, since we do not have the $(\epsilon, \delta)$-pseudodensity assumption on $\Sa$).
    
    \item In our $++$ theorem, the original DMT occurs both ``locally'' (on each $P \in \Pa$) and ``globally'' (on $\X)$. 
    Moreover, as it was the case of IHCL$++$ and PAME$++$, we can recover the original DMT theorem from our DMT$++$ by ``gluing'' together the ``models'' $\V\vert_P$ of each piece $P$ (according to the weight of each piece as dictated by the distribution $\V$). 
    The density $\Pr_{x \sim \V}[x \in P]$ of each model is such that when we bring back the assumption that $\Sa$ is $(\F_{t, q}, \epsilon, \delta)$-pseudodense, we obtain that the ``global'' model for $\Sa$ has density $O(\delta-O(\epsilon))$, as in the original DMT (Theorem~\ref{thm:DMT1}).
\end{itemize}

An illustration of a simplified version of the DMT$++$ Theorem can be found in Figure~\ref{fig:dmt}.

\begin{figure}[h!]
\centering
    \includegraphics[width=11.5cm]{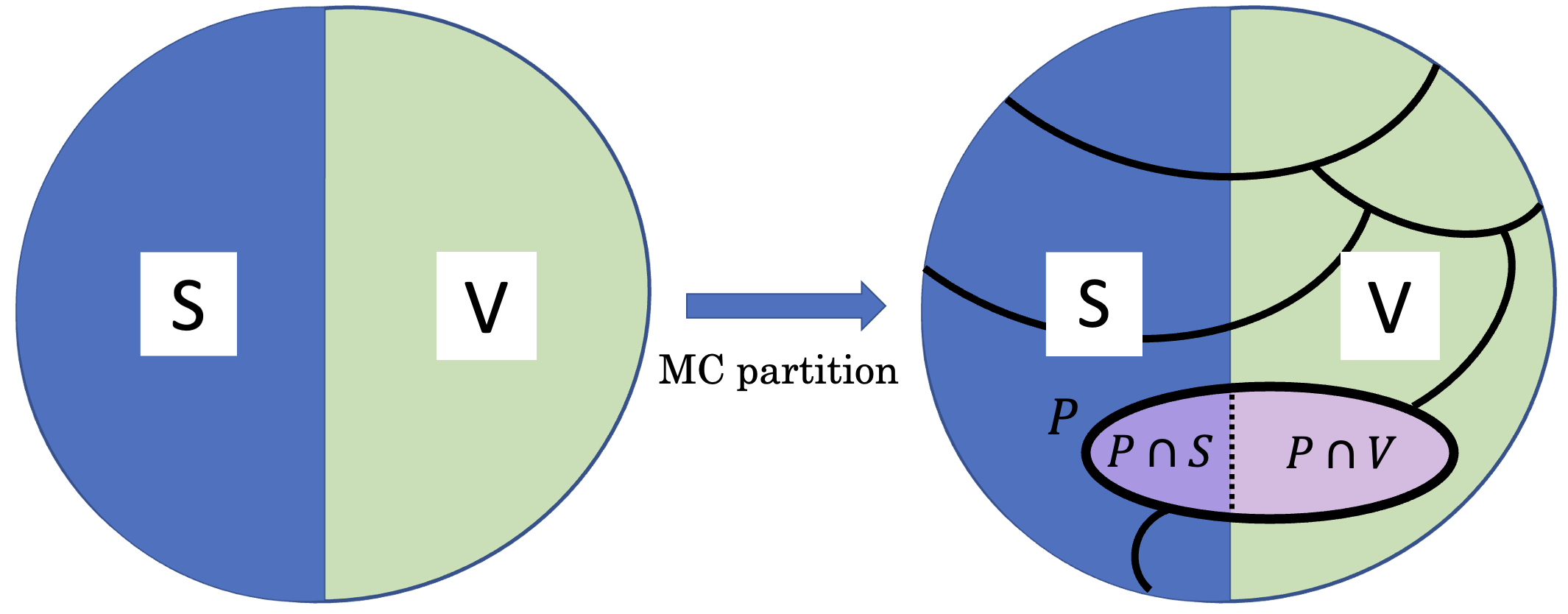}
    \caption{Visual depiction of Theorem~\ref{thm:dmt++} in the case where $\Sa$ and $\V$ correspond to the uniform distributions over $S$ and $V$, respectively. In this case, $\Pr_{x \sim \Sa}[x \in P] = |P \cap S|/|S|$ and $\Pr_{x \sim \V}[x \in P] = |P \cap V|/|V|$. The MC partition (indicated with black lines in the right figure) is such that, for each level set $P$ in the partition, the uniform distribution over $P \cap V$ is a model for the corresponding uniform distribution over $P \cap S$. One such level set is illustrated in purple in the right figure.}\label{fig:dmt}
\end{figure}

\begin{proof}[Proof of Theorem~{\rm\ref{thm:dmt++}}]

Let $g: \X \rightarrow \{0,1\}$ be the indicator function for the set $S$, and we let $\D$ be the distribution over $\X$ defined as $\frac{1}{2} \Sa + \frac{1}{2} \V$. 
We invoke the MC Theorem (Theorem~\ref{thm:mcpartition}) with the domain $\X$, function family $\F$, function $g$, distribution $\D$, and the original $\epsilon, \gamma$ parameters.
This yields a partition $\Pa \in \F_{t, q, k}$ with $t = O(1/(\epsilon^4 \gamma) \cdot \log(|\X|/\epsilon))$, $q = O(1/\epsilon^2)$, $k = O(1/\epsilon)$
satisfying
\[
    \Big|\E_{x \sim \D\vert_P} [f(x) \cdot (g(x) - v_P)] \Big| \leq \epsilon
\]
for all $P \in \Pa$ such that $\eta_P \geq \gamma$. 

We claim that this partition $\Pa$ satisfies the conditions of Theorem~\ref{thm:dmt++}.\footnote{Note that there is no function $g$ in the statement of Theorem~\ref{thm:dmt++}, and so here $v_P$ and $\eta_P$ are as defined in the theorem statement. The point is that when $g$ corresponds to the characteristic function of the set $S$ and $\D$ corresponds to $\frac{1}{2} \Sa + \frac{1}{2} \V$, then $v_P$ and $\eta_P$ correspond to our usual balance and size parameters (Definitions \ref{def:mcpartition}, \ref{def:densityparam}).} 
By the definition of the distribution $\V$, it follows that
\[
    \eta_P = \Pr_{x \sim \D}[x \in P] = 
    \dfrac{1}{2} \Big( \Pr_{x \sim \Sa}[x \in P] + \Pr_{x \sim \V}[x \in P] \Big).
\]
Fix any $P \in \Pa$ such that $\eta_P \geq \gamma$.
First, we notice that since $g$ is boolean and $g(x) = 1$ if and only if $x \in S$ by definition, it follows that
\[
    v_P = \dfrac{\Pr_{x \sim \Sa}[x \in P]}{\Pr_{x \sim \Sa}[x \in P] + \Pr_{x \sim \V}[x \in P]} =
    \dfrac{\Pr_{x \sim \Sa}[x \in P]}{2\eta_P}.
\]
We need to show that the distributions $\Sa\vert_P$ and $\V\vert_P$ are $(\F, \epsilon_P)$-indistinguishable, where 
\[
    \epsilon_P = \dfrac{\epsilon}{v_P (1-v_P)}.
\]
By definition of indistinguishability, this corresponds to showing that
\[
    \Big| \E_{x \sim \Sa\vert_P}[f(x)] - \E_{x \sim \V\vert_P}[f(x)] \Big| \leq \dfrac{\epsilon}{v_P (1-v_P)}.
\]
This follows from our Lemma~\ref{thm:tfae}: by the MC Theorem, we satisfy the assumption of Lemma~\ref{thm:tfae} within each $P \in \Pa$ such that $\eta_P \geq \gamma$. 
Therefore, by statement (\ref{item:0s1s}) in Lemma~\ref{thm:tfae} it follows that the 0s and the 1s of $g$ are indistinguishable.
By our definition of $g$ as the characteristic function of $S$, this corresponds precisely to saying that the distinguishers $f \in \F$ cannot tell whether we are sampling from $\Sa\vert_P$ or from $\V\vert_P$ with advantage larger than $\frac{\epsilon}{v_P (1-v_P)}$, as we wanted to show.
\end{proof}

\subsection{Recovering the original DMT from DMT$++$}\label{sec:recoveringdmt}

We show how to recover the original DMT theorem from our DMT$++$.
As it was the case of IHCL and PAME, the key idea is to ``glue together'' the models $\V_P$ for each ``good'' $P \in \Pa$. 
Namely, those $P \in \Pa$ that have $\eta_P \geq \gamma$ and $b_P \geq \tau$ for parameters $\gamma, \tau$ (Definition~\ref{def:goodP} from Section~\ref{sec:ihcl}). 
Recall that in the DMT statement we are bringing back the assumption that $\Sa$ is $(\F_{t, q}, \epsilon, \delta)$-pseudodense. 
As we did in Section~\ref{sec:ihcl}, we begin by showing an intermediate proposition:

%\begin{tcolorbox}
\begin{proposition}\label{prop:pseudodense}

Let $\X, \F, S, V, \Sa, \V, \epsilon, \gamma, \Pa, \eta_P, t, q, k$ be as in Theorem~{\rm\ref{thm:dmt++}}. Moreover, assume that 
$\Sa$ is $(F_{t+k, q}, \epsilon, \delta)$-pseudodense in $\V$ for some $\delta>0$. 
Then, for all $P \in \Pa$ such that $\eta_P \geq \gamma$,
\[
    \delta \cdot \Pr_{x \sim \Sa}[x \in P] \leq \Pr_{x \sim \V}[x \in P] + \epsilon.
\]
\end{proposition}
%\end{tcolorbox}

\begin{proof}
Recall that by definition of pseudodensity, and by the assumption that $\Sa$ is $(\F_{t+k,q}, \epsilon, \delta)$-pseudodense in $\V$, we know that for all $f \in \F_{t+k, q}$,
\[
    \delta \cdot \Pr_{x \sim \Sa}[f(x) = 1] - \epsilon \leq \Pr_{x \sim \V}[f(x) = 1].
\]
Fix some $P \in \Pa$.
Let $\hat{f} \in \F_{t, q}$ where $t = O(1/(\epsilon^4 \gamma) \cdot \log(|\X|/\epsilon))$, $q = O(1/\epsilon^2)$ be the partition membership function for $\Pa$ as given by Definition~\ref{def:Ftqk}. 
That is, $\Pa = \{\hat{f}^{-1}(1), \ldots, \hat{f}^{-1}(k)\}$. 
Given $P \in \Pa$, we define the function $f^P$ as $f^P = f^P_{\post} \circ \hat{f}$, where $f^P_{\post}(i) = \mathbbm{1}[\hat{f}(i) = P]$.
That is, we think of $f^P$ as the indicator function for the set $P$.
As argued in Section~\ref{sec:ihcl++toihcl},  $f^P_{\post}$ can be computed by a circuit of size $k$ (see \cite[\S 9.1.1.]{boazbook2}), so $f^P \in \F_{t+k, q}$.

Then, by applying the pseudodensity assumption on $\Sa$ to this function $f^P \in \F_{t+k, q}$, we obtain that
\[
    \delta \cdot \Pr_{x \sim \Sa}[f^P(x) = 1] - \epsilon \leq \Pr_{x \sim \V}[f^P(x) = 1].
\]
Then, by the definition of $f^P$ it follows that
\[
    \Pr_{x \sim \Sa}[f^P(x) = 1] = \Pr_{x \sim \Sa}[x \in P], \quad \Pr_{x \sim \V}[f^P(x) = 1] = \Pr_{x \sim \V}[x \in P], 
\]
Therefore, we have that
\[
    \delta \cdot \Pr_{x \sim \Sa}[x \in P] \leq \Pr_{x \sim \V}[x \in P] + \epsilon,
\]
which corresponds to the inequality stated in Proposition~\ref{prop:pseudodense}. 
\end{proof}

Given this proposition, we can now prove the DMT from our DMT$++$.
In doing so, we use a trick originally due to Impagliazzo, in his series of works where he explores the relationship between IHCL and the DMT \cite{imp08, imp09, gik12, simons17}.
More specifically, Impagliazzo shows that IHCL implies the DMT: given a pseudodense set $S$, one can prove that the function $g$ defined as the characteristic function of $S$ is weakly hard, if we consider an appropriate distribution $\D$ over an augmented domain $\X'$,\footnote{One needs to be careful with defining the distribution $\D$, because the DMT is most interesting when $S$ has very low density, and thus we need to magnify $S$ in order to ensure that we sample from $S$ with some constant probability.} where we append the bit $b=1$ to $x \in \X$ such that $x \in S$, and append the bit $b=0$ to $x \in \X$ such that $x \in \X \setminus S$.
We can then construct a ``model'' for $S$ by using the hardcore distribution for $g$ given by the IHCL.
A full proof of this implication can be found in \cite[\S 6.2]{casacuberta2023thesis}.

In the following proof, we use Impagliazzo's trick of augmenting the domain using $\Sa$:

\begin{proof}[Proof of DMT from DMT$++$]

Let $\X, \F, \Sa, \V, \epsilon, \delta$ be the assumption parameters in the DMT statement (Theorem~\ref{thm:DMT1}). 
We define the parameters $\epsilon' := \epsilon^2 \delta$, $\gamma := \epsilon \epsilon'$, $\tau := \epsilon \delta$,
and invoke the DMT$++$ Theorem (Theorem~\ref{thm:dmt++}) with the domain $\X' = \X \times \{0,1\}$, $S = \X \times \{1\}$, $V = \X \times \{0\}$, $\Sa' = \Sa \times \{1\}$, and $\V'  = \V \times \{0\}$, the function family $\F' = \{f'(x, b) = f(x)\}$ for some $f \in \F$, and parameters $\epsilon', \gamma$.
By the DMT$++$ (Theorem~\ref{thm:dmt++}), we obtain a partition $\Pa \in \F_{t,q,k}$ of $\X'$ with $t = O(1/(\epsilon'^4 \gamma) \cdot \log(|\X'|/\epsilon'))$, $q = O(1/\epsilon'^2)$, $k = O(1/\epsilon')$ such that, for each $P \in \Pa$ where $\eta_P \geq \gamma = \epsilon \epsilon'$, the distribution $\V'\vert_P$ is $(\F', \epsilon'_P)$-indistinguishable from $\Sa'\vert_P$. 
Given these ``local models'' $\V'\vert_P$, we construct the claimed global distribution $\mu$ as follows: we define $\mu$ to be the weighted average of the ``models'' $\V'\vert_P$ such that $P \in \Pa$ is $(\gamma, \tau)$-good.
Formally, for each $x \in V$, let
\[
    \mu(x) =  \dfrac{\Pr_{(x, 1) \sim \Sa'}[(x, 1) \in P] \cdot \V'(x, 0) \cdot \mathbbm{1}_G(P)}{\E_{Q \sim \Pa(\Sa')} [\mathbbm{1}_G(Q)]},
\]
where $P$ corresponds to the unique $P \in \Pa$ such that $x \in P$ (which is unique since $\Pa$ is a partition), and recall that $\mathbbm{1}_G(P)$ returns 1 if $\eta_P \geq \gamma$ and $k_P \geq \tau$, and 0 otherwise.
Note that the purpose of the denominator is to normalize $\mu$ so that $\sum_x \mu(x) = 1$.
The expression for $\mu(x)$ should be understood as follows: in order to choose a $x \in V$, we can instead first choose a set $P$ with probability $\Pr_{x \sim \Sa'}[x \in P]$, and then sample from $\V'\vert_P$.

Let $v_P$ as in Theorem~\ref{thm:dmt++}.
First, by combining Proposition~\ref{prop:pseudodense} and the fact that
\[
    \Pr_{x \sim \V'}[x \in P] = 2 \eta_P (1-v_P),
\]
it follows that
\[
    \dfrac{1}{\delta} + \dfrac{\epsilon'}{\delta} \cdot \Pr_{x \sim \V'}[x \in P] = \dfrac{1}{\delta} + 
    \dfrac{\epsilon'}{2\delta \eta_P (1-v_P)} = \dfrac{1}{\delta \cdot (1-O(\epsilon'/(\eta_P(1-v_P)))}.
\]
Recall that we only ``glue up'' the parts $P$ that are $(\gamma, \tau)$-good.
Then, following a similar analysis as the one described in Section~\ref{sec:ihcl++toihcl}, it follows that
\[
    \E_{P \sim \D(\Pa)} \Bigg[ \dfrac{\epsilon'}{v_P (1-v_P)} \cdot \mathbbm{1}_G(P) \Bigg] \leq \dfrac{\epsilon'}{\tau},
\]
where $\epsilon' = \epsilon^2 \delta$, $\gamma = \epsilon \epsilon'$, and $\tau = \epsilon \delta$.
Therefore, the density of the ``glued up'' distribution $\mu$ when we only glue up the $(\gamma, \tau)$-good sets is equal to $\delta \cdot (1 - O(\epsilon))$. 

Next, we show the indistinguishability condition; namely, that $\mu$ is $\F, O(\epsilon/\delta)$-indistinguishable from $\Sa$.
By Theorem~\ref{thm:dmt++}, we know that for all good $P \in \Pa$,
\[
    \Big| \E_{x \sim \Sa'_P}[f(x)] - \E_{x \sim \V'_P}[f(x)] \Big| \leq \epsilon'_P.
\]
By the definition of $\mu$ (which implies that $\mu(x)$ is never contained in $S$ given that all values $x \sim \mu$ are in $V$) and the law of total probability, it follows that
\begin{align*}
    \big| \E_{x \sim \Sa'}[f(x)=1] - \E_{x \sim \mu}[f(x)] \big| &= 
    \Big| \sum_{P \in \Pa} \Big( \Pr_{x \sim \Sa'_P}[f(x)=1] \cdot \Pr_{x \sim \Sa'}[x \in P] \\
    &\qquad \qquad \qquad \quad \, \, -\Pr_{x \sim \V'}[f(x)=1] \cdot 
    \Pr_{x \sim \V'}[x \in P] \Big) \cdot \mathbbm{1}_G(P) \Big| \\
    &\leq \sum_{P \in \Pa} \Pr_{x \sim \Sa'}[x \in \Pa] \cdot \Big| \E_{x \sim \Sa'_P}[f(x)] - \E_{x \sim \V'_P}[f(x)] \Big| \cdot \mathbbm{1}_G(P).
\end{align*}
By applying Theorem~\ref{thm:dmt++} to each of the summands, the above expression becomes
\begin{align*}
    \sum_{P \in \Pa} \epsilon'_P \cdot \Pr_{x \sim \Sa'}[x \in P] \cdot \mathbbm{1}_G(P) &= \sum_{P \in \Pa} \dfrac{\epsilon'}{v_P \cdot (1-v_P)} \cdot \Pr_{x \sim \Sa'}[x \in P] \cdot \mathbbm{1}_G(P) \\
    &\leq \dfrac{1}{\delta} \sum_{P \in \Pa} \dfrac{\epsilon'}{v_P \cdot (1-v_P)} \cdot \Pr_{x \sim \V'}[x \in P] \cdot \mathbbm{1}_G(P).
\end{align*}
Since $\Pr_{x \sim \V'}[x \in P] = 2\eta_P \cdot (1-v_P)$, following the analysis from Section \ref{sec:ihcl++toihcl} it follows that each term in the summand is at most $2 \tau / v_P \leq 2 \epsilon' / \tau$.
Since $\epsilon' = \epsilon^2 \delta$ and $\tau = \epsilon \delta$ imply that $\epsilon' / \tau = \epsilon$, it follows that 
\[
    \Big| \E_{x \sim \Sa'}[f(x)=1] - \E_{x \sim \mu}[f(x)] \Big| \leq \dfrac{\epsilon}{\delta},
\]
as we wanted to show.
As in Section~\ref{sec:ihcl}, we can modify $\mu$ so that we obtain density exactly $\delta-O(\epsilon)$ while maintaining $O(\epsilon/\delta)$-indistinguishability for the model.
Thus, we have recovered the original DMT theorem.
\end{proof}

\textbf{Entropy interpretation of the DMT$++$.} 
To conclude, we observe that the DMT and DMT$++$ theorems can also be restated in terms of pseudo-min-entropy.
Indeed, Reingold, Trevisan, Tulsiani, and Vadhan related the DMT and the notion of pseudoentropy (Definition~\ref{def:pseudominen}) \cite[\S 4.1]{rttv08}.
Essentially, we can relate the pseudodensity of a set $\Sa$ to the notion of pseudo-min-entropy by re-stating the DMT theorem as follows: If a distribution $\D$ is $\delta$-dense in an $\epsilon$-pseudorandom set (i.e., $\D$ is $\epsilon$-indistinguishable from the uniform distribution on $\X$), then $\D$ has $\epsilon$-pseudo-min-entropy $\log|\X| - \log(1/\delta)$ \cite[Propositions 4.2, 4.3]{rttv08}.
As future work, we believe it would be illustrative to re-state our DMT$++$ theorem in terms of pseudo-min-entropy and compare it to the conclusion of the PAME$++$ theorem, which would further clarify how the fundamental concepts that we have considered in this paper (indistinguishability, hardness, pseudo-average min-entropy, pseudo-min-entropy, pseudodensity) relate to each other.

\section{Conclusions \& Future Work}\label{sec:conclusions}

We conclude by providing some other directions for future work.

\smallskip
\textbf{Further generalizations of our results.} 
Our results regarding the PAME$++$ could be generalized in two different ways.
First, while our results in Section~\ref{sec:pame} only concern non-uniform distinguishers (i.e., where the distinguishers correspond to boolean circuits), the original results by Vadhan \& Zheng also consider uniform distinguishers (i.e., where the distinguishers correspond to polynomial time algorithms), and their PAME theorem is thus stated for both the uniform and non-uniform settings \cite{vz12, zhe14}.
While we expect our results to carry onto the uniform setting, they require a different type of technical treatment.
Second, given the exponential dependence on the number of classes in the Multiclass MC Theorem \cite{gkrsw21}, our PAME$++$ theorem only holds for $\ell = O(\log \log n)$, whereas the original PAME theorem holds for $\ell = O(\log n)$.
We would hope to be able to recover the original parameters.
The complexity of multiclass MC has been recently studied in \cite{gopalan2024computationally}.

\smallskip
\textbf{Future multicalibration algorithms.} There are two possible ways of improving our PAME$++$ parameters: One is to design a multiclass multicalibration algorithm with better complexity parameters (a topic which is already a current line of work within the multicalibration literature), and the other is to try to achieve the PAME$++$ conclusion with a weaker notion of multicalibration instead that has a polynomial dependence on the number of classes.

Relatedly, all of the partition complexity parameters in our $++$ theorems directly depend on the state-of-the-art multicalibration parameters from the MC literature, given that we use the approximate MC theorems of \cite{hkrr18} and \cite{gkrsw21} as black-boxes, and thus any future improvement of the MC algorithms is applicable to our results.
In particular, we expect that the number of calls to the weak agnostic learner in these type of algorithms can be reduced, given that the original MC algorithms were more concerned with showing feasibility rather than efficiency. 

\smallskip
\textbf{Other possible applications of our $++$ approach.} 
The observation that Multicalibration Theorem corresponds to a stronger and more general version of the Regularity Lemma, which is the starting point of both this paper and the concurrent work by Dwork, Lee, Lin, and Tankala \cite{dllt23}, as well as the framework and tools that we have developed in this paper (in particular, our complexity-theoretic reframing of multicalibration), can be applied to other important implications of the Regularity Lemma outside of IHCL, PAME, the DMT, and Szemer\'edi's Regularity Lemma.
As mentioned in the introduction, other implications of the Regularity Lemma include applications in leakage resilient cryptography \cite{jp14, ccl18}, weak notions of zero-knowledge \cite{clp15}, Yao’s XOR Theorem \cite{gnw11}, chain rules for computational entropy \cite{gnw11, jp14}, and Chang’s inequality in Fourier analysis of boolean functions  \cite{imr14}.

We expect that viewing these theorems as corollaries of the Multiaccuracy Theorem will provide new insights into these established results, as well as new connections and insights into the power of a multicalibrated partition.

%\iffalse 
\section*{Acknowledgements}

We are indebted to  Parikshit Gopalan, Fabian Gundlach, Daniel Lee, Huijia (Rachel) Lin, Omer Reingold, Jessica Sorrell, Pranay Tankala, and Benjy Firester for helpful conversations throughout the development of this work. 
We are also thankful to the Simons' workshop on Multigroup Fairness and the Validity of Statistical Judgement that took place in April 2023, where we presented this work and received insightful feedback.
Lastly, we are grateful to the anonymous STOC and FORC reviewers for their additional suggestions.

%\fi

\printbibliography

\appendix

\section{Multicalibrated Partitions}\label{sec:mcpartitions}

In this section, we provide the proof for Theorem~\ref{thm:mcpartition}.

\begin{theorem}[Theorem~\ref{thm:mcpartition}, MC Theorem \cite{hkrr18}]
Let $\X$ be a finite domain, $\F$ a class of functions $f\colon \X \rightarrow [0,1]$, $g\colon \X \rightarrow [0,1]$ an arbitrary function, $\D$ a probability distribution over $\X$, and $\epsilon, \gamma > 0$. There exists an $(\F, \epsilon, \gamma)$-approximately multicalibrated partition $\Pa$ of $\X$ for $g$ on $\D$  
such that $\Pa \in \F_{t, q, k}$, where
\begin{itemize}
    \item[{\rm 1.}] $t = O(1/(\epsilon^4 \gamma) \cdot \log(|\X|/\epsilon))$,
    \item[{\rm 2.}] $q = O(1/\epsilon^2)$,
    \item[{\rm 3.}] $k = O(1/\epsilon)$.
\end{itemize}
\end{theorem}

The goal is to build an approximate MC predictor $h$ that only has $k = O(1/\epsilon)$ level sets; the parameters $t$ and $q$ are given by the H\'ebert-Johnson et al. MC algorithm.
From this predictor, we then obtain an MC partition as stated in Theorem~\ref{thm:mcpartition} by taking the partition induced by the level sets of the predictor.

The intuitive idea behind why we can obtain an MC predictor $h$ with $O(1/\epsilon)$ level sets is to discretize the domain $[0,1]$ into ``buckets'' of width $\gamma$ and then round the values of $h$ to the closest point in the discretization.
After the rounding, $h$ has at most $1/\lambda$ level sets, yet the rounding can only have changed each output of $h$ by at most $\lambda$.
Therefore, the MC property of $h$ is maintained after the rounding (up to an additive factor of $\lambda$). 

However, it is not enough to apply this procedure to an approximate $\gamma$-MC predictor, given that we want the predictor to have $k = O(1/\epsilon)$ level sets \emph{independent of $\gamma$.}
To do so, we use in the proof a different relaxation of the notion of multicalibration which we mentioned in the introduction; namely, \emph{multicalibration on average (MCoA)}.
In the case of MCoA, there is no parameter $\gamma$; instead, in order to deal with the small level sets, the MC condition is only required to hold on average over the level sets:

\begin{definition}
\rm{
Let $\X$ be a finite domain, $\F$ a collection of functions $f\colon \X \rightarrow [0,1]$, $g\colon \X \rightarrow [0,1]$ an arbitrary function, $\D$ a probability distribution over $\X$, and $\epsilon > 0$. 
%Let $X_v = \{x \in \X \mid h(x) = v\}$ for all $v \in \textrm{range} (h)$. 
We say that a predictor $h\colon \X \rightarrow [0,1]$ is \emph{$(\F, \epsilon)$-multicalibrated on average} (MCoA) for $g$ on $\D$ if, for all $f \in \F$,
\[
    \E_{X_v \sim \Pa(\D)} \Big|\E_{x \sim \D|_{v}} [f(x) \cdot (g(x) - h(x))] \Big| \leq \epsilon,
\]
where $\D|_v$ denotes the conditional distribution $\D|_{h(x)=v}$ for $v \in [0,1]$ in the support of~$h$, and $\Pa(\D)$ denotes the distribution that selects each $X_v$
with probability $\big(\sum_{x \in X_v} \D(x)\big) / \big(\sum_{x \in \X} \D(x)\big)$.
}
\end{definition}
%\end{tcolorbox}

Then, when we consider a fixed level set $X_v = \{x \in \X \mid h(x) = v\}$, it follows that
\[
    \Big|\E_{x \sim \D|_v} [f(x) \cdot (g(x) - h(x))] \Big| \leq \dfrac{\epsilon}{\Pr_{\D}[x \in X_v]}.
\]
Therefore, in the case of MCoA, we allow the indistinguishability parameter to degrade with the size of the level set (the smaller $X_v$, the worse the guarantee becomes), but if we re-parametrize this second approach by setting $\epsilon \leftarrow \epsilon \gamma $, we obtain exactly the definition of approximate MC, given that 
\[
    \Big|\E_{x \sim \D|_v} [f(x) \cdot (g(x) - h(x))] \Big| \leq \epsilon  \cdot \gamma \cdot \dfrac{\Pr_{\D}[x \in X]}{\Pr_{\D}[x \in X_v]} \leq \epsilon \cdot \dfrac{\Pr_{\D}[x \in X_v]}{\Pr_{\D}[x \in X]} \cdot \dfrac{\Pr_{\D}[x \in X]}{\Pr_{\D}[x \in X_v]} = \epsilon.
\]
Similar to how we have an approximate MC theorem, the notion of MCoA is likewise efficiently realizable.

Next, we formally introduce the discretization of the domain:

\begin{definition}[$\lambda$-discretization]
{\rm
Let $0<\lambda <1$. The \emph{$\lambda$-discretization} of $[0,1]$ is the set
%denoted by 
\[
    \Lambda[0,1] = \Big\{0, \frac{\lambda}{2}, \frac{3\lambda}{2}, 
    \frac{5\lambda}{2}, 
    \ldots, \frac{n\lambda}{2}, 1
    %1 - \frac{\lambda}{2} 
    \Big\},
\]
where $n$ is the largest odd integer smaller than $2/\lambda$. 
}
%is the set of $1/\lambda$ evenly spaced real values over $[0,1]$.}
\end{definition}

Given a $\lambda$-discretization, we show that the number of level sets of an MCoA predictor $h$ can be bounded by $O(1/\lambda)$:

 \begin{claim}\label{claim:numlevelsets}
Let $\X$ be a finite domain, let $\F$ be a family of functions $f\colon \X \rightarrow [0,1]$, let $g\colon \X \rightarrow [0,1]$ be an arbitrary function, and let $\epsilon, \lambda \in (0,1)$. 
If $h$ is an $(\F, \epsilon)$-MCoA predictor for $g$, then there exists an $(\F, \epsilon+\lambda/2)$-MCoA predictor $h'$ for $g$ such that $h'$ has $O(1/\lambda)$ level sets.
%If $h$ is an $(\F, \epsilon, \gamma)$-approximate multicalibrated predictor for $g$ with $\gamma := \epsilon \cdot \lambda$, then there exists an $(\F, \epsilon+\lambda, \gamma)$-approximate multicalibrated predictor $h'$ for $g$ such that $h'$ has $O(1/\lambda)$ level sets.
 \end{claim}

 \begin{proof}
%For each $x$, 
We define $h'$ as follows: for each $x \in \X$, the value $h'(x)$ is equal to $\textrm{Round}(h(x))$, where the function $\textrm{Round}\colon [0,1] \rightarrow \Lambda[0,1]$ maps $h(x)$ to the closest value in $\Lambda[0,1]$ (rounding up in the case of ties). By definition of $\lambda$-discretization, this implies that $|h(x) - h'(x)| \leq \lambda/2$. 
If we denote $X'_w=\{x\in\X\mid h'(x)=w\}$ for each $w\in{\rm range}(h')$, then
$X'_w=\cup_v\, X_v$, where $w-\frac{\lambda}{2}\le v\le w+\frac{\lambda}{2}$ with $v\in {\rm range}(h)$. 
For each $w\in{\rm range}(h')$, we have
\begin{align*}
\Big|\E_{x\sim\D|_w}[f(x)\cdot (g(x)-h'(x)]\Big|
& \le
\Big|\E_{x\sim\D|_w}[f(x)\cdot (g(x)-h(x)]\Big|+
\E_{x\sim\D|_w}\Big[f(x)\cdot |h(x)-h'(x)|\Big] \\[0.2cm]
& \le \Big|\E_{x\sim\D|_w}[f(x)\cdot (g(x)-h(x)]\Big|+
\lambda/2,
\end{align*}
where $\D|_w$ denotes $\D|_{h'(x)=w}$ in this case.
Moreover, if we denote by $\Pa'(\D)$ the distribution that selects $X'_w$ for $w\in{\rm range}(h')$ and we denote by $\Pa(\D|_w)$ the one that selects $X_v$ for $v\in{\rm range}(h)$ with $w-\frac{\lambda}{2}\le v\le w+\frac{\lambda}{2}$, then 
\begin{align*}
\E_{X_w'\sim\Pa'(\D)}\Big|\E_{x\sim\D|_w}[f(x) & \cdot (g(x)-h(x)]\Big|
=
\E_{X_w'\sim\Pa'(\D)}\Big|\E_{X_v\sim\Pa(\D|_w)}\; \E_{x\sim\D|_v}[f(x) \cdot (g(x)-h(x)]\Big|
\\[0.2cm]
& 
\le
\E_{X_w'\sim\Pa'(\D)}\;\E_{X_v\sim\Pa(\D|_w)}\Big| \E_{x\sim\D|_v}[f(x) \cdot (g(x)-h(x)]\Big|
\\[0.2cm]
& =
\E_{X_v\in\Pa(\D)}\Big|\E_{x\sim\D|_v}[f(x)\cdot (g(x)-h(x)]\Big|
\le \epsilon.
\end{align*}
Therefore, $h'$ is $(\F,\epsilon + \lambda/2)$-multicalibrated on average.
\end{proof}

Now that we know how to construct an MCoA predictor with $O(1/\lambda)$ level sets, we show how to build an approximate MC predictor with $O(1/\epsilon)$ (as required in Theorem~\ref{thm:mcpartition}). 

\begin{claim}\label{claim:mcoatoapproxmc}
    Let $\X$ be a finite domain, $\F$ a family of functions $f:\X \rightarrow [0,1]$, $g: \X \rightarrow [0,1]$ an arbitrary function, and $\epsilon, \gamma \in (0, 1)$.
    There exists a $(\F, \epsilon, \gamma)$-approximate MC predictor $h$ for $g$ such that $h$ has $O(1/\epsilon)$ level sets.
\end{claim}

\begin{proof}
    Given $\epsilon, \gamma$, we define $\epsilon'$ and the discretization parameter $\lambda$ as follows:
    \[
        \epsilon' = \dfrac{\epsilon \gamma}{2}, \quad \lambda = \dfrac{\epsilon}{2}.
    \]
    We build an MCoA predictor $h$ with parameter $\epsilon'$ using the MCoA algorithms from the multicalibration literature \cite{gkrsw21, gkr23}.
    Next, we use the $\lambda$-discretization to transform this MCoA predictor $h$ into an approximate MC predictor $h'$ by merging the values of $h$ so that they only take values in $\Lambda[0,1]$ for $\lambda = \epsilon/2$.
    By the re-parametrization above, this gives an approximate MC predictor with $O(1/\epsilon)$ level sets. 
\end{proof}

We can now prove Theorem~\ref{thm:mcpartition}:

\begin{proof}[Proof of Theorem~\ref{thm:mcpartition}]
    We begin by applying Claim~\ref{claim:mcoatoapproxmc} to obtain an $(\F, \epsilon, \gamma)$-approximate MC predictor $h$ that has $O(1/\epsilon)$ level sets.
    By the MC theorem given by H\'ebert-Johnson et al., $h \in \F_{t, q}$, where $t = O(1/(\epsilon^4 \gamma) \cdot \log(|\X|/\epsilon))$ and $q = O(1/\epsilon^2)$ \cite{hkrr18}.
    We remark that each iterations in the \cite{hkrr18} algorithm requires a scalar multiplication, a finite-precision addition, and projection onto $[0,1]$, which can be handled by $O(\log(1/\epsilon) + \log(1/|\X|))$ gates. 
    The $\log(1/\epsilon)$ term corresponds to the fact that we perform computations up to a fixed precision of $\Theta(\epsilon)$, and the $\log(1/|\X|)$ term corresponds to bitlength of the elements in $\X$, given that each $x \in \X$ requires $\log(1/|\X|)$ bits to represent. 
    For full details on the circuit implementation, see \cite[\S 5]{dkrry21}.
    The partition $\Pa$ stated in Theorem~\ref{thm:mcpartition} is then given by the collection of level sets of this approximate MC predictor $h$.
\end{proof}

\section{From Hardcore Distributions to Hardcore Sets}\label{sec:setihcl}

\begin{lemma}[From distributions to sets {\cite[Lemma 24]{ks03}}]\label{lemma:setmeasure}
Let $\X = \{0,1\}^n$, $g\colon \X \rightarrow \{0,1\}$ an arbitrary function, $\epsilon, \delta > 0$, such that $\delta \geq 1/|\X|^{1/2}$, and 
$\F$ a class of functions $f\colon X \rightarrow \{0,1\}$ such that $\log(|\F|) \leq |\X| \epsilon^2 \delta^2$.
Suppose that $\mu$ is an $(\F, \epsilon)$-hardcore distribution for $g$ such that $d(\mu) \geq \delta$.
%$s \leq 2^n \epsilon^2 \delta^2/2n$. 
Then there exists a set $H$ with $|H| \geq \delta \cdot |\X|$ such that $H$ is an $(\F, 4\epsilon)$-hardcore set for $g$. 
\end{lemma}

\begin{proof}

Given the distribution $\mu$, we construct the claimed set $H$ as follows: for each $x \in \X = \{0,1\}^n$, include $x$ to $H$ with probability $\mu(x)$. This is a non-constructive set $H$, but we are still able to argue about its density. We denote the characteristic function of $H$ by $\chi_H$. 

Let $f \in \F$ and let $t(x)$ be an arbitrary function. We now use the following fact, where without loss of generality we view functions $f$ and $t$ as taking values in $\{-1,1\}$:

\begin{lemma}\label{lemma:fact}

Let $\rho>0$. Then, 
\[
    \Pr_{\D_{\mu}}[f(x) = t(x)] = \dfrac{1}{2} + \rho \iff \sum_{x \in \{0,1\}^n} \mu(x) f(x) t(x) = 2 \rho |\mu|,
\]
where $\D_{\mu}$ denotes the probability distribution induced by $\mu$ and $|\mu| = \sum_{x \in \X} \mu(x)$.
\end{lemma}

This is a direct consequence from the definition of $|\mu|$. Next, by the definition of $\chi_H$, it follows that
\[
    \E_{x \sim H}[\chi_H(x)] = \mu(x).
\]
By linearity of expectation, it follows that
\[
    \E_{x \sim H} \Big[ \sum_{x \in \X} \chi_H(x) f(x) g(x) \Big] = \sum_{x \in \X} \mu(x) f(x) g(x).
\]
Then, by applying Lemma~\ref{lemma:fact} with $t:=g$, it follows that
\[
    \sum_{x \in \X} \mu(x) f(x) g(x) \leq 2 \epsilon |\mu|,
\]
since by assumption $\mu$ is an $(\F_t, \epsilon)$-hardcore distribution for $g$. By definition, 
%this means that 
$\Pr_{x \sim D_{\mu}}[f(x)=g(x)] \leq \epsilon$, and hence the parameter $\rho$ in Lemma~\ref{lemma:fact} corresponds to $\epsilon$. 

Next, we use Hoeffding's inequality:

\begin{claim}[Hoeffding's tail bound]\label{claim:hoeffding}

Let $X_1, \ldots, X_N$ be independent random variables with $X_i \in [a, b]$ for all $i$. Then, for all $t \geq 0$,
\[
    \Pr \Big[ \frac{1}{n} \sum_{i=1}^n \left(X_i - \E[X_i]\right) \geq t \Big] \leq \exp \Big( \dfrac{-2nt^2}{(b-a)^2}  \Big). 
\]
\end{claim}

In our case, for each value of $x \in \X$ the quantity $ X_x := \chi_H(x) f(x) g(x)$ is a random variable in the interval $[-1, 1]$. Hence, each of these random variables corresponds to one $X_i$ in Hoeffding's bound.
We just showed that
\[
    \E_{x \sim \X}\Big[\sum_x \chi_H(x) f(x) g(x)\Big] \leq 2 \epsilon |\mu| = 2 \epsilon \cdot |\X| \cdot d(\mu),
\]
since by definition of density of a distribution we know that 
\[
    d(\mu) = \dfrac{|\mu|}{|\X|}.
\]
Set $t := 2 \epsilon d(\mu)$. Since $X_x \in [-1,1]$ for all $x \in \X$, it follows that $(b-a)^2 = 4$. Then, by Hoeffding's tail bound it follows that
\[
    \Pr \Big[ \dfrac{1}{|\X|} \sum_x X_x \geq 4 \epsilon d(\mu) \Big] \leq \exp \Big( \dfrac{-2 \cdot |\X| (2 \epsilon d(\mu))^2}{4}     \Big).
\]
Since by assumption $d(\mu) \geq \delta$, it follows that the above probability is less than $\exp(-2 \cdot |\X| \cdot \epsilon^2 \delta^2)$. Since by assumption $|\F| \leq 2^{|\X| \cdot \epsilon^2 \cdot \delta^2} \ll \frac{1}{10} \exp(2 \cdot |\X| \cdot \epsilon^2 \delta^2)$, and hence by the the union bound this implies that the probability that there exists some $f \in \F$ such that $\sum_x \chi_H(x) f(x) g(x) \geq 4 \epsilon |\mu|$ is less than $1/10$.

Next, applying the Hoeffding bound to $
|H|$ (which is a sum of $|\X|$ independent random variables), by similar calculations and using the assumptions that $d(\mu) \geq \delta$ and $\delta \geq 1/|\X|^{1/2}$, it follows that $|H| \geq 2 \delta \cdot |\X|$.

Lastly, putting everything together, we conclude that there exists some set $H$ such that
\begin{align*}
   & (1) \quad |H| \geq \delta \cdot |\X|, \\[0.2cm]
&    (2) \quad \sum_x \chi_H(x) f(x) g(x) \leq 4 \epsilon |\mu| \leq 8 \delta |H| = 8 \epsilon |\chi_H|. 
\end{align*}
Using Lemma~\ref{lemma:fact} with $t:=g$ and $\rho:= 4\epsilon$ we thus obtain that 
\[
    \Pr_{x \sim H}[f(x) = g(x)] \leq \dfrac{1}{2} + 4\epsilon
\]
for all $f \in \F$. Hence, $H$ is an ($\F, 4\epsilon)$-hardcore set for $g$, as we wanted to show.
\end{proof}

%\begin{remark}
{\rm

For further details, we defer to reader to Section~6 in \cite{imp95}, Section~2.1.1 in \cite{hol05}, or Section 4.4 in \cite{ks03}.

\section{Generalizing the IHCL Statement from \cite{rttv08}}\label{sec:rttvproof}

As summarized in Section~\ref{sec:rttv08}, we claim that we can remove the assumption that $g$ is $\delta$-hard in the IHCL statement shown by Reingold, Trevisan, Tulsiani, and Vadhan (\cite[Thm. 3.2]{rttv08}).
More specifically, Theorem 3.2 in \cite{rttv08} gives a partition $\Pa$ of $\X$ of complexity $\exp(\poly(1/\epsilon, 1/\tau))$ such that \emph{either}:
\begin{enumerate}
    \item $\bigcup_{P \in \Pa} P^{\mathrm{min}} \leq \tau \cdot |\X|$, or
    \item For every $f \in \F$,
    \[
        \E_{P \in D_{\min}} \big[| \Pr_{x \in P^{\mathrm{maj}}}[f(x)=1] - \Pr_{x \in P^{\mathrm{min}}}[f(x)=1]|\big],
    \]
    where $D_{\min}$ is the distribution that selects $P \in \Pa$ with probability proportional to $|P^{\min}|$.
\end{enumerate}

Indeed, if Condition (1) does not hold, then there is a distinguisher that can compute $g$ on more than a $1-\tau$ fraction of inputs.
On the other hand, in the energy potential argument given in \cite{rttv08}, Condition (2) can only hold if $\sum_P |P^{\mathrm{min}}| / |\X| \geq \tau$, which contradicts Condition (1).
Moreover, we remark that Condition (1) is analogous to our condition that we cannot make guarantees about the pieces $P \in \Pa$ such that the balance parameter $b_P$ or the size parameter $\eta_P$ are too small (indeed, the indistinguishability parameter in our IHCL$++$ degrades with $b_P$, and by the approximate MC theorem we only consider the pieces that satisfy $\eta_P \geq \gamma$), given that it implies that $b_P \cdot \eta_P \leq \tau$ for all $P \in \Pa$ when $\D$ is the uniform distribution.
Indeed, we claim that Condition~(1) implies that
\[
    b_P \cdot \eta_P \leq \dfrac{\big|\bigcup_{P \in \Pa} P^{\mathrm{min}}\big|}{|\X|},
\]
which, by our definitions of $b_P$ and $\eta_P$, is equivalent to
\[
    \min\{\E_P[g(x)], 1-\E_P[g(x)]\} \cdot \dfrac{|P|}{|\X|} \leq \dfrac{\big|\bigcup_{P \in \Pa} P^{\mathrm{min}}\big|}{|\X|}.
\]
If $b_P \leq 1/2$, then the minority element in $P$ is 1 and $\E_P[g(x)] \leq |\{x \in P \mid g(x)=1\}| / |P|$.
Hence, the inequality simplifies to $ |P^{\mathrm{min}}| \leq \big|\cup_{P \in \Pa} P^{\mathrm{min}}\big|$, which is clearly true. 
If $b_P > 1/2$, then the minority element in $P$ is 1 and and $\E_P[g(x)] \leq |\{x \in P \mid g(x)=0\}| / |P|$, so the expression again simplifies to $ |P^{\mathrm{min}}| \leq \big|\cup_{P \in \Pa} P^{\mathrm{min}}\big|$.

Lastly, we show that Reingold's et al. Theorem 3.2 modified as above to hold for an arbitrary function $g$ implies the original IHCL.
Hence, we bring back the assumption that $g$ is $(\F, \delta)$-weakly hard for $s = \poly(1/\epsilon, 1/\delta)$ and $\delta > \tau$. 
This implies that $\cup_{P \in \Pa} P^{\mathrm{min}} \geq \delta \cdot |\X|$, because otherwise this would contradict the assumption that $g$ is $\delta$-weakly hard.
Thus, Condition (1) does not hold, and by our modified version of Reingold et al.'s Theorem 3.2, this implies that Condition (2) must hold, which yields a hardcore distribution $D_{\min}$.

\end{document}